\documentclass[10pt,a4paper]{article}
\usepackage[utf8]{inputenc}
\usepackage{amsmath, amsthm}
\usepackage{amsfonts}
\usepackage{amssymb}
\usepackage{bbm}
\usepackage{hyperref}	

\usepackage[left=2cm,right=2cm,top=2cm,bottom=2cm]{geometry}

\usepackage{enumerate}

\theoremstyle{plain}
\newtheorem{thm}{Theorem}[section]

\newtheorem{lem}[thm]{Lemma}
\newtheorem*{lems}{Lemma}
\newtheorem{prop}[thm]{Proposition}

\newtheorem{conjecture}[thm]{Conjecture}
\newtheorem*{prop2}{Proposition}

\theoremstyle{remark}
\newtheorem{rmk}{Remark}

\theoremstyle{definition}

\newtheorem{defn}{Definition}
\newtheorem{hyp}{Assumption}

\addtocounter{tocdepth}{-1}
\usepackage{color}

\usepackage{bm}
\usepackage{amsmath}
\usepackage{graphicx}
\usepackage{authblk}
\numberwithin{equation}{section}
\begin{document}
	\title{Stability and instability of the sub-extremal Reissner-Nordstr\"{o}m  black hole interior for the Einstein-Maxwell-Klein-Gordon equations in spherical symmetry}

	\author[1]{Maxime Van de Moortel  \thanks{E-mail : mcrv2@cam.ac.uk}}
	
	\affil[1]{University of Cambridge, Department of Pure Mathematics and Mathematical Statistics, Wilberforce Road, Cambridge CB3 0WA, United Kingdom}
	
	\date{\vspace{-8ex}}
	\maketitle
	\abstract
	
	We show non-linear stability and instability results in spherical symmetry for the interior of a charged black hole -approaching  a sub-extremal Reissner-Nordstr\"{o}m background fast enough- in presence of a massive and charged  scalar field, motivated by the strong cosmic censorship conjecture in that setting :

	\begin{enumerate}
		\item \textit{Stability} :	We prove that spherically symmetric characteristic initial data to the Einstein-Maxwell-Klein-Gordon equations approaching a Reissner-Nordstr\"{o}m background with a sufficiently decaying polynomial decay rate on the event horizon gives rise to a space-time possessing a Cauchy horizon in a neighbourhood of time-like infinity. Moreover if the decay is even stronger, we prove that the space-time metric admits a continuous extension to the Cauchy horizon. This generalizes the celebrated stability result of Dafermos for Einstein-Maxwell-real-scalar-field in spherical symmetry. 
		
		\item  \textit{Instability} : We prove that for the class of space-times considered in the stability part, whose scalar field   in addition obeys a polynomial averaged-$L^{2}$ (consistent) lower bound on the event horizon, the scalar field obeys an integrated lower bound transversally to the Cauchy horizon. As a consequence we prove that the non-degenerate energy is infinite on any null surface crossing the Cauchy horizon and the curvature of a geodesic vector field blows up at the Cauchy horizon near time-like infinity. This generalizes an instability result due to Luk and Oh for Einstein-Maxwell-real-scalar-field in spherical symmetry.  
	\end{enumerate}
	
	This instability of the black hole interior can also be viewed as a step towards the resolution of the $C^2$ strong cosmic censorship conjecture for one-ended asymptotically flat initial data.

	\tableofcontents

	\section{Introduction}
	
	In this paper, we study the stability and instability of the Reissner-Nordstr\"{o}m Cauchy horizon for the Einstein-Maxwell-Klein-Gordon equations in spherical symmetry  :	\begin{equation} \label{1} Ric_{\mu \nu}(g)- \frac{1}{2}R(g)g_{\mu \nu}= \mathbb{T}^{EM}_{\mu \nu}+  \mathbb{T}^{KG}_{\mu \nu} ,    \end{equation} 
	\begin{equation} \label{2} \mathbb{T}^{EM}_{\mu \nu}=2\left(g^{\alpha \beta}F _{\alpha \nu}F_{\beta \mu }-\frac{1}{4}F^{\alpha \beta}F_{\alpha \beta}g_{\mu \nu}\right),
	\end{equation}
	\begin{equation}\label{3}  \mathbb{T}^{KG}_{\mu \nu}= 2\left( \Re(D _{\mu}\phi \overline{D _{\nu}\phi}) -\frac{1}{2}(g^{\alpha \beta} D _{\alpha}\phi \overline{D _{\beta}\phi} + m ^{2}|\phi|^2  )g_{\mu \nu} \right), \end{equation} \begin{equation} \label{4}\nabla^{\mu} F_{\mu \nu}= \frac{ q_{0}}{2} i (\phi \overline{D_{\nu}\phi} -\overline{\phi} D_{\nu}\phi) , \; F=dA ,
	\end{equation} \begin{equation} \label{5} g^{\mu \nu} D_{\mu} D_{\nu}\phi = m ^{2} \phi ,
	\end{equation}
	
	where the constants $m^2$ and $q_0$ are respectively called the mass and the charge \footnote{This charge $q_0$ is also the constant that couples the electromagnetic and the scalar field tensors.} of the scalar field $\phi$.	
	
	This problem is motivated by Penrose's strong cosmic censorship conjecture (c.f section \ref{SCC}.) which claims that general relativity is a deterministic theory. The general strategy to address this question is to exhibit a singularity at the boundary of the maximal domain of predictability, which can be done with instability estimates. 
	
	We prove that assuming an upper and lower bound on the scalar field $\phi$ on the event horizon of the black hole, the Cauchy
	horizon exhibits both stability and instability features, namely :

	\begin{enumerate} \item Stability : the perturbed black hole still admits a Cauchy horizon -near time-like infinity- like the original unperturbed Reissner-Nordstr\"{o}m black hole, and in some cases we can even extend the metric continuously beyond this Cauchy horizon. 
		
		\item Instability : the curvature along the Cauchy horizon blows up, which represents an obstruction \footnote{Although an appropriate global setting - as opposed to the perturbative one that this paper is concerned with-  is necessary to formulate the $C^2$ inextendibility properly. } to a $C^2$ extension, at least near time-like infinity. As a by-product, we see that the metric is not $C^1$ for the constructed continuous extension\footnote{Although it does not give a general geometric impossibility to extend in $C^1$ the metric across the Cauchy horizon.}.
		
	\end{enumerate}
	
	Similar results are known in the special case $m^2=q_0=0$ see \cite{Mihalis1} and \cite{JonathanStab}. 
	However, in our case the expected decay of the scalar field on the event horizon is much slower, which makes the stability part more difficult. The previous instability result depends strongly on the special structure of the equation in the absence of mass and charge of the scalar field \footnote{More precisely, in the work of Dafermos \cite{Mihalis1}, it relies on a special mononoticity property occuring only in that model. }. When $q_0 \neq 0$ but $m^2=0$, a previous work of Kommemi \cite{Kommemi} shows a stability result but his assumed decay on the event horizon is only expected to hold for a sub-range of the charge $q_0$ that depends on the black hole parameters.  In \cite{JonathanStab}, the key argument for the instability is to use an almost conservation law that exists only in the absence of mass and charge. This is the underlying reason why \cite{Kommemi} does not contain any instability result.
	
	This work can also be viewed as a first step towards the understanding of the spherically symmetric charged black holes with one-ended initial data. This is because when the scalar field is uncharged, the total charge of the space-time arises completely from the topology. On the contrary, the model that we consider allows for a dynamical total charge which makes $\mathbb{R}^3$ type initial data possible.

	The introduction is outlined as follows : in section \ref{context} we present the strong cosmic censorship conjecture and mention earlier works, then in section \ref{motivation} we explain the reasons to study a \textbf{charged} and \textbf{massive} scalar field and give the results of the present paper.
	We then sketch the methods of proof in the last section \ref{ideaproof}. Finally in section \ref{outline} we outline the rest of the paper.
	
	\subsection{Context of the problem and earlier works} \label{context}

	\subsubsection{Strong cosmic censorship conjecture} \label{SCC}

	The study of self-gravitating isolated bodies relies crucially on the vacuum Einstein equation  : 
	\begin{equation}  \label{Einstein} Ric_{\mu \nu}(g)- \frac{1}{2}R(g)g_{\mu \nu}= 0 .\end{equation}

	
	The simplest non-trivial solution, discovered by Schwarzschild is a spherically symmetric family of black holes, indexed by their mass.
	These black holes exhibit a very strong singularity, as observers that fall into them experience \textbf{infinite tidal deformations}.

	
	A more sophisticated family of solutions indexed by mass and angular momentum and which describes rotating black holes has been discovered by Kerr in 1963. Unfortunately, Kerr's black holes have the very undesirable feature that they break determinism : the maximal globally hyperbolic development of their initial data is future extendible as a smooth solution to the Einstein equation \eqref{Einstein} in many \textbf{non-unique} ways. In some sense, it represents a failure of global uniqueness of solutions.
	
	One way to restore determinism which has been suggested by numerous heuristic and numerical works is that Kerr black holes feature of non-unique extendibility is \textbf{non-generic}, in other words whenever their initial data is slightly perturbed then the maximal globally hyperbolic development is actually future \textbf{inextendible} as a suitably regular Lorentzian manifold. 
	
	The nature of this singularity was controversial though : it was widely debated in the physics community whether perturbations of Kerr black holes exhibit a Schwarzschild black hole like singularity and  observers experience infinite tidal deformations when they get close to it. One convenient way -although not exactly equivalent- to formulate this question geometrically is to study  $C^0$ inextendibility. 
	
	The inextendibility question has been formulated by Penrose in the following conjecture :

	\begin{conjecture}[ Strong Cosmic Censorship, Penrose]
		Maximal  globally hyperbolic developments of asymptotically flat initial data are generically future inextendible as a suitably regular Lorentzian manifold . \end{conjecture}
	
	In the case of $C^0$ inextendibility, suitably regular is to be understood as continuous.
	
	\begin{rmk} \label{RNSCC}		Without the word ``generically'', the conjecture is false since Kerr black holes would provide counter examples, in the sense that they have a Cauchy horizon over which the metric can be smoothly extended in a non-unique way.  Strong cosmic censorship claims that these counter examples are non-generic.	\end{rmk}

	Due to the complexity of the Kerr geometry, early works on this problem studied instead Reissner-Nordstr\"{o}m charged black holes. Although there are not solutions to the vacuum Einstein equation \eqref{Einstein}, they solve the Einstein-Maxwell equations : 	\begin{equation} \label{EinsteinMaxwell}  Ric_{\mu \nu}(g)- \frac{1}{2}R(g)g_{\mu \nu}= \mathbb{T}^{EM}_{\mu \nu} , \end{equation} 	
	\begin{equation}  \label{EinsteinMaxwell2}  \mathbb{T}^{EM}_{\mu \nu}=2\left(g^{\alpha \beta}F _{\alpha \nu}F_{\beta \mu }-\frac{1}{4}F^{\alpha \beta}F_{\alpha \beta}g_{\mu \nu}\right),
	\end{equation}
	\begin{equation}   \label{EinsteinMaxwell3}\nabla^{\mu} F_{\mu \nu}= 0  , \; dF=0. \end{equation}
	Reissner-Nordstr\"{o}m  black holes have the same global geometry as Kerr's but have the simplifying feature that they are spherically symmetric.

	In their pioneering numerical work \cite{Penrose}, Penrose and Simpson  studied linear test fields on Reissner-Nordstr\"{o}m black holes and discovered an instability of the Cauchy horizon.

	Later Hiscock in \cite{Hiscock},  Poisson and Israel in \cite{Poisson}, \cite{Poisson2} exhibited - in a spherically symmetric but \textbf{non-linear} setting- a so-called weak null singularity with an expected curvature blow-up i.e a $C^{2}$ explosion of the metric, but \textbf{finite} tidal deformations allowing for a $C^{0}$ extension.
	
	They studied the Einstein-null-dust equations which model non self-interacting matter transported on null geodesics :  \footnote{This model can be thought of as a high frequency limit, away from $\{r=0\}$ of the Einstein-Scalar-Field model.}

	\begin{equation}  Ric_{\mu \nu}(g)- \frac{1}{2}R(g)g_{\mu \nu}= \mathbb{T}_{\mu \nu},     \end{equation} 
	\begin{equation}
		\mathbb{T}_{\mu \nu}= f^2 \partial_{\mu}u \partial_{\nu}u+
		h^2 \partial_{\mu}v \partial_{\nu}v,
	\end{equation}
	\begin{equation}
		g^{\mu \nu}\partial_{\mu}u \partial_{\nu}u=0
	\end{equation}
	\begin{equation}
		g^{\mu \nu}\partial_{\mu}v \partial_{\nu}v=0
	\end{equation}
	\begin{equation}
		g^{\mu \nu}\partial_{\mu}u \partial_{\nu}f +(\Box_{g}u)f=0
	\end{equation}
	\begin{equation}
		g^{\mu \nu}\partial_{\mu}v \partial_{\nu}h +(\Box_{g}v)h=0
	\end{equation}
	

	
	
	
	In his seminal work \cite{MihalisPHD}, \cite{Mihalis1}, Dafermos studied mathematically  the \textbf{non-linear} stability of Reissner-Nordstr\"{o}m black holes in spherical symmetry for the Einstein-Maxwell-Scalar-Field equations  : 
	
	\begin{equation}  \label{EMSF} Ric_{\mu \nu}(g)- \frac{1}{2}R(g)g_{\mu \nu}= \mathbb{T}^{EM}_{\mu \nu}+  \mathbb{T}^{SF}_{\mu \nu},     \end{equation} 
	\begin{equation}  \label{EMSF2} \mathbb{T}^{EM}_{\mu \nu}=2\left(g^{\alpha \beta}F _{\alpha \nu}F_{\beta \mu }-\frac{1}{4}F^{\alpha \beta}F_{\alpha \beta}g_{\mu \nu}\right),
	\end{equation}
	\begin{equation} \label{EMSF3}  \mathbb{T}^{SF}_{\mu \nu}= 2 \left( \partial_{\mu}\phi \partial_{\nu} \phi -\frac{1}{2}(g^{\alpha \beta} \partial_{\alpha}\phi \partial_{\beta}\phi)g_{\mu \nu} \right), \end{equation} \begin{equation} \label{EMSF4} \nabla^{\mu} F_{\mu \nu}= 0 , \; dF=0 ,
	\end{equation} \begin{equation}   \label{EMSF5} g^{\mu \nu} \nabla _{\mu} \nabla _{\nu}\phi = 0.
	\end{equation}


	Dafermos studied the interior of the black hole and proved conditionally the existence of a Cauchy horizon near time-like infinity with a $C^0$ extension for the metric, but $C^1$ inextendibility of the $C^0$ extension which manifests itself by the blow-up of the (Hawking) mass, which partially confirmed the insights from the work of Poisson-Israel. \\
	
	Later Dafermos and Rodnianski in \cite{PriceLaw} proved a stability result on the black hole exterior (c.f Section~\ref{Price}) which combined with \cite{Mihalis1} ruled out the $C^0$ inextendibility scenario : 
	
	\begin{thm} [Dafermos \cite{Mihalis1}, Dafermos-Rodnianski \cite{PriceLaw} ]
		For the Einstein-Maxwell-Scalar-Field equations \eqref{EMSF},\eqref{EMSF2},\eqref{EMSF3}, \eqref{EMSF4}, \eqref{EMSF5} in spherical symmetry, the ${C^0}$ strong cosmic censorship is false. 
	\end{thm} 
	
	The question was finally settled in the work of Luk and Oh \cite{JonathanStab}, \cite{JonathanStabExt} : they confirmed the weak null singularity scenario, due to a curvature instability : 
	
	\begin{thm} [Luk-Oh \cite{JonathanStab}, \cite{JonathanStabExt} ]
		For the Einstein-Maxwell-Scalar-Field equations \eqref{EMSF},\eqref{EMSF2},\eqref{EMSF3},\eqref{EMSF4},\eqref{EMSF5} in spherical symmetry, the $C^2$ strong cosmic censorship conjecture is true. 
	\end{thm}

	\subsubsection{Earlier works relating to singularities at the Cauchy horizon}
	
	As sketched in the previous section, singularities are tightly related to extendibility question. For the stability of the Cauchy horizon, recent progress have been made in different directions c.f \cite{Franzen}, \cite{Hintz} for the linear stability, \cite{JonathanInstab}, \cite{KerrInstab} for the linear instability and \cite{Kommemi} for the non-linear problem.

	In this section, we review in more details stability and instability results in the black hole interior established in previous works leading to the proof of the $C^2$ strong cosmic censorship conjecture. These results should be compared to the  main theorems of this paper, stated in section \ref{results}.
	
	In \cite{Mihalis1}, Dafermos proves a ${C}^0$ stability and a ${C}^1$ instability result  of the Reissner-Nordstrom solution for an uncharged massless scalar field perturbation suitably decaying along the event horizon.
	
	The instability essentially relies on a blow-up of the modified mass $\varpi$  over the Cauchy horizon,as a consequence of a lower bound on the scalar field. Hence the metric is not ${C}^{1}$  extendible\footnote{It can also be proven that the mass blow-up implies also the blow-up of the Kretschmann scalar (c.f \cite{Kommemi}) which establishes $C^2$ inextendibility.} in spherical symmetry.

	\begin{thm}[$\mathcal{C}^0$ stability, ${C}^1$ instability, Dafermos \cite{Mihalis1}] \label{CondInst}
		
		Let $(M,g,\phi,F)$ be a solution of the Einstein-Maxwell-Scalar-Field equations in spherical symmetry such that  for some $s > \frac{1}{2}$, we have on the event horizon parametrized by the coordinate $v$ as defined by gauge \eqref{gauge2} of Theorem \ref{Stabilitytheorem}, and under the gauge choice  \eqref{gaugeAv}: 
		
		$$ |\phi|_{|\mathcal{H}^+}(0,v)+|\partial_v \phi|_{|\mathcal{H}^+}(0,v) \lesssim v^{-s},$$
		
		then : 
		
		\begin{enumerate}

			\item \textbf{Existence of a Cauchy horizon} : in a neighbourhood of time-like infinity, the space-time has the Penrose diagram of Figure 1.


			\item \textbf{Continuous extension}	: if moreover $s>1$ then the metric g and the scalar field $\phi$ extend as continuous functions along the Cauchy horizon $\mathcal{CH}^{+}$. Moreover, the extended metric can be chosen to be spherically symmetric.

			\item \textbf{Mass inflation and $C^1$ inextendibility} :	  coming back to general case $s>\frac{1}{2}$, 		  if we assume the following point-wise lower bound \footnote{			  	This lower bound -although supported by numerical evidences- has never been exhibited for any particular solution.} on the scalar field for some $\epsilon>0$ :  $$  v^{-3s+\epsilon}  \lesssim |\partial_{v} \phi |_{|\mathcal{H}^{+}} \lesssim   v^{-s},$$		then  ,  the modified mass blows up as one approaches the Cauchy horizon : $\varpi(u,v) \rightarrow_{v \rightarrow +\infty}+\infty$ hence it is impossible to extend the metric g to a \textbf{spherically symmetric} $C^{1}$ metric across the Cauchy horizon $\mathcal{CH}^{+}$. In particular the constructed $C^{0}$ extension is not $C^{1}$.							 \end{enumerate}			\end{thm}
	
	\begin{figure}
		
		\begin{center}
			
			\includegraphics[width=65 mm, height=65 mm]{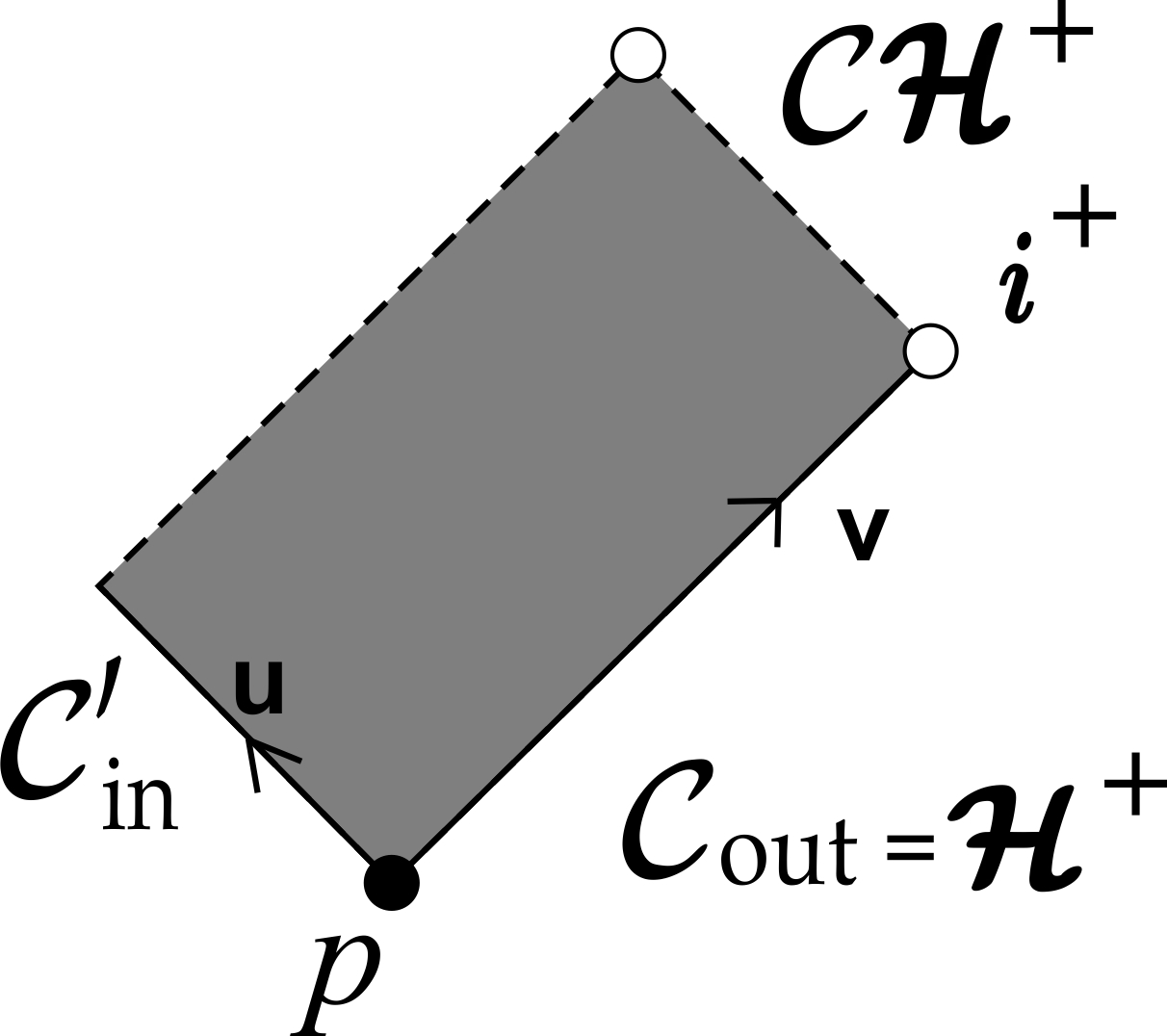}
			
		\end{center}
		
		\caption{Penrose diagram for the characteristic initial value problem appearing in \cite{Mihalis1} .}
		\label{MihalisPenrose}
	\end{figure}
	
	In contrast, the $C^2$ strong cosmic censorship conjecture paper dealing with the black hole interior \cite{JonathanStab} relies on an averaged polynomial decay, as opposed to point-wise and proves a curvature instability :

	\begin{thm}[ ${C}^2$ instability Luk-Oh \cite{JonathanStab}]
		
		Under the same hypothesis as Theorem \ref{CondInst}, we also assume that $s>2$ and the following lower bound holds for some $2s-1 \leq p<4s-2$ and some $C>0$ : 
		
		\begin{equation}
			C v^{-p} \leq \int_{v}^{+\infty} |\partial_v \phi|^2_{|\mathcal{H}^+}(0,v')dv'
		\end{equation}

		The solution admits a continuous extension $\bar{M}$ across the Cauchy horizon.
		
		Then a component of the \textbf{curvature blows-up} identically along that Cauchy horizon.

		As a consequence, $(M,g,\phi,F)$	is  \bm{${C}^2$} \textbf{future-inextendible}.
		
		Moreover $\phi \notin W^{1,2}_{loc}(\bar{M})$ and the metric is not in $C^1$ for the constructed continuous extension $\bar{M}$.

	\end{thm}
	
	\subsection{A first version of the main results} \label{motivation}
	
	In this paper we prove that the expected asymptotic decay of the scalar field on the event horizon -known as generalised Price's law- \footnote{Namely an polynomial decay for an initially compactly supported scalar field on the event horizon of the black hole.}   implies some stability and instability features for a more realistic and richer generalization of the charged space-time model of Dafermos in spherical symmetry. 
	
	Instead of studying this problem starting from Cauchy data, we will only consider characteristic initial data on the event horizon with the ``expected'' behaviour. This should be thought of as an analogue of the previous black hole interior studies \cite{Mihalis1} and \cite{JonathanStab}.

\subsubsection{Motivation to study a massive and charged field and the results of the present paper }

The goal of this paper is to generalize the known results for the Einstein-Maxwell-Scalar-Field equations near a Reissner-Nordstr\"{o}m background to the case of a massive and charged scalar field model called Einstein-Maxwell-Klein-Gordon. Since the charge and the mass are a priori two different issues, we give motivation for each of them. \\
\begin{enumerate}

\item \textit{A charged scalar field} : The model of Dafermos is a good toy model which gave very good insight on the Kerr case but it suffers from a major disadvantage : the topology of the initial data -i.e the initial time slice which is a Riemannian manifold- is constrained to be that of $\mathbb{S}^2 \times \mathbb{R}$ i.e \textbf{two-ended} initial data like for the Reissner-Nordstr\"{o}m case. This does not seem so relevant to study isolated collapsing matter  : we would like to consider \textbf{one-ended} initial data, diffeomorphic to $\mathbb{R}^3$ , but it is not possible in that model where the radius cannot go to $0$ on a fixed time slice. 

This fact is due to the topological character of the total charge of the space-time. This is better understood by the formula :	$$ F= \frac{Q}{2 r^{2}}  \Omega^{2} du \wedge dv, $$

where $(u,v)$ are null coordinates built from the radius $r$ and the time $t$, $Q$ is the total charge of the space-time, $\Omega^2$ is the metric coefficient in $(u,v)$ coordinates (c.f section \ref{Coordinates}) and $F$ is the electromagnetic field 2-form.

Heuristically we see that, if $Q \equiv e$ is fixed with $e \neq 0$, $r$ is not allowed to tend to 0 without a blow-up of $F$ (if the metric does not degenerate). For more details on these issues, c.f \cite{Kommemi}.

It turns out that if we impose that the scalar field is uncharged then the charge of the space-time $Q$ is necessary fixed to be some $e \in \mathbb{R}$, as it will be seen in equations \eqref{chargeUEinstein} and \eqref{ChargeVEinstein} of section \ref{eqcoord}.

As a conclusion, considering more natural $\mathbb{R}^3$ initial data imposes to study a generalisation of Dafermos' model namely the Einstein-Maxwell-\textbf{Charged}-Scalar-Field equations.\\

\item \textit{A massive scalar field} :	Another variant is to allow for the scalar field to carry a mass, independently of the presence or absence of charge : it now propagates according to the Klein-Gordon equation : 

\begin{equation}  \label{KG} g^{\mu \nu} \nabla _{\mu} \nabla _{\nu}\phi = m ^{2} \phi .
\end{equation}

One reason to study the Klein-Gordon equation is to understand the effect of a different kind of matter on the results of mathematical general relativity and the strong cosmic censorship in particular.

Klein-Gordon equation is also fruitful to study \textbf{boson stars}. These uncharged objects -already present in the simple framework of spherical symmetry- in addition to being interesting for theoretical physics, give an example of a non-black-hole new ``final state'' of gravitational collapse. 

More importantly, they are soliton-like (even though the metric is static), in particular they are \textbf{non-perturbative solutions} which do not converge towards a Schwarzchild or Kerr background ! They even exhibit a new behaviour as the scalar field is time-periodic in contrast to vacuum where periodicity is impossible (all periodic vacuum space-time are actually stationary, c.f \cite{Volker}). If we let aside the fact that the scalar field is not stationary, boson stars are counter-examples to the generalized \textbf{no-hair conjectures} which broadly suggest that the set of stationary and asymptotically flat solutions to the Einstein equations coupled with any reasonable matter should reduce to a finite dimensional family indexed by physical parameters measured at infinity, like Kerr's black hole (indexed by mass and angular momentum) or Reissner-Nordstr\"{o}m's (indexed by mass and electric charge).	
For more developments on boson stars, c.f \cite{Boson}.

Outside of spherical symmetry \footnote{Getting rid of the spherical symmetry assumption allows for a new very important physical phenomenon to arise, namely superradiance. This instability feature results in the presence of exponentially growing modes as  discussed in \cite{Otis} and \cite{Yakov}. }, a recent work of Chodosh-Shlapentokh-Rothman \cite{Otis} constructs a continuous 1-parameter family of periodic space-times between a Kerr black hole and a boson star. Interestingly they exhibit solutions with exponentially growing modes, which is impossible in vacuum as proved (in the linear case) in \cite{KerrDaf} ! In contrast, LeFloch and Ma prove in \cite{Lefloch} that the Minkowski space-time is stable for the Einstein-Klein-Gordon equations.

As a conclusion, the Klein-Gordon model enriches the dynamics of gravitational collapse and generates behaviours that are not present for a simple wave propagation. Despite these rich dynamics, the perturbative regime sometimes behave like the massless case as in \cite{Lefloch} or the present paper, and sometimes behaves rather differently as in the work \cite{Otis}.  \\

In this paper, we are going to consider both problems simultaneously by studying a charged and massive field propagating according to the Klein-Gordon equation \eqref{KG}. The full problem is written in section \ref{EinsteinMaxKG}.

\item \textit{Mathematical differences with Dafermos' model} : After dealing with physical aspects, we want to emphasize the technical differences between our new model and the uncharged massless one. 

A first remark is that the monotonicity of the modified mass as defined in \eqref{electromass} and that of the scalar field which is strongly relied on in the instability argument of \cite{Mihalis1} are no longer available.

More importantly, the expected asymptotics (Price's law) of the field on the event horizon are different : in particular, the oscillations due to the charge should give only an averaged \footnote{Which does not make a difference to prove the $C^0$ stability because we only need an upper bound but does for the $C^1$ instability where point-wise estimates are no longer enough. } polynomial decay -as opposed to point-wise decay-  and in many cases, the decay is expected to be always much weaker than for the uncharged and massless case. In particular it should be often  \textbf{non-integrable}.

Moreover, the charge is no longer a topological constraint but a dynamical quantity which obeys an evolutionary P.D.E and that should be controlled like the scalar field or the metric which is what renders \textbf{one-ended} asymptotically flat initial data possible.

\end{enumerate}

\subsubsection{Price's law conjecture}	\label{Price}

We now state the expected asymptotics for the scalar field on the event horizon. This was first heuristically discovered by Price in \cite{Pricepaper} for the Schwarzschild solution, and proven rigorously by Dafermos and Rodnianski in \cite{PriceLaw} on Schwarzschild and Reissner-Nordstr\"{o}m perturbations for an uncharged and massless field. The statement that the tail of the scalar field decays polynomially - for all models - is now called generalized Price's law.

This conjecture is still an open problem for the charged and massive model of the present paper and requires a stability study of the black hole exterior. The statement is however supported by numerical studies of the black hole exterior, c.f \cite{Phycists2} and \cite{Physicists}.

\begin{conjecture} [Price's law decay] \label{Priceconj} Let $ (M,g,\phi,F)$ be a spherically symmetric solution of the Einstein-Maxwell-Klein-Gordon system which is a perturbation of a Reissner-Nordstr\"{o}m background of mass $M$ and charge $e$ satisfying $0<|e|<M$, with a massive charged field $\phi \in C_c^{\infty}(\Sigma)$ of charge $q_0$ -as appearing in equations \eqref{chargeUEinstein}, \eqref{ChargeVEinstein}- and of mass $m^2$ -as appearing in the Klein-Gordon equation \eqref{KG}- where $\Sigma$ is an asymptotically flat complete Riemannian manifold initial data slice.

Then on the event horizon of the black hole $\mathcal{H}^+$ parametrized by the coordinate $v$ as defined by gauge \eqref{gauge2} of Theorem \ref{Stabilitytheorem}, we have:

\begin{enumerate}
	\item In the massive uncharged case, i.e., $m^2 > 0$, $q_0=0$, \begin{equation} 
		\phi_{|\mathcal{H}^+}(v)(v) = C(m\cdot M, D)  \sin( m v + \omega_{err}(v)) \cdot v^{-\frac{5}{6}} 
	\end{equation} for a constant $C(m\cdot  M,D) \neq 0$ depending on $m\cdot M$ and the initial data $D$, and a sublinear growing phase $\omega_{err}(v)= -\frac{3m}{2} (2\pi M)^{\frac{2}{3}}  v^{\frac{1}{3}} + \omega(m \cdot M) $.
	\item In the massless charged case, i.e., $m^2 = 0$, $q_0 \neq 0$,
	\begin{equation}  \label{conj.decay.2.eq}
		\phi _{|\mathcal{H}^+}(v)(v) = C_H(q_0e,D) \cdot e^{ \frac{iq_0 e}{r_+} v} \cdot v^{-1-\sqrt{1-4(q_0 e)^2}},
	\end{equation} where $C_H(q_0e,D)\neq 0$ is a constant depending on $q_0 e$ and the initial data $D$.
	\item  \label{conj.decay.3} In the massive charged case, i.e.\ $m^2 > 0$, $q_0 \neq 0$,
	\begin{align} \label{conj.decay.3.eq}
		\phi _{|\mathcal{H}^+}(v) =  C(M \cdot m,D) \cdot  e^{ \frac{iq_0 e}{r_+} v} \cdot \sin( m v + \omega_{err}(v)) \cdot v^{-\frac{5}{6}}.
	\end{align} 
	
\end{enumerate}

\end{conjecture}

\begin{rmk}
Notice that $s(e,q_0,m^2)> \frac{1}{2}$ always but that the integral decay $s>1$ holds \footnote{Note that the decay of the massless charged scalar field depends on the dimension-less quantity $q_0e$ only.} only for $m^2=0$, $|e| < \frac{1}{2|q_0|}$. Since integrability is the crucial point in the $C^0$ extendibility proof, it explains why we required the field to be massless and not too charged to claim the $C^0$ extendibility.
\end{rmk}

Dafermos and Rodnianski in \cite{PriceLaw} first proved rigorously and in the non-linear setting an upper bound for Price's law in the uncharged and massless case $m^2=q_0=0$.





Later, Luk and Oh proved in \cite{JonathanStabExt} the sharpness of this upper bound, still in the non-linear setting, as a consequence of a $L^2$ averaged \footnote{Note that for the case $q_0=0$ it is expected that the function $f$ is constant i.e the oscillations should not arise when the scalar field is uncharged. Nonetheless, no point-wise lower bound has ever been established, even for a particular solution.} lower bound.

\subsubsection{Statement of the main results} \label{theoremsketch}

In this section we explain roughly the achievement of the present paper. 
The stability result is very analogous to Dafermos' in \cite{Mihalis1} and the instability result is a local near time-like infinity version of Luk and Oh's interior instability of \cite{JonathanStab}.

More precisely, we establish the following result : 

\begin{thm} We assume Price's law decay of conjecture \ref{Priceconj} on the event horizon for a solution of the Einstein-Maxwell-Klein-Gordon system of section \ref{EinsteinMaxKG}  in spherical symmetry.

Then near time-like infinity, the solution remains regular \footnote{More precisely, the Penrose diagram -locally near timelike infinity- of the resulting black hole solution is the same as Reissner-Nordstr\"{o}m's as illustrated by Figure 1. } up to its Cauchy horizon \footnote{On the other hand in general the metric may not extend even continuously to that Cauchy horizon.}  , along which a $C^2$ invariant quantity \footnote{Namely $Ric(V,V)$ where $V$ is a radial null geodesic vector field that is transverse to the Cauchy horizon.} blows up.

Furthermore, defining  $e \in \mathbb{R}$ to be the asymptotic charge of the space-time measured on the event horizon \footnote{It corresponds to the parameter $e$ of the sub-extremal Reissner-Nordstr\"{o}m background $(M,e)$ towards which our space-time converges on the event horizon.}- for $m^2=0$ and for $4q_0^2 e^2<1$ -  the metric is  $C^0$ extendible.

\end{thm}

The proof relies on a non-linear stability and instability study of the Reissner-Nordstr\"{o}m black hole interior. The $C^0$ extendibility was first proven by Dafermos in \cite{Mihalis1} in the uncharged and massless setting but it is really a direct adaptation of the methods of \cite{JonathanStab} that gives $C^0$ extendibility in the massless and charged (for $4q_0^2 e^2<1$ only) scalar field setting.

\begin{rmk}
.
One actually needs a much weaker assumption than conjecture \ref{Priceconj} : the stability conclusion requires only pointwise upper bounds
on the scalar field and its derivative. The instability conclusion
additionally requires the  lower bound \eqref{lower.intro} (c.f section \ref{results} for a precise statement).
\end{rmk}

\begin{rmk}
We do not prove $C^0$ extendibility in the case $4q_0^2 e^2 \geq 1$, which remains an open problem.
\end{rmk}

\begin{rmk}
Even though we show that a $C^2$ invariant blows up, we do not show that given characteristic initial data on both event horizon satisfying our assumptions, the maximal globally hyperbolic development is (future) $C^2$ inextendible. This is because our result only applies in a neighbourhood of time-like infinity, in contrast\footnote{In \cite{JonathanStab} a special monotonicity property is exploited to propagate the curvature blow-up along the whole Cauchy horizon. Such a property is absent when $q_0 \neq 0$ or $m^2 \neq 0$. } with \cite{JonathanStab}, \cite{JonathanStabExt}. Nevertheless, it is likely that if one assumes that the data are everywhere close to  Reissner-Nordstr\"{o}m then one can use the methods of \cite{JonathanStab} to conclude $C^2$ inextendibility. We will however not pursue this.
\end{rmk}



\subsection{Ideas of proof and methods employed} \label{ideaproof}

In this last introductory section, we describe the techniques that we use to prove our main results as stated in section \ref{results} later. 
Some methods are adapted and modified from the work \cite{JonathanStab} for the stability part and \cite{KerrInstab} for the instability part.

\subsubsection{Methods for the stability part}

In the $m^2=q_0=0$ case, stability was first proven by the seminal work of Dafermos \cite{Mihalis1} in the case $s > \frac{1}{2}$. His work considers geometric quantities $(r,\phi,\varpi)$ where $\varpi$ is the modified mass defined in \eqref{electromass}, $r$ is the area-radius and $\phi$ is the scalar field. However, these quantities do not decay - in particular $\varpi$ blows-up. Remarkably, this was overcome using the very special structure of the equation. This structure is not exhibited when the mass or charge of the scalar field are present.

In contrast, the approach of Luk and Oh in \cite{JonathanStab} controls a non geometric coordinate dependent quantity $ \Omega^2$ namely the metric coefficient (c.f section \ref{Coordinates} for a definition). They actually compare $(\Omega^2,r,\phi)$ to their counterpart $(\Omega^2_{RN},r_{RN},0)$ on the Reissner-Nordstr\"{o}m background to which the space-time converges.

This has the advantage that the \textbf{difference} of these quantities and their degenerate derivatives are bounded and in fact decay towards infinity, allowing for a $C^0$ stability statement.

They establish this decay using the \textbf{non-linear wave structure in a null foliation} \bm{$(u,v)$} -as illustrated by Figure \ref{Figure2}- of the equation. 
They integrate  the difference along the wave characteristics with the help of a \textbf{bootstrap} method after splitting the space-time into smaller regions.

The result of Luk and Oh is therefore more quantitative but on the other hand it relies crucially on the hypothesis $s>1$ giving an initial \textbf{integrable decay} 
of $\Omega^2-\Omega^2_{RN}$, $r-r_{RN}$ and $\phi$.


This is why -although the method can be easily adapted in the presence of a charged and massive field- the proof fails \footnote{Essentially because $\Omega^2-\Omega^2_{RN}$, $r-r_{RN}$ and $\phi$ are no longer integrable. } for $ s \leq 1$ which is unfortunately the expectation in many interesting cases as claimed by Price's law of conjecture \ref{Priceconj}. \\

In our proof, we will again control the non-geometric coordinate dependent metric coefficients $\Omega^2$ but since the decay is so weak we cannot consider directly the difference with the background value.

Instead, we consider new natural combinations of these quantities -adapted to the geometry- which obey better estimates, notably those involving the degenerate derivatives $\partial_u$ and $\partial_v$.

In all previous work\footnote{ Notably in Dafermos' proof, the gauge derivatives of the scalar field $\frac{\partial_u \phi}{\partial_u r}$ and $\frac{\partial_v \phi}{\partial_v r}$ decay in the red-shift region and grow in the blue-shift region.}, the proof proceeds in splitting the 
space-time into a \textbf{red-shift region} near the event horizon which is very stable and a \textbf{blue-shift region} near the Cauchy horizon where many quantities can blow-up. This is illustrated by Figure \ref{Figure2}.	\begin{figure}

\begin{center}

\includegraphics[width=65 mm, height=65 mm]{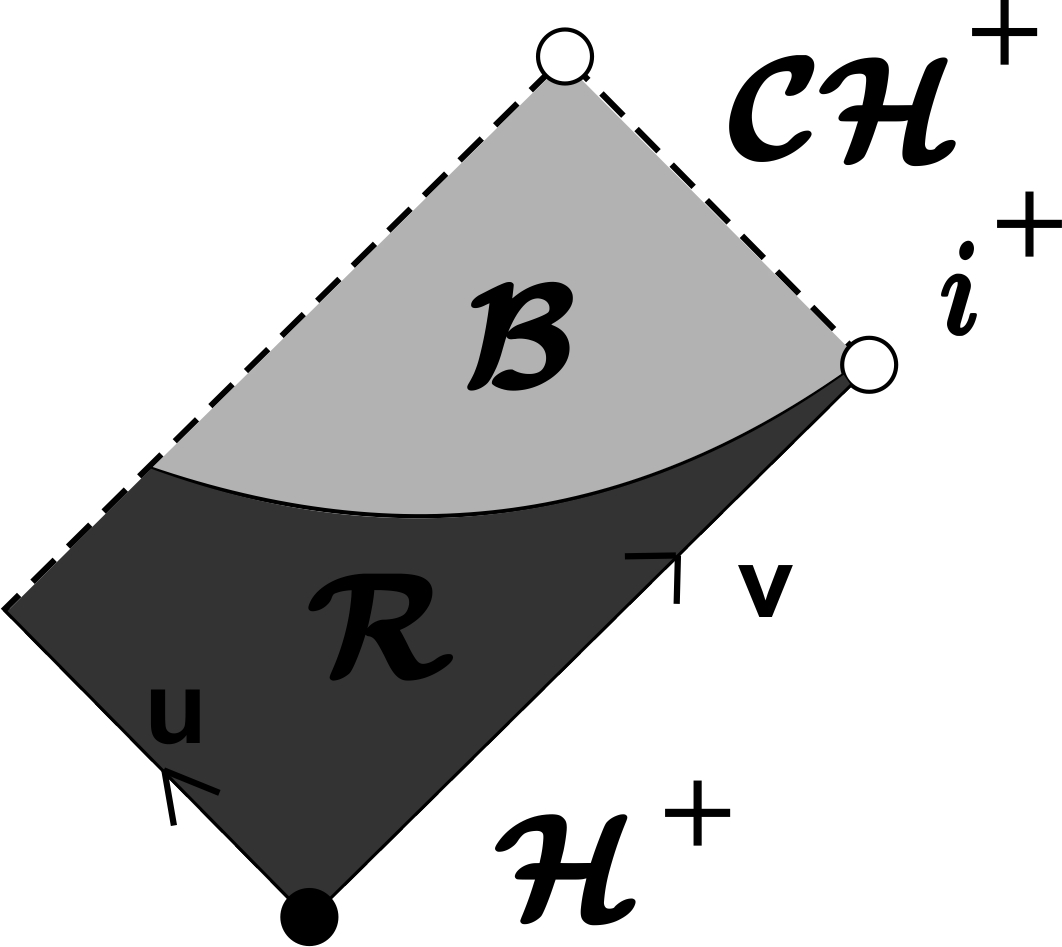}

\end{center}

\caption{Penrose diagram illustrating the division between a red-shift $\mathcal{R}$ and a blue-shift region $\mathcal{B}$.}
\label{Figure2}
\end{figure}



In the red-shift region, decay is proven using that \bm{$|\frac{-4 \partial_u r}{\Omega^2}-1|$} and \bm{$|\frac{-4 \partial_v r}{\Omega^2}-1|$} decay polynomially \footnote{Note that on Reissner-Nordstr\"{o}m, these quantities are zero.}, thanks to the Raychaudhuri equations, which allows us to \textbf{replace} \bm{$\partial_v r$} \textbf{and} \bm{$\partial_u r$} \textbf{by} \bm{$\Omega^2 \backsimeq e^{2K_+ \cdot (u+v)}$}   which enjoys an exponential structure. This avoids to lose one power when we integrate a polynomial decay on a large region c.f Lemma \ref{calculuslemma}.

In the blue-shift region,  we essentially use the \textbf{polynomial decay of} \bm{$\partial_v r$}, \bm{$\partial_u r$} and the \textbf{exponential decay of} \bm{$\Omega^2$} to propagate the estimates. 

Another important point is that we are able to find two \textbf{decaying quantities} \footnote{	These two quantities are zero on a Reissner-Nordstr\"{o}m background so we can expect them to be small in the perturbative setting. } which capture the red and blue shift effect : 
\bm{$\partial_u \log(\Omega^2) -2K$} and \bm{$\partial_v \log(\Omega^2) -2K$} -where $K$ is a geometric quantity defined by \eqref{Kdef} -
and we control the sign of $K$ : positive in the red-shift region,  negative \footnote{Except maybe close to the Cauchy horizon where $K$ may blow-up like the Hawking mass.} in the blue-shift region. 

In particular the good control of $\partial_v \log(\Omega^2) -2K$ can be fruitfully integrated to control \textbf{the smallness of \bm{$\Omega^2$} according to the different regions} but requires a bit of care close to the Cauchy horizon where $\partial_v \log(\Omega^2) -2K$ is no longer integrable in general. \\

To sum up, unlike the strategy of \cite{JonathanStab} which purely deals with differences whose decay is propagated like a wave, we mainly use propagative arguments for the scalar field only and rely on the geometry of the space-time and on the Raychaudhuri equations \eqref{RaychU}, \eqref{RaychV} to prove our estimates.

\subsubsection{Methods for the instability part} \label{methodinstab}

The first instability result is due to Dafermos in \cite{Mihalis1}. Like its stability counterpart, it relies crucially on the special structure of the equation and notably a very specific monotonicity property that does not hold in the presence of a massive or charged scalar field.

The work \cite{JonathanStab} also proves an instability statement. Nevertheless both the presence of the mass or of the charge also destroy the main argument. Indeed the argument makes use of an almost conservation law for the scalar field stress-energy tensor $\mathbb{T}^{SF}$. With a non-zero mass, a new term appears (c.f \eqref{3}) which has the wrong sign and cannot be easily controlled. 
If the field is charged, this time the two conservation laws -previously independent- coming from $\mathbb{T}^{SF}$ and $\mathbb{T}^{EM}$ are now coupled and therefore Luk and Oh's method does not apply.

Instead, we borrow ideas from a paper of Luk and Sbierski \cite{KerrInstab} in which the authors prove the linear instability of Kerr's interior. They simplify their methods and adapt them to the Reissner-Nordstr\"{o}m case \footnote{For a scalar field that is not necessarily spherically symmetric, unlike in the present paper.  } in an introductory section. The point is essentially to prove the blow-up of $\partial_V \phi$ on a constant $u$ hypersurface close to the Cauchy horizon, where $(u,V)$ is a regular coordinate system near the Cauchy horizon thanks to a \textbf{polynomial lower bound on} \bm{$\int_{v}^{+\infty} |\partial_v \phi|^2(u,v') dv'$ }.

For this they use an integrated $L^2$ stability estimate coupled with a vector field method\footnote{For an introduction to the vector field method and interesting applications c.f \cite{BH}. } - namely an energy estimate- using the Killing vector field $\partial_t=\partial_v-\partial_u $ -which boils down to the \textbf{conservation of the energy}. They manage to control the integral of $\partial_v \phi$ on the event horizon by its values on an intermediate curve $\gamma_{\sigma}$ (which marks the limit between their red-shift and their blue-shift region)  \textbf{on which} \bm{$\Omega^2$} \textbf{decays polynomially} like $v^{-\sigma}$ for a very large power $\sigma>0$.

After they control this value by the integral of $\partial_v \phi$ on a constant $u$ hypersurface close to the Cauchy horizon using again a vector field method with the vector field $\partial_v$. They conclude using \textbf{the positivity of the energy} which  allows for the $\partial_v$ terms to control the $\partial_t=\partial_v-\partial_u $ ones on $\gamma_{\sigma}$. \\

Their approach relies on the linearity of the problem and in particular the use of a Killing vector field of the Reissner-Nordstr\"{o}m background , which does not exist any more in the non-linear setting that we consider.

Another important difference is the existence -in the uncharged field case- of two independent (approximate) conservation laws, namely one for the scalar field $\mathbb{T}^{SF}$ -which the authors of \cite{KerrInstab} use- and one for the electromagnetic field $\mathbb{T}^{EM}$ - which they ignore. In our case the charged field interacts with the charge of the black hole coupling the Klein-Gordon and the Maxwell equation. This gives a single (approximate) conservation law involving $\mathbb{T} = \mathbb{T}^{KG} + \mathbb{T}^{EM}$.

Moreover, the use of a vector field method in a blue-shift region for a charged and massive scalar field generates terms which do not decay, in particular those related to the charge \footnote{Which is expected to tend to a constant $e$ so that we cannot hope for decay, unlike for $\phi$ which is zero on the underlying Reissner-Nordstr\"{o}m background.} of the black hole $Q$  and which have the inadequate sign. \\



Our approach differs from Luk--Oh in that  we instead use an energy identity from the conservation of the renormalized Hawking mass $\varpi$ between the red-shift and the blue-shift region. However, in view of the expected blow-up of $\varpi$ near the Cauchy horizon, we must calibrate the spacetime  to be    \textbf{close enough to the Cauchy horizon} so that we enjoy a sufficient decay of $\Omega^2$ in the future to propagate the decay of the wave equations but must also be \textbf{close enough to the event horizon} so that the Hawking mass $\varpi$ still converges to its asymptotic value $M>0$. Compared to \cite{KerrInstab} where the limiting curve was chosen to be as far as possible in the future, this is a completely different strategy.

We also emphasize that we require a lower bound on the event horizon on the scalar field and its phase of the following form for some $p>0$, where $(M,e)$ are the black hole asymptotic mass and charge, and $r_{\pm}(M,e)= M\pm\sqrt{M^2-e^2}$:  \begin{equation}\label{lower.intro}
\frac{1}{2}\int_{v}^{+\infty} |D_v \phi|^2_{|\mathcal{H}^+}(v')dv'
+   q_0 e\  [\frac{1}{r_-}-\frac{1}{r_+}]  \int_{v}^{+\infty} \Im(\bar{\phi}_{|\mathcal{H}^+}(v') D_v \phi_{|\mathcal{H}^+}(v')) dv' \gtrsim v^{-p},
\end{equation}  the latter term proportional to $q_0$ being absent in the uncharged-scalar-field  case. We note that \eqref{lower.intro} is conjecturally satisfied for generic solutions in the black hole exterior in view of Conjecture~\ref{Priceconj}.

This challenge is addressed using fine \textbf{stability estimates}, notably the quantities $\frac{-4 \partial_u r}{\Omega^2}$ and $\frac{-4 \partial_v r}{\Omega^2}$ which are precisely the coefficients of $T$ and that are controlled in the vicinity of $\gamma$.

In the blue-shift region, since vector field methods are now hard to use, we simply \textbf{propagate point-wise} $\partial_v \phi$ using the wave equation and the sufficient decay of $\Omega^2$ in the future of $\gamma$. We strongly rely on the stability estimates proven in the first part.

Lastly, once this lower bound is proven, we use exactly and without modifications the techniques employed in \cite{JonathanStab} to prove the blow-up of a $C^2$ geometric invariant quantity for any $s > \frac{1}{2}$ and the $H^1$ blow up of the scalar field if $s>1$, leading to the $C^1$ inextendibility of the $C^0$ extension constructed in the stability part.

\subsection{Outline of the paper} \label{outline}

We conclude this introduction by presenting the rest of the article.  

Section \ref{Section2} is devoted to preliminaries : we notably define the main notations, introduce the equations and express them in the form that we use later. A brief review of the Reissner-Nordstr\"{o}m background is also presented.

In section \ref{results}, we phrase the main results of the paper precisely, namely the stability and the instability theorems. They are preceded by a reminder on the characteristic initial value problem and the coordinate dependency.

In section \ref{proofstab}, the proof of the stability theorem is carried on. The proof of one minor proposition is deferred to appendix \ref{appendixproof} and a simple local existence lemma is proven in appendix \ref{appendixlemma}

In section \ref{proofinstab}, the proof of the instability theorem is carried on.

Finally, in the appendix \ref{appendixapparent}, we use our stability framework to ``localise'' in coordinates the part of the apparent horizon that is close to time-like infinity.

\subsection{Acknowledgements}
I would like to express my deepest gratitude to my PhD advisor Jonathan Luk for suggesting this problem, for his continuous enlightening guidance, for his precious advice, his patience and for his invaluable help to work in good conditions.

My special thanks go to Hayd\'{e}e Pacheco for her crucial graphical contribution, namely drawing the Penrose diagrams.

I also would like to thank two anonymous referees for valuable suggestions. 

I gratefully acknowledge the financial support of the EPSRC, grant reference number EP/L016516/1.

This work was completed while I was a visiting student in Stanford university and I gratefully acknowledge their financial support.

\section{Geometric framework and equations} \label{Section2}

\subsection{The equations in geometric form} \label{EinsteinMaxKG}

We look for solutions to  the Einstein-Maxwell equations coupled with a  charged and massive scalar field $\phi$ of constant mass \footnote{$m^2 \geq 0$  ensures that the dominant energy condition is satisfied. It does not play a role for the proof of the stability estimates but is crucial for the instability part.} $m^2  \geq 0$ and constant charge $q_{0} \neq 0$  propagating according to the Klein-Gordon equation in curved space-time \footnote{One important difference compared to real scalar field models is that the Maxwell and the wave equations are now coupled because the field is charged.} : 

A solution is described by a quadruplet $(M,g,\phi, F)$ - where $(M,g)$ is a Lorentzian manifold of dimension $3+1$, $\phi$ is a complex-valued \footnote{The second important difference with the uncharged case is that it is not no longer possible to take a real scalar field :  $\phi$ must be complex-valued.} function on $M$ and $F$ is a real-valued 2-form on $M$ - which satisfies the following equations :


\begin{equation} \label{1.1}   Ric_{\mu \nu}(g)- \frac{1}{2}R(g)g_{\mu \nu}= \mathbb{T}^{EM}_{\mu \nu}+  \mathbb{T}^{KG}_{\mu \nu} ,    \end{equation} 

\begin{equation}  \mathbb{T}^{EM}_{\mu \nu}=2\left(g^{\alpha \beta}F _{\alpha \nu}F_{\beta \mu }-\frac{1}{4}F^{\alpha \beta}F_{\alpha \beta}g_{\mu \nu}\right),
\end{equation}
\begin{equation}  \mathbb{T}^{KG}_{\mu \nu}= 2\left( \Re(D _{\mu}\phi \overline{D _{\nu}\phi}) -\frac{1}{2}(g^{\alpha \beta} D _{\alpha}\phi \overline{D _{\beta}\phi} + m ^{2}|\phi|^2  )g_{\mu \nu} \right), \end{equation} \begin{equation} \label{4.1} \nabla^{\mu} F_{\mu \nu}= \frac{ q_{0} }{2}i (\phi \overline{D_{\nu}\phi} -\overline{\phi} D_{\nu}\phi) , \; F=dA ,
\end{equation} \begin{equation} \label{5.1} g^{\mu \nu} D_{\mu} D_{\nu}\phi = m ^{2} \phi ,
\end{equation}

where  $D:= \nabla+  iq_{0}A$ is the gauge derivative, $\nabla$ is the Levi-Civita connection of $g$ and $A$ is the potential one-form \footnote{ $F=dA$ is to be interpreted as `` there exists real-valued a one-form $A$ such that $F=dA$ ''. This determines $A$ up to a closed form only. It means that there is a gauge freedom, c.f section \ref{Coordinates}.}. $\mathbb{T}^{EM}_{\mu \nu}$ and  $\mathbb{T}^{KG}_{\mu \nu}$ are the electromagnetic and the Klein-Gordon stress-energy tensor respectively.

\eqref{1.1} is the Einstein equation, \eqref{4.1} is the Maxwell equation and \eqref{5.1} is the Klein-Gordon equation. Note that they are all coupled one to another.

\subsection{Metric in null coordinates, mass, charge and main notations} \label{Coordinates}

Let $(M,g,\phi,F)$ be a spherically symmetric solution of the Einstein-Maxwell-Klein-Gordon equations. By this we mean that  $SO(3)$ acts on  $(M,g)$ by isometry with spacelike orbits and for all $R_{0} \in SO(3)$, the pull-back of $F$ and $\phi$ by $R_{0}$ coincides with itself. 

We define $\mathcal{Q}=M/SO(3)$, the quotient 2-dimensional manifold induced by the action of $SO(3)$.

$\Pi : M \rightarrow  \mathcal{Q} $ is the canonical projection taking a point of $M$ into its spherical orbit.

The metric on $M$ is then given by $g= g_{\mathcal{Q}}+ r^{2}d \sigma_{\mathbb{S}^{2}}$  where $g_{\mathcal{Q}}$ is the push-forward of $g$ by $\Pi$ and $d\sigma_{\mathbb{S}^{2}}$ the standard metric on the sphere.

$g_{\mathcal{Q}}$  as a general Lorentzian metric over a 2-dimensional manifold, can be written in null coordinates $(u,v)$ as a conformally flat metric :

$$ g_{\mathcal{Q}}:= - \frac{\Omega^{2}}{2} (du  \otimes dv+dv  \otimes du).$$ 

We define the area-radius function $r$ over $\mathcal{Q}$ by $r(p)= \sqrt{\frac{Area(\Pi^{-1}(p))}{4 \pi}}$.

We can then define $\kappa$ and $\iota$ as : 

\begin{equation} \label{kappa}
\kappa = \frac{-\Omega^2}{4\partial_u r}\in \mathbb{R} \cup \{ \pm \infty\} ,
\end{equation}
\begin{equation} \label{iota}
\iota = \frac{-\Omega^2}{4\partial_v r} \in  \mathbb{R} \cup \{ \pm \infty\}. 
\end{equation}

\begin{rmk}
Notice that $\kappa$ is invariant under $u$-coordinate change : if $du'=f(u)du$, then in the new coordinate system $(u',v)$, $\kappa(u',v)=\kappa(u,v)$. 
Similarly , $\iota$ is invariant under $v$-coordinate change. \footnote{Note however that rescaling $v$ also rescales $\kappa$ and rescaling $u$  rescales $\iota$.}
\end{rmk}

We can also define the Hawking mass and mass ratio as geometric quantities, at least in spherical symmetry :

$$ \rho := \frac{r}{2}(1- g_{\mathcal{Q}} (\nabla r, \nabla r )),$$
$$ \mu := \frac{2\rho}{r}.$$

In what follows, we will abuse notation and denote by $F$ the 2-form over $\mathcal{Q}$ that is the push-forward by $\Pi$ of the electromagnetic 2-form originally on $M$, and same for $\phi$.

It turns out that spherical symmetry allows us to set : 

$$ F= \frac{Q}{2 r^{2}}  \Omega^{2} du \wedge dv, $$

where $Q$ is a scalar function that we call the electric charge.

\begin{rmk}
It should be noted that in the Einstein-Maxwell-Scalar-Field of \cite{Mihalis1} and \cite{JonathanStab}, $Q \equiv e$ was forced to be a constant because it was coupled with vacuum Maxwell's equation $ div \; F=0$. \\
\end{rmk}

$F=dA$ also allows us to chose a spherically symmetric potential 
$A$  written as : 

$$ A = A_u du + A_v dv .$$

The equations of section \ref{EinsteinMaxKG} are invariant under the following gauge transformation : 

$$ \phi \rightarrow  e^{-i q_0 f } \phi ,$$
$$ A \rightarrow  A+ d f. $$

where $f$ is a smooth real-valued function. 

Therefore we can choose the following gauge for some constant $v_0$ and for all $(u,v)$ :  

\begin{equation} \label{gaugeAv}
A_{v}(u,v) \equiv 0,
\end{equation}
\begin{equation} \label{gaugeAponctuelle}
A_{u}(u,v_{0})=0.
\end{equation}

\begin{rmk}
Notice that this gauge depends only on the null foliation and therefore is invariant if $u$ or $v$ is re-parametrized.
\end{rmk}

This gauge will be used in the rest of the paper, for $v_0$ to be specified in the statement of Theorem \ref{Stabilitytheorem}.

For a  more justified and complete discussion of the Einstein-Maxwell-Klein-Gordon setting, c.f \cite{Kommemi}. \\

Now we introduce the modified mass $\varpi$ that takes the charge $Q$ into account : 

\begin{equation} \label{electromass}
\varpi := \rho + \frac{Q^2}{2r}= \frac{\mu r}{2} + \frac{Q^2}{2r} .
\end{equation} 

An elementary computation relates coordinate-dependent quantities to geometric \footnote{Notice that $1-\mu$ and $K$ do \underline{not} depend on the coordinate choice $(u,v)$.} ones :

\begin{equation} \label{mu}
1 - \mu = \frac{-4 \partial_u r \partial_v r}{\Omega^2} = \frac{-\Omega^2 }{4 \iota \kappa}= 1- \frac{2 \varpi}{r}+ \frac{Q^2}{r^2}.
\end{equation}

We then define the geometric quantity\footnote{On Reissner-Nordstr\"{o}m,  $2K= \partial_u \log |1-\mu|  =\partial_v \log |1-\mu|$. } $2K$ :

\begin{equation} \label{Kdef}
2K = \frac{2}{r^2}(\varpi- \frac{Q^2}{r}).
\end{equation} 

We will also denote, for fixed constants $M$ and $e$ :

$$ 2K_{M,e}(r)=\frac{2}{r^2}(M- \frac{e^2}{r}).$$

Finally we introduce the following notation, first used by Christodoulou : 

$$ \lambda = \partial_v r ,$$
$$ \nu = \partial_u r. $$

\subsection{The Reissner-Nordstr\"{o}m solution} \label{RNsolution}
In this section we present the sub-extremal Reissner-Nordstr\"{o}m solution. Because the space-time that we consider converges at late time towards a member of the Reissner-Nordstr\"{o}m family and that we aim at proving stability estimates, it is important to recall their main qualitative features to see which are conserved in the presence of a perturbation.
\subsubsection{The Reissner-Nordstr\"{o}m interior metric}

The Reissner-Nordstr\"{o}m black hole is a 2-parameter family of spherically symmetric and static space-times indexed by the charge and the mass $(e,M)$, which satisfy the Einstein-Maxwell equations i.e the system of section \ref{EinsteinMaxKG} with $\phi \equiv 0$ with  $\mathbb{R}_{+}^{*} \times \mathbb{S}^{2}$ initial data.

We are interested in \underline{sub-extremal} Reissner-Nordstr\"{o}m black holes, which is expressed by the condition  $0<|e|<M$. 

Define then for such $(e,M)$ : $$ r_+(M,e) = M+ \sqrt{M^2-e^2}>0,$$
$$ r_-(M,e) = M- \sqrt{M^2-e^2}>0.$$

The metric in the interior of the black hole can be written in coordinates as :

\begin{equation} \label{RN}
g_{RN}=\frac{\Omega^{2}_{RN}}{4}dt^{2}-4\Omega^{-2}_{RN}dr^{2}+r^{2}[ d\theta^{2}+\sin(\theta)^{2}d \psi^{2}],
\end{equation}	\begin{equation} \label{OmegaRN}
\Omega^{2}_{RN}(r):=-4(1-\frac{2M}{r}+\frac{e^2}{r^2}),
\end{equation}

Where $(r,t,\theta, \psi) \in (r_{-},r_{+}) \times \mathbb{R} \times  [0,\pi) \times [0, 2 \pi]$ .

\subsubsection{$(u,v) $ coordinate system on Reissner-Nordstr\"{o}m background}

We have seen in Section \ref{Coordinates} how to build any null coordinate $(u,v)$. Now that the metric is explicit, we would like to find such a $(u,v)$ system that is related to the variables $(r,t)$ appearing in equation \eqref{RN}.

Define

$$r^{*}=r+ \frac{1}{2 K_{+}} \log(r_{+}-r)+\frac{1}{2 K_{-}}  \log(r-r_{-}), $$

where $ 2K_+(M,e)$ and  $ 2K_-(M,e)$, respectively called the surface gravity \footnote{For an physical explanation of the terminology, c.f \cite{Reall}.} of the event horizon and the surface gravity of the Cauchy horizon, are defined by \footnote{Note that  $K_-<0$ like in \cite{JonathanInstab} but unlike in \cite{JonathanStab}.} : 

$$ K_+(M,e)= \frac{1}{ r_+^2}(M - \frac{e^2}{ r_+})= \frac{ r_+- r_-}{2 r_+^2}>0,$$

$$ K_-(M,e)= \frac{1}{ r_-^2}(M - \frac{e^2}{ r_-})= \frac{ r_- - r_+}{2 r_-^2}<0.$$ 

\begin{rmk}
Note that if $\varpi=M$ and $Q=e$ then
$K(r_+)= K_+(M,e)>0$ and $K(r_-)= K_-(M,e)<0$, where $K$ is defined in equation \eqref{Kdef}.

\end{rmk}

We then set $(u,v) \in \mathbb{R}  \times  \mathbb{R}  $ as :

$$ v= \frac{1}{2}(r^{*}+t), \, u= \frac{1}{2}(r^{*}-t), $$

and claim that equation \eqref{RN} can then be rewritten as : 

$$ g_{RN} =- \frac{\Omega^{2}_{RN}}{2} (du  \otimes dv+dv  \otimes du)+r^{2}[ d\theta^{2}+\sin(\theta)^{2}d \psi^{2}].$$

\subsubsection{Behaviour of $\Omega^2_{RN}$}

We define \footnote{We could have defined in more generality the event horizon to be the past boundary of the black hole region and the Cauchy horizon the future boundary of the maximal globally hyperbolic development. Strictly speaking the Cauchy horizon is not part of the space-time but  can be attached as a double null boundary and we then consider the space-time as a manifold with corners.} the event horizon $\mathcal{H}^+=\{ u\equiv -\infty, v\in \mathbb{R} \}$, and the Cauchy horizon $\mathcal{CH}^+=\{ v\equiv +\infty, u \in \mathbb{R} \}$

$\Omega^2_{RN}$ cancels on both $\mathcal{H}^+$ and $\mathcal{CH}^+$. A computation shows that : 

$$ \Omega^{2}_{RN} \sim_{ r \rightarrow r_{+} } C_{e,M} e^{2K_{+}r^{*}}=C_{e,M} e^{2K_{+}\cdot (u+v)},$$

and similarly that :

$$ \Omega^{2}_{RN} \sim_{ r \rightarrow r_{-} } C'_{e,M} e^{2K_{-}r^{*}}= C'_{e,M} e^{2K_{-}\cdot (u+v)},$$

for some $C_{e,M}>0$, $C_{e,M}'>0$.

\begin{rmk}
Notice that $\Omega^2_{RN}$ exhibits an exponential behaviour in $(u+v)$, exponentially increasing from 0 near the event horizon and exponentially decreasing to 0 near the Cauchy horizon.
\end{rmk}

Notice also that for $r$ bounded away from $r_{+}$ and $r_{-}$, $ \Omega^{2}_{RN}$ is upper and lower bounded. 

\subsubsection{Kruskal coordinates $(U,V)$ and Eddington-Finkelstein coordinates $(U,v)$, $(u,V)$}

From the previous section, one could fear that the metric could be singular across the horizons $\mathcal{H}^+$ and $\mathcal{CH}^+$. Actually it is not : like for the Scharwzchild's event horizon horizon, it suffices to define Kruskal-like coordinates $(U,V)$ from the $(u,v)$ coordinates as : 
$$U:=\frac{1}{2K_{+}}e^{2K_{+}u}, $$ and $$
V:=1-\frac{1}{2|K_{-}|}e^{2K_{-}v}.$$

Note that $U$ and $V$ now range in $(U,V) \in [0, +\infty) \times  (-\infty, 1]$ and that 
$\mathcal{H}^+ = \{U \equiv 0\}$ ; $\mathcal{CH}^+ = \{V \equiv 1\}$. \\

We then write the metric in the Eddington-Finkelstein-type mixed $(U,v)$ coordinates as : 

$$ g_{RN}:= - \frac{\Omega^{2}_{RN,H}}{2} (dU  \otimes dv+dv  \otimes dU)+r^{2}[ d\theta^{2}+\sin(\theta)^{2}d \psi^{2}].$$

We find that $(U,v)$ is a regular coordinate system near the event horizon $\mathcal{H}^+$ :  

$$\Omega^{2}_{RN,H}(U,v)= -\frac{1}{2K_{+}U}(1-\frac{2M}{r}+\frac{e^2}{r^{2}}) \rightarrow_{U \rightarrow 0  } C_{e,M} e^{2K_{+}v}.$$

In $(u,V)$ coordinates we write now write the metric as : 

$$ g_{RN}:= - \frac{\Omega^{2}_{RN,CH}}{2} (du  \otimes dV+dV  \otimes du)+r^{2}[ d\theta^{2}+\sin(\theta)^{2}d \psi^{2}].$$

We then see that $(u,V)$ is a regular \footnote{Moreover, as mentioned in remark \ref{RNSCC} the metric is actually smoothly extendible beyond $\mathcal{CH}^+$, which would pose a problem for the strong cosmic censorship conjecture  but does not because Reissner-Nordstr\"{o}m is expected to be non generic.}
coordinate system near the Cauchy horizon $\mathcal{CH}^+$ : 

$$\Omega^{2}_{RN,CH}(u,V)= \frac{1}{2K_{-}(1-V)}(1-\frac{2M}{r}+\frac{e^2}{r^{2}}) \rightarrow_{V \rightarrow 1  } C_{e,M}' e^{-2K_{-}u}.$$

\subsubsection{Constant quantities on Reissner-Nordstr\"{o}m}

Since we consider the stability of a Reissner-Nordstr\"{o}m background under perturbation, it is useful to identify which quantities are zero on this fixed background : these are the ones that we can hope decay for in the non-linear perturbative setting with the Klein-Gordon field. 

Reissner-Nordst\"{o}rm has four main qualitative features which distinguishes it from general dynamical solutions :

\begin{enumerate}
\item  Both the charge and the modified mass are fixed : 

$$ \varpi \equiv M, $$

$$ Q \equiv e. $$

Hence $1-\mu = 1- \frac{2M}{r}+ \frac{e^2}{r^2}$.

\item The metric is symmetric \footnote{Which is essentially equivalent to the fact that $\partial_t$ is a Killing vector field or that $\Omega^2_{RN}(r)$ is a sole function of $r$.} in $u$ and $v$  and in particular :  

$$\partial_{u}r=\partial_{v}r,$$

$$ \partial_{u}\log(\Omega^2_{RN})=\partial_{v}\log(\Omega^2_{RN})= \frac{2}{r^2}(M - \frac{e^2}{r})= 2K_{M,e}(r).$$

\item The horizons are \underline{constant $r$} null hypersurfaces :  $$\mathcal{H}^+=\{ u\equiv -\infty, v\in \mathbb{R} \}=\{r \equiv r_{+}\},$$ $$\mathcal{CH}^{+}=\{ v\equiv +\infty, u \in \mathbb{R} \}=\{r\equiv r_{-}\}.$$ Hence $\partial_{v}r_{|\mathcal{H}^+} \equiv 0$ and  $\partial_{u}r_{|\mathcal{CH}^+} \equiv 0$ which is consistent with the following relation : 

$$ \partial_{u}r=\partial_{v}r= 1- \frac{2M}{r}+ \frac{e^2}{r^2}.$$

\item The event horizon $\mathcal{H}^+$ coincides with the apparent horizon $\mathcal{A}:= \{ \partial_v r =0\}$ so all the 2-spheres inside the black hole are trapped.

This does not hold for general dynamical space-times where $\mathcal{A}$ is strictly in the future of $\mathcal{H}^+$. 

However, in the perturbative regime, we can expect that $\mathcal{A}$ is not too far \footnote{Actually we can prove that $ 0 \leq \partial_v r_{|\mathcal{H}^+} \lesssim v^{-2s}$ if $|\partial_v \phi| \lesssim v^{-s}$.} from $\mathcal{H}^+$, c.f Appendix \ref{appendixapparent}.

\end{enumerate}

In the end, we can sum up all the relations by : 

$$\partial_{u}r=\partial_{v}r=1-\frac{2M}{r}+\frac{e^2}{r^{2}}=-\frac{\Omega^{2}_{RN}}{4} \leq 0, $$

which also means that :  

$$ \kappa_{RN} = \iota_{RN} \equiv 1. $$

\subsection{The Einstein-Maxwell-Klein-Gordon equations in null coordinates} \label{eqcoord}

Finally, we express the Einstein-Maxwell-Klein-Gordon system in spherical symmetry in any (u,v) coordinates as in section \ref{Coordinates} and under the gauge choice for the potential \eqref{gaugeAv}, \eqref{gaugeAponctuelle}. 

We start by the wave part of the Einstein equation :

\begin{equation}\label{Radius}\partial_{u}\partial_{v}r =\frac{- \Omega^{2}}{4r}-\frac{\partial_{u}r\partial_{v}r}{r}
+\frac{ \Omega^{2}}{4r^{3}} Q^2 +  \frac{m^{2}r }{4} \Omega^2 |\phi|^{2}  =-\frac{\Omega^2}{4}.2K +    \frac{m^{2}r }{4} \Omega^2 |\phi|^{2} , \end{equation}

\begin{equation}\label{Omega}
\partial_{u}\partial_{v} \log(\Omega^2)=-2\Re(D_{u} \phi \partial_{v}\bar{\phi})+\frac{ \Omega^{2}}{2r^{2}}+\frac{2\partial_{u}r\partial_{v}r}{r^{2}}- \frac{ \Omega^{2}}{r^{4}} Q^2,
\end{equation}

the Raychaudhuri equations : 

\begin{equation}\label{RaychU}\partial_{u}(\frac {\partial_{u}r}{\Omega^{2}})=\frac {-r}{\Omega^{2}}|  D_{u} \color{black}\phi|^{2}, \end{equation}

\begin{equation} \label{RaychV}\partial_{v}(\frac {\partial_{v}r}{\Omega^{2}})=\frac {-r}{\Omega^{2}}|\partial_{v}\color{black}\phi|^{2},\end{equation}

the Klein-Gordon wave equation : 

\begin{equation}\label{Field}
\partial_{u}\partial_{v} \phi =-\frac{\partial_{u}\phi\partial_{v}r}{r}-\frac{\partial_{u}r \partial_{v}\phi}{r} +\frac{ q_{0}i \Omega^{2}}{4r^{2}}Q \phi
-\frac{ m^{2}\Omega^{2}}{4}\phi- i q_{0} A_{u}\frac{\phi \partial_{v}r}{r}-i q_0 A_{u}\partial_{v}\phi,\end{equation}

and the propagative part of Maxwell's equation : 

\begin{equation} \label{chargeUEinstein}
\partial_u Q = -q_0 r^2 \Im( \phi \overline{ D_u \phi}) .
\end{equation}

\begin{equation} \label{ChargeVEinstein}
\partial_v Q = q_0 r^2 \Im( \phi \overline{\partial_v \phi}) .
\end{equation}

Also the existence of an electro-magnetic potential $A$ implies that :

\begin{equation} \label{potentialEinstein}
\partial_v A_u = \frac{-Q\Omega^2}{2r^2} .
\end{equation}

Now we can reformulate the former equations to put them in a form that is more convenient to use. 

It is interesting to use \eqref{Radius}, \eqref{RaychU}, \eqref{RaychV}, \eqref{chargeUEinstein}, \eqref{ChargeVEinstein} to derive an equation for the modified mass :

\begin{equation} \label{massUEinstein}
\partial_u \varpi = \frac{r^2}{2\iota}| D_u\color{black} \phi |^2  + \frac{m^2}{2} r^2 \partial_u r|\phi|^2 -   q_0  Q r \Im( \phi \overline{D_u \phi}),
\end{equation}
\begin{equation} \label{massVEinstein} 
\partial_v \varpi = \frac{r^2}{2\kappa}| \partial_v \phi |^2 +  \frac{m^2}{2} r^2\partial_v r |\phi|^2+   q_0  Q r \Im( \phi \overline{\partial_v \phi}).
\end{equation}

Moreover, the following reformulation of \eqref{Radius} will be  useful :

\begin{equation} \label{Radius2}
\partial_v \log( | \partial_u r |) = \kappa (2K - r m^2 |\phi|^2).
\end{equation} 

\begin{rmk}
Note that the left-hand-side, like $\kappa$ is invariant under u-coordinate changes.
\end{rmk}

We also reformulate \eqref{Omega} as : 	
\begin{equation} \label{Omega2}
\partial_{u}\partial_{v} \log(\Omega^2)=\kappa \partial_u (2K)-2\Re(D_{u} \phi \partial_{v}\bar{\phi})- \frac{2\kappa}{r^2}( \partial_u \varpi- \frac{\partial_u Q^2}{r})= \\
\iota \partial_v (2K) -2\Re(D_{u} \phi \partial_{v}\bar{\phi})-  \frac{2\iota}{r^2}( \partial_v \varpi- \frac{\partial_v Q^2}{r}).
\end{equation}

We can also rewrite \eqref{Field} to control $|\partial_{v} \phi|$ more easily : 

\begin{equation}\label{Field2}
e^{-iq_0 \int_{u_{0}}^{u}A_u} \partial_{u}( e^{iq_0 \int_{u_{0}}^{u}A_u}\partial_{v} \phi) =-\frac{\partial_{v}r D_u\phi}{r}-\frac{\partial_{u}r \partial_{v}\phi}{r} +\frac{ q_{0}i \Omega^{2}}{4r^2}Q \phi
-\frac{ m^{2}  \Omega^{2}}{4}\phi, \end{equation}

or to control $|D_u \phi|$ more easily :  

\begin{equation}\label{Field3}
\partial_v(D_u\phi)
=-\frac{\partial_u r\,\partial_v\phi}{r}
-\frac{\partial_v r\,D_u\phi}{r}
-\frac{iq_0Q\Omega^2}{4r^2}\phi
-\frac{m^2\Omega^2}{4}\phi.
\end{equation}
Finally we can also write the Raychaudhuri equations as : 

\begin{equation}\label{RaychU2}\partial_{u}(\kappa^{-1})=\frac {4r}{\Omega^{2}}|  D_{u} \color{black}\phi|^{2}, \end{equation}
\begin{equation} \label{RaychV2}\partial_{v}( \iota^{-1})=\frac {4r}{\Omega^{2}}|\partial_{v}\color{black}\phi|^{2}.\end{equation}

\section{Precise statement of the main results} \label{results}

\subsection{Preliminaries on characteristic initial value problem and coordinate choice}	

Before stating the theorem, we want to demystify a little the framework used to define the gauges and the coordinate dependent objects. The context is the same as for \cite{Mihalis1} and  \cite{JonathanStab} , the only difference is the presence of the (dynamical) charge of the space-time $Q$.

We want to phrase the characteristic initial value problem for the Einstein-Maxwell-Klein-Gordon system of section \ref{EinsteinMaxKG}. The reader familiar with the framework can skip this section. \\

We first consider two connected and oriented smooth, 1-dimensional manifolds  $\mathcal{C}_{in} $ and $ \mathcal{C}_{out}$ -each with a boundary point (c.f Figure \ref{MihalisPenrose}.) 


We can identify the surfaces at their boundary point to get $\mathcal{C}_{in} \cup_{\{p\}} \mathcal{C}_{out}$, on which we now want to build a $(U,v)$ \textbf{null} regular coordinate system. For this, we have four choices to make : 

\begin{enumerate}

\item Choosing an increasing \footnote{By increasing, we mean parallel to the orientation of the 1-dimensional surface.}  parametrization $U$ of $\mathcal{C}_{in}$.

\item Choosing an increasing parametrization $v$ of $\mathcal{C}_{out}$.

\item Choosing the $U$-coordinate  $U_0 \in \mathbb{R} \cup \{ \pm \infty\}$ of the intersection point p.

\item  Choosing the $v$-coordinate  $v_0 \in \mathbb{R} \cup \{ \pm \infty\}$ of the intersection point p.

\end{enumerate}

In this coordinate system, $\mathcal{C}_{in}$ and $  \mathcal{C}_{out}$ can be written as : 

$$ \mathcal{C}_{in} = \{ (U, v_0), U \in [U_0, U_{max}) \},$$

$$ \mathcal{C}_{out} = \{ (U_0, v), v \in [v_0, v_{max}) \},$$

with $U_{max}\in \mathbb{R} \cup \{ \pm \infty\}$,  $v_{max} \in \mathbb{R} \cup \{ \pm \infty\}$. \\






As our initial data we shall consider $(r,\Omega^2_H,\phi,A)$  as follows : 

$ (r,\phi)$ are $C^1$ scalar \footnote{It should be emphasized that $r$ and $\phi$ -like the metric $g$ will be later- are geometric quantities, namely they do not depend on the coordinate choice. However $\Omega^2_H$ \textbf{does} depend on the coordinate choice.} functions and $A$ a $C^1$ 1-form on $\mathcal{C}_{in} \cup_{\{p\}} \mathcal{C}_{out}$.

$r$ and $\phi$ induce -in the $(U,v)$ coordinate system- some  functions on $ [v_0, v_{max}) \times \{U_0\} \cup \{v_0\} \times [U_0, U_{max})$ that we still call $r$ and $\phi$ by notation abuse. 

$A$ induces a function $A_v$ on $ [v_0, v_{max}) \times \{u_0\}$ by $A_{|C_{out}} = A_v dv$
and another function $A_U$ on $ \{v_0\} \times [U_0, U_{max})$ by   $A_{|C_{in}} = A_U dU$.

The remaining part of the data will be a $C^1$ function $\Omega^2_H :  [v_0, v_{max}) \times \{u_0\} \cup \{v_0\} \times [U_0, U_{max}) \rightarrow \mathbb{R}_+^{*}$. We will use this later to build a metric of the form $g =  - \frac{\Omega^{2}_H}{2} (dU  \otimes dv+dv  \otimes dU)++r^{2}[ d\theta^{2}+\sin(\theta)^{2}d \psi^{2}]$.

The prescription of $\Omega^2_H$ as above will be coordinate dependent.

This coordinate dependent framework allows us to define the Raychaudhuri equations on the initial surfaces, seen as constraints for the characteristic initial value problem.

However, they are still valid under any re-parametrization of $U$ or $v$ :  

\begin{defn} [Raychaudhuri equations]		We say that the data $(r,\Omega^2_H, \phi, A)$ satisfy the Raychaudhuri  equations if on $\{v_0\} \times [U_0, U_{max})$ : 		$$\partial_{U}(\frac {\partial_{U}r}{\Omega^{2}_{H}})=\frac {-r}{\Omega^{2}_{H}}|  D_{U} \phi|^{2}. $$		And on $ [v_0, v_{max}) \times \{U_0\} $ : 	$$\partial_{v}(\frac {\partial_{v}r}{\Omega^{2}_{H}})=\frac {-r}{\Omega^{2}_{H}}|D_{v}\phi|^{2},$$		where $D$ depends on $A$ by $ D = \partial + i q_0 A$ as an operator on scalar functions.		\end{defn}

We now want to talk of ``the solution'' - up to gauge transforms- of the Einstein-Maxwell-Klein-Gordon equations. To do so, we solve the partial differential equation system of section \ref{eqcoord} ``abstractly'' for some data $(r,\Omega^2,\phi,A)$. Since it is standard that the Raychaudhuri equations -once satisfied on the initial surfaces- are propagated, we see the solution actually satisfies the  Einstein-Maxwell-Klein-Gordon equations in their geometric form of section \ref{EinsteinMaxKG}.

\begin{thm}[Characteristic initial value problem] \label{ExistenceMGHD}

Let $\mathcal{C}_{in} $, $ \mathcal{C}_{out}$ be as before.

We assume moreover that the data $(r,\Omega^2_H, \phi, A)$ are as before and satisfy the Raychaudhuri equations. Moreover we suppose that $r>0$. \\

Then   there exists a unique $C^{1}$ maximal globally hyperbolic development $(M,g,\phi,F)$, spherically symmetric solution of Einstein-Maxwell-Klein-Gordon equations of section \ref{EinsteinMaxKG}  such that \begin{enumerate}

\item $\mathcal{C}_{out}$ and $\mathcal{C}_{in}$ embed into $\mathcal{Q}=\mathcal{M}/SO(3)$ as null boundaries with respect to the metric $g$.

\item     $D^{+}(\mathcal{C}_{in} \cup_{\{p\}} \mathcal{C}_{out} )\cap \mathcal{Q}=J^{+}({\{p\}})\cap \mathcal{Q}$ 

where $D^{+}$ denotes the future domain of dependence and $J^{+}$ the causal future \footnote{For a definition c.f \cite{Reall}.}.

\item  $(M,g,\phi,F)$ satisfy : 

$$g =  - \frac{\Omega^{2}_H}{2} (dU  \otimes dv+dv  \otimes dU)+r^{2}[ d\theta^{2}+\sin(\theta)^{2}d \psi^{2}],$$

$$F=dA.$$ And  $(r,\Omega^2_H,\phi, A)$ restrict on the initial surfaces to the value prescribed

by the initial data $(r,\Omega^2_H,\phi, A)_{|\mathcal{C}_{in} \cup_{\{p\}} \mathcal{C}_{out}}$.

\item The equations in null coordinates of section \ref{eqcoord} are satisfied.

\end{enumerate}

\end{thm}

For a more thorough discussion of the uniqueness problem in that framework, c.f \cite{Costa}.

\subsection{The stability theorem}

We can now formulate the main stability theorem. The main point is the presence of a Cauchy horizon, reflected by the form of the Penrose diagram, instead of space-like Schwarzschild-type singularity.

\begin{thm}[Non-linear stability theorem] \label{Stabilitytheorem}

Let $\mathcal{C}_{in} $, $ \mathcal{C}_{out}$ and $(r, \phi, \Omega^2_H, A)$ satisfy the assumptions of Theorem \ref{ExistenceMGHD}.

Moreover, we will make the following geometric assumptions : 

\begin{hyp}
$\mathcal{C}_{out}$ is affine complete \footnote{We define affine completeness by the relation  $  \int_{v_0}^{v_{max}} \Omega^2_H(U_{0},v) dv = +\infty$. This is a \textbf{coordinate-independent} statement. }.	
\end{hyp}
\begin{hyp}
$r>0$ is a strictly decreasing function on	$\mathcal{C}_{in}$ with respect to any increasing parametrization.
\end{hyp}

From now on we will denote $\mathcal{H}^+:=\mathcal{C}_{out}$ and call $\mathcal{H}^+$ the event horizon.

For some constant $v_0 >0$ , we parametrize  $\mathcal{H}^+:=\mathcal{C}_{out} = \{ U \equiv 0 , v \geq v_0\}$ with a coordinate $v$  defined \footnote{It is then easy to see that \eqref{Radius} and assumption \ref{fieldevent} together with the affine completeness prove that $v_{max}= +\infty$.} by 

\begin{equation} \label{gauge2}
\kappa_{|\mathcal{H}^+} = (\frac{-\Omega^2_{H}(0,v)}{4\partial_{U}r(0, v)})_{|\mathcal{H}^+}\equiv 1,
\end{equation}

and for some $U_{max} >0$, we parametrize $\mathcal{C}_{in} = \{ v \equiv v_0 , 0 \leq  U \leq U_{max} , \}$  with a coordinate $U$ defined by 
\begin{equation} \label{gauge1}
(\partial_{U}r)_{|\mathcal{C}_{in}}(U, v_0) \equiv -1.
\end{equation}

We also make the following no-anti-trapped surfaces \footnote{Notice that this assumption together with \eqref{RaychU} proves that $\partial_U r <0$ everywhere on the space-time.} assumption : 

\begin{hyp}
$	\partial_U r(0,v)_{| \mathcal{H}^+} <0$
\end{hyp}

We assume the following decay on the field in $(U,v)$ coordinates  :  there exists $ C>0$ and $s>\frac{1}{2}$ such that 

\begin{hyp} \label{fieldevent}

$$|  \phi(0,v)| _{|\mathcal{H}^+} +| \partial_{v} \phi(0,v)| _{|\mathcal{H}^+} \leq  Cv^{-s},$$

\end{hyp}

\begin{hyp} \label{fieldUevent} \footnote{Notice that in the gauge \eqref{gaugeAponctuelle}, this is equivalent to saying $| \partial_{U} \phi| (U, v_0) \leq C. $}\label{fieldv0}
$$| D_{U} \phi| (U, v_0) \leq C. $$

\end{hyp}

We also ask the following convergences towards infinity on the event horizon :

\begin{hyp} \label{radiusevent}
$$r_{|\mathcal{H}^+}(0,v) \rightarrow r_{\infty} $$ as $v  \rightarrow +\infty,$ 

where $r_{\infty} >0$ is a constant.
\end{hyp}

\begin{hyp} \label{hypchargeevent}
$$ 0< Q_{+}<r_{\infty}, $$ 

where $Q_{+}:=\limsup_{v \rightarrow +\infty} |Q|_{|\mathcal{H}^+} $
\end{hyp}

We consider the unique $C^{1}$ maximal globally hyperbolic development $(M,g,\phi,F)$ of Theorem \ref{ExistenceMGHD}, 

Then, after restriction to a small enough connected subset $ p \in \mathcal{C}_{in}' \subset \mathcal{C}_{in} $, i.e  $\mathcal{C}_{in}' = \{ v \equiv v_0 , 0 \leq  U \leq U_{s} , \}$ for $0 < U_{s} $ small enough,  
$D^{+}(\mathcal{C}_{in}' \cup_{\{p\}} \mathcal{C}_{out} )\cap \mathcal{Q}$ has the Penrose diagram of Figure \ref{MihalisPenrose}. \\

Moreover, if $s>1$, $(M,g,\phi,F)$ admits a continuous extension to the Cauchy horizon. 

More precisely, we can attach a future null boundary $\mathcal{CH}^+:=\{ v \equiv +\infty , 0 \leq  U \leq U_{s}  \}$ to the space-time (M,g) such that $(g,\phi,F)$ each admits a continuous extension to the new space-time $\bar{M}:=M \cup \mathcal{CH}^+$ seen as a manifold with boundaries.

\end{thm}

\begin{rmk} Because of \eqref{mu}, 
\eqref{gauge2} is exactly equivalent\footnote{
Notice that the gauge \ref{gauge2} is the same as \cite{Mihalis1} but slightly different from \cite{JonathanStab}, although it actually only differs from a multiplicative function of $v$ bounded above and below. } to : 
\begin{equation} \label{gauge2hawking}
\partial_v r_{|\mathcal{H}^+} = 1-\frac{2 \varpi}{r}+\frac{Q^2}{r^2}=1-\mu.
\end{equation}	
\end{rmk}

\begin{rmk}  The present paper introduces the first stability result dealing with \textbf{all} the possible values of $m^2$ and $ q_0$.
However the continuous extension statement when $s>1$ was already established in the work  \cite{Mihalis1} and \cite{JonathanStab} although stated in the chargeless case $q_0=0$ only. Some continuous extension results for the charged case have also been proved in \cite{Kommemi}. Notice (c.f section \ref{Price}) that the case $s>1$ should be relevant in our context only if the scalar field is massless and not too charged \footnote{Namely $m^2=0$ and $|q_0 e| < \frac{1}{2} $ with the notation of section \ref{Price}.} compared to the black hole.
\end{rmk}

\begin{rmk}Notice also that the assumptions are (almost) the same as those of  \cite{JonathanStab}, except for the strength of the decay rate, which was integrable unlike in the present paper. \\
\end{rmk}

In the rest of the paper, we will write $A \lesssim B $ if there exists a constant $\tilde{C}=\tilde{C}(C,Q_+,q_0,m^2,r_{\infty},s,v_0)$ such that $A \leq \tilde{C} B $. 

If we need to specify this constant, we shall call it consistently $\tilde{C}$ when there are no ambiguities.

We denote also $ A \sim B$ if $A \lesssim B$ and $B \lesssim A$.


\subsection{The instability theorem}

We can now phrase our instability theorem that relies very much on the non-linear stability claimed in the preceding section.

\begin{thm} [Non-linear instability theorem] \label{Instabilitytheorem}
Let $\mathcal{C}_{in} $, $ \mathcal{C}_{out}$ and $(r,\Omega^2_H, \phi,A)$ satisfying all the assumptions of Theorem \ref{Stabilitytheorem}  and in particular assumption \ref{fieldevent} with  $s > \frac{1} {2}$.

We assume, using the same gauges as for Theorem \ref{Stabilitytheorem}, that the field in addition satisfies the following $L^2$ averaged polynomial lower-bound on the event horizon $C_{out}=\mathcal{H}^+$: 

\begin{hyp} \label{instabhyp}

\begin{equation}		\label{instabapparent1}	v^{-p} \lesssim \frac{r_+^2}{2}\int_{v}^{+\infty} |\partial_v \phi|^2_{|\mathcal{H}^+}(v')dv'
	+   q_0 e\  [\frac{r_+}{r_-}-1] r_+ \int_{v}^{+\infty} \Im(\bar{\phi}_{|\mathcal{H}^+}(v') \partial_v \phi_{|\mathcal{H}^+}(v')) dv' \end{equation}

for  $ 2s-1 \leq p < \min \{2s, 4s-2\}$.

\end{hyp}

Then for any $u \in \mathbb{R}$ negative enough, and for all $v$ large enough (depending on $u$), 

\begin{equation} \label{instabilityresult}
\int_{v}^{+\infty} |\partial_v \phi|^2(u,v')dv' \gtrsim   v^{-p}.
\end{equation}

In particular the following component of the curvature blows-up on the Cauchy horizon : 	

$$ \limsup_{v \rightarrow +\infty} Ric( \Omega^{-2} \partial_v,\Omega^{-2} \partial_v)(u,v) = +\infty. $$

Moreover for $s>1$,  $\phi \notin W^{1,2}_{loc}$ and the metric is not $C^1$ for the continuous extension constructed in Theorem \ref{Stabilitytheorem}.

\end{thm}

\begin{rmk}
This theorem is the very first instability result outside the uncharged and massless case. As explained in section \ref{methodinstab}, the methods of previous instability works do not apply here.
\end{rmk}

\begin{rmk}
In view of the result of \cite{JonathanStab}, one can very reasonably hope that this curvature blow up leads to a \bm{$C^2$} \textbf{inextendibility of the metric} in an appropriate global setting \footnote{At least for two-ended black holes.}. The reason for this is that  $Ric( \Omega^{-2} \partial_v,\Omega^{-2} \partial_v)$ is a geometric quantity since $\Omega^{-2} \partial_v$ is a geodesic vector field. The only remaining argument is to extend the blow-up far from time-like infinity namely to get a global statement as opposed to perturbative. 
\end{rmk}
\begin{rmk}
Note that Assumption~\ref{instabhyp} is, in particular, satisfied if \begin{equation}		v^{-p} \lesssim \frac{r_+^2}{2}\int_{v}^{+\infty} |\partial_v \phi|^2_{|\mathcal{H}^+}(v')dv',\end{equation}
and 
$\partial_v[ Q^2_{|\mathcal{H}^+}](v) \leq 0$ (or, more generally, if $Q^2_{|\mathcal{H}^+}(v) \geq~ e^2$ for all $v$), a condition which is  equivalent to  \begin{equation}\label{osc}
q_0 e\ \Im(\bar{\phi}_{|\mathcal{H}^+}\partial_v \phi_{|\mathcal{H}^+})\geq 0.
\end{equation} Note that \eqref{osc} is indeed satisfied if $\phi$ obeys the asymptotics of Conjecture~\ref{Priceconj}.
\end{rmk}

\section{Proof of the stability Theorem \ref{Stabilitytheorem}} \label{proofstab}

We recall that we write $A \lesssim B $ if there exists a constant \footnote{This is equivalent to saying that $\tilde{C}$ will depend only on $q_0$,$m^2$, $v_0$, the initial data and on $(e,M)$ as defined in section \ref{convergence}. } $\tilde{C}=\tilde{C}(C,Q_+,q_0,m^2,r_{\infty},s,v_0)$ such that $A \leq \tilde{C} B $. 

If we need to specify this constant, we shall call it consistently $\tilde{C}$ when they are no ambiguities.

We denote also $ A \sim B$ if $A \lesssim B$ and $B \lesssim A$. \\

When we write ``with respect to the parameters'', we actually mean ``with respect to $C$,$Q_+$,$q_0$,$m^2$,$r_{\infty}$ and $s$''. \\

We shall use repetitively the following technique : if we are in a region where $|u| \leq D v$ where $D$ is a constant, then we can take $|u_s|$ large enough (equivalently $U_s$ small enough) so that for any $u \leq u_s$ and any function of $v$, $\epsilon(u,.)= o(1)$ where  $v \rightarrow +\infty$ and any positive number $\eta$ then $|\epsilon(u,v)| \leq \eta$ for all $|u| \leq D v$. When we do so, we write ``for $|u_s|$ large enough'' or equivalently in $(U,v)$ coordinates 
``for $U_s$ small enough''.

\subsection{Strategy of the proof}

The main idea of the proof is to split the space-time into smaller regions where the red-shift and blue-shift effect manifest themselves as already done in \cite{Mihalis1} and \cite{JonathanStab} and to integrate along the characteristic for the wave equations.

The main novelty is to deal with a non-integrable field decaying on $\mathcal{H}^{+}$ like $v^{-s}$ with $ s>\frac{1}{2}$ only. The reason why stability estimates still proceed is that the Raychaudhuri equation on $\mathcal{H}^{+}$ involve the square of the field of the order $v^{-2s}$ which is integrable. \\

We will use five different regions : \begin{enumerate}

\item The event horizon $\mathcal{H}^+:= \{U=0, v\geq v_0\}$ where we use crucially the Raychaudhuri equation and exhibit the right Reissner-Nordstr\"{o}m space-time to which our dynamical space-time is expected to converge at infinity. We find that $\Omega^2$ behaves likes $e^{2K_+ \cdot (u+v+h(v))}= 2K_+ U  e^{2K_+ \cdot (v+h(v))}$ where $h(v)=o(v)$.

\item The red-shift region $\mathcal{R} =  \{ u+v+h(v) \leq -\Delta\}$: this is a large region where $\Omega^2$ is small enough and $|D_u \phi| \lesssim \Omega^2 v^{-s}$. This strong stability feature is the key to prove the estimates. Another important feature is that $\Omega^2$ can almost be written as a product $f(u) \cdot g(v)$ which simplifies most of the calculations. This comes from the fact that $\Omega^2_H(U,v)$ is almost $\Omega^2_H(0,v)$, up to a arbitrary small constant $e^{-\tilde{C} \Delta}$.

\item The no-shift region $\mathcal{N}:= \{ -\Delta \leq u+v+h(v) \leq \Delta_N \}$ : the function of this small region is to allow $r$ to vary  from its event horizon limit value $r_+$ to its Cauchy horizon limit value $r_-$, up to arbitrarily small constants. The smallness of the region allows us to conserve the estimates of its past region $\mathcal{R}$ while initiating the blue-shift effect in its future. 

\item The early blue-shift transition region \footnote{The idea to have a curve at a logarithmic distance from the no-shift region comes back -in a different form- to the early papers of Dafermos \cite{MihalisPHD}, \cite{Mihalis1}. } $\mathcal{EB}:=  \{\Delta_N \leq  u+v+h(v) \leq -\Delta'+
\frac{2s}{2|K_-|} \log(v)\}$ : this small region is the first where the blue-shift happens and as a consequence the metric coefficients $\Omega^2(u,v)$ start to be small enough to facilitate the decay of propagating waves but do not decay too much so that we can still treat the problem as almost linear : in particular \footnote{Recall that $\kappa$ and $\iota$ were defined in \eqref{kappa} and \eqref{iota}.} $\kappa^{-1}$ and $\iota^{-1}$ stay bounded. 

\item The late blue-shift region $\mathcal{LB}:=  \{-\Delta'+
\frac{2s}{2|K_-|} \log(v) \leq  u+v+h(v) \}$ : this very large region exhibits the strongest blue-shift : the metric coefficients $\Omega^2(u,v)$ start from inverse polynomial decay but decrease  exponentially in $v$ near the Cauchy horizon. We use this smallness to prove decay for the propagation problem. However, we do not prove enough decay to get a continuous extension of the space-time in the case $s \leq 1$.
\end{enumerate}

The core of the proof is to control $\partial_v \log(\Omega^2)$ and $\partial_u \log(\Omega^2)$ and use Lemma \ref{calculuslemma} :  \\

In $\mathcal{H}^+$ and $\mathcal{R}$, as a consequence of the red-shift effect, they are lower bounded by a strictly positive constant, which allows us to consider $\Omega^2$ as an increasing exponential in $u$ and as an increasing exponential in $v$,  avoiding the loss of one power when we integrate a polynomial decay. \\

In $\mathcal{N}$, $\partial_v \log(\Omega^2)$ and $\partial_u \log(\Omega^2)$ change sign and can be close to $0$, but it does not matter for the decay of the scalar field because the region is small enough \footnote{More precisely the $u$ difference is bounded.}. \\

In $\mathcal{EB}$ and  $\mathcal{LB}$, as a consequence of the blue-shift effect, they are upper bounded \footnote{Strictly speaking, we do not prove however that  $\partial_u \log(\Omega^2)$ is upper bounded in $\mathcal{LB}$ if $s\leq1$.} by a strictly negative constant, which allows us to consider $\Omega^2$ as a decreasing exponential in $u$ and as a decreasing exponential in $v$, which also avoids the loss of power when we integrate a polynomial decay.

\subsection{A calculus lemma}

We begin this proof section by a calculus lemma, which broadly says that integrating a polynomial decay -as expected for $\phi$- with a $\Omega^2$ or $\Omega^{-2}$ weight avoids to lose one power as we would otherwise. 

\begin{lem} \label{calculuslemma}

Let $q \geq 0$, $a=a(e,M,q_0,m^2,s)>0$ and $\gamma_1$ be a one-dimensional curve on which $|u| \approx v$ with $u_1(v)$ being the only $u$ such that $ (u,v) \in \gamma_1$ and $v_1(u)$ being the only $v$ such that $ (u,v) \in \gamma_1$.

Then for any positive $C^1$ function $\Omega^2$, the following hold true :

\begin{enumerate}

\item \textbf{Red-shift bounds in $|u|$} : assume that for all $ u' \in [u_1(v),u]$, $\partial_u \log(\Omega^2)(u',v) >a$. Then : 

\begin{equation*}
	\int_{u_1(v)}^u \Omega^2(u',v) |u'|^{-q} du' \lesssim \Omega^2(u,v) |u|^{-q},
\end{equation*}

\begin{equation*}
	\int_{u_1(v)}^u \Omega^{-2}(u',v) |u'|^{-q} du' \lesssim \Omega^{-2}(u_1(v),v)v^{-q}.
\end{equation*}
\item \textbf{Red-shift bounds in v} : assume that for all $ v' \in [v_1(u),v]$, $\partial_v \log(\Omega^2)(u',v) >a$. Then : 

\begin{equation*}
	\int_{v_1(u)}^v\Omega^2(u,v') v'^{-q} dv' \lesssim \Omega^2(u,v) v^{-q}.
\end{equation*}

\begin{equation*}
	\int_{v_1(u)}^v \Omega^{-2}(u,v') v'^{-q} dv' \lesssim \Omega^{-2}(u,v_1(u))|u|^{-q}.
\end{equation*}

\item \textbf{Blue-shift bounds in $|u|$} : assume that for all $ u' \in [u_1(v),u]$, $\partial_u \log(\Omega^2)(u',v) <-a$. Then : 

\begin{equation*}
	\int_{u_1(v)}^u \Omega^2(u',v) |u'|^{-q} du' \lesssim \Omega^{2}(u_1(v),v)v^{-q}.
\end{equation*}

\begin{equation*}
	\int_{u_1(v)}^u \Omega^{-2}(u',v) |u'|^{-q} du' \lesssim \Omega^{-2}(u,v) |u|^{-q}.
\end{equation*}
\item \textbf{Blue-shift bounds in v} : assume that for all $ v' \in [v_1(u),v]$, $\partial_v \log(\Omega^2)(u',v) < -a$. Then : 

\begin{equation*}
	\int_{v_1(u)}^v\Omega^2(u,v') v'^{-q} dv' \lesssim \Omega^2(u,v_1(u)) |u|^{-q},
\end{equation*}

\begin{equation*}
	\int_{v_1(u)}^v \Omega^{-2}(u,v') v'^{-q} dv' \lesssim \Omega^{-2}(u,v)v^{-q}.
\end{equation*}

\end{enumerate}

\end{lem}

\begin{proof}

We will only prove one case when $\partial_u \log(\Omega^2)>a$, the others being similar.  For $u \geq u_1(v)$ : 

$$ \int_{u_1(v)}^u \Omega^{-2}(u',v)|u'|^{-q} du' \leq \frac{1}{a} \int_{u_1(v)}^u \Omega^{-2}(u',v) \partial_u \log(\Omega^2)(u',v)|u'|^{-q} du'=-\frac{1}{a} \int_{u_1(v)}^u \partial_u (\Omega^{-2})(u',v)|u'|^{-q} du'.$$

Then we integrate by parts to write : 

$$ \int_{u_1(v)}^u \Omega^{-2}(u',v)|u'|^{-q} du'  \leq  \frac{q}{a} \int_{u_1(v)}^u  \Omega^{-2}(u',v) |u'|^{-q-1} du' + \frac{1}{a}  \Omega^{-2}(u_1(v),v) |u_1(v)|^{-q} - \frac{1}{a}  \Omega^{-2}(u,v) |u|^{-q}.$$

Then clearly $\int_{u_1(v)}^u  \Omega^{-2}(u',v) |u'|^{-q-1} du' = o(\int_{u_1(v)}^u  \Omega^{-2}(u',v) |u'|^{-q} du')$ so the dominant term is the second, and $a$ depends on the parameters only, giving : 

$$ \int_{u_1(v)}^u \Omega^{-2}(u',v)|u'|^{-q} du'  \lesssim  \Omega^{-2}(u_1(v),v) |u_1(v)|^{-q} .$$

\end{proof}

\subsection{The event horizon }

\subsubsection{Convergence at infinity towards a Reissner-Nordstr\"{o}m background} \label{convergence}

\begin{prop}
There exists constants $0<|e|<M$ such that on the event horizon $\mathcal{H}^+ = \{U=0, v \geq v_0 \}$ 

\begin{equation} \label{massevent}
| \varpi(0,v) - M | \lesssim v^{1-2s},
\end{equation}
\begin{equation} \label{chargeevent}
| Q(0,v) - e | \lesssim v^{1-2s}.
\end{equation}

And moreover $r_\infty= r_+(M,e)$ where $r_\infty$ is as in hypothesis \ref{radiusevent} and 
$$	 K(0,v)  \rightarrow K_+(M,e)>0, $$ as $v \rightarrow +\infty$.

\end{prop}

\begin{proof}

First we use \eqref{ChargeVEinstein} together with the decay of assumption \ref{fieldevent} and the boundedness of $r$ to get the existence of $ e\in \mathbb{R}$ such that \eqref{chargeevent} holds. In particular $Q$ is bounded. Moreover, due to assumption \ref{hypchargeevent}, $e \neq 0$.

For the mass, notice that by integration by parts and the decay of assumption \ref{fieldevent} : 

$$| \int_{v}^{+\infty} r \partial_v r |\phi|^2 dv'| = |-\int_{v}^{+\infty} r^2 \Re (\bar{\phi}  \partial_v \phi) dv'-\frac{r^2}{2} |\phi|^2(0,v)| \lesssim v^{1-2s}.$$

Therefore - the other terms being easier in \eqref{massVEinstein}-  by using gauge \eqref{gauge2} and assumption  \ref{fieldevent}, together with the boundedness of $r$, we prove that there exists $M \in \mathbb{R}$ such that \eqref{massevent} holds.

Gauge \eqref{gauge2} then gives the following convergence when $v$ tends to $+\infty$ on $\mathcal{H}^+$ : 

$$ \partial_v r =  1- \mu =  1- \frac{2 \varpi}{r}+ \frac{Q^2}{r^2}  \rightarrow  1- \frac{2 M}{r_\infty}+ \frac{e^2}{r_\infty^2}:= l .$$

Since $r$ admits a limit at infinity, $l=0$ so $r_\infty$ is a strictly positive root of the polynomial $x^2-2Mx+e^2$ hence : 

$$ r_\infty = M \pm \sqrt{M^2-e^2} ,$$

$$ |e| \leq  |M|,$$

$$ 0<M.$$

We then use assumption \ref{hypchargeevent} to rule out the case  $r_\infty = M  -\sqrt{M^2-e^2}$ since $r_-(M,e) \leq |e|$ for all  $0<|e| \leq  M$

Assumption \ref{hypchargeevent} also gives the sub-extremality condition $|e| < M $.


The last claim follows from the definition of $K$ and the fact that for all $0<|e| <  |M|$	, $ M - \frac{e}{r_+(M,e)}>0$.

\end{proof}

Now that $M$ and $e$ are known, we shall denote $K_+$ instead of $K_+(M,e)$ and $K_-$ instead of $K_-(M,e)$.

We know the Reissner-Nordstr\"{o}m background -indexed by $(M,e)$- towards which our space-time converges at infinity and we can define the null coordinates $u$ and $V$ in the spirit of section \ref{RNsolution} - given that the $(U,v)$ coordinates are already defined by the statement of Theorem \ref{Stabilitytheorem} - : 

\begin{defn}
Recalling that $(U,v) \in [0, U_{s}] \times  [v_0, +\infty]$ , we define $u \in [-\infty, u_s]$ by the relation : 
$$U:=\frac{1}{2K_{+}}e^{2K_{+}u} ,$$

and $V \in [V_0,1]$ by : 
$$V:=1-\frac{1}{2|K_{-}|}e^{2K_{-}v}.$$ 	\end{defn}

We write the metric \footnote{C.f section \ref{Coordinates} for a definition.} on $\mathcal{Q}$ in these different coordinates systems as : 

$$  g_{\mathcal{Q}} = - \frac{\Omega^{2}}{2} (du  \otimes dv+dv  \otimes du) = - \frac{\Omega^{2}_{H}}{2} (dU  \otimes dv+dv  \otimes dU) = - \frac{\Omega^{2}_{CH}}{2} (du  \otimes dV+dV  \otimes du).$$ 

Notice that : 

$$ 2K_+ U \Omega_H^2(U,v) = \Omega^2(u,v) =  2|K_-| (1-V) \Omega_{CH}^2(u,V).$$

We will also define $\nu_{H}:= \partial_U r$.
Notice that $\nu_H <0$ everywhere on the space-time. This is because it is strictly negative on $\mathcal{H}^{+}$ -due to the no anti-trapped surface assumption- therefore so is $\frac{\nu_H}{\Omega^2_H}$ and this quantity is decreasing in $U$ due to \eqref{RaychU}. \\

Now that the parameters $(M,e)$ are determined, we  translate the notation $\lesssim $ :   $A \lesssim B $ means that there exists a constant  $\tilde{C}=\tilde{C}(C,e,q_0,m^2,M,s,v_0)$ such that $A \leq \tilde{C} B $.

\subsubsection{Reduction to the case where $K$ is lower bounded on the event horizon.} \label{reduction}

In order to use the red-shift effect in all its strength near the event horizon, we have to prove that $K$ is close enough to its limit value -the surface gravity $K_+$- and in particular is lower bounded by a strictly positive constant on the event horizon.

To do so, we need to be far away in the future, i.e to consider large $v$.

We are going to prove that for $v_0'=v_0'(C,e,M,q_0, m^2,s)$ large enough -with the assumptions of Theorem \ref{Stabilitytheorem}-   bounds of the following form are still true :

$$ |D_U \phi(U,v_0')| \lesssim D(v_0') ,$$

$$ |\partial_U r(U,v_0') | \gtrsim 1.$$

In the second step, we restart our problem, replacing $v_0$ by $v_0'$ in the hypothesis of Theorem \ref{Stabilitytheorem} - in particular $C_{in}$ is redefined to be $C_{in} = \{ v \equiv v_0', 0 \leq U \leq U_s \}$ and \eqref{gaugeAponctuelle}, \eqref{gauge1} are true on $ v \equiv v_0'$ instead.

This is can be done introducing a new coordinate system $(U',v)$ with $\partial_{U'} r(U',v_0')  = -1$. This can only multiply the bound for $D_{U'} \phi(U',v_0') $ by a constant. Notice that $|D_{U'} \phi(U,v_0')|$ is not modified by any gauge transform on $A$. After this section, we will abuse notation and still call $(U,v)$ this new coordinate system $(U',v)$.

We now take $v_0'=v_0'(C,e,M,q_0, m^2,s)$ to be large enough so that $2K-2K_+ + r m^2 |\phi|^2$ is arbitrarily close to 0.

To be able to do it, we must use \footnote{This essentially boils down to an easy local existence theorem.} the Einstein-Maxwell-Klein-Gordon equations on the space-time rectangle $[0, U_s] \times [v_0, v_0']$ which is the object of the following lemma : 

\begin{lem} \label{reductionlemma}
Under the same hypothesis than before and for $v_0' > v_0$, if $U_s$ is sufficiently small there exists a constant $D>0$ depending on $C,e,M,q_0, m^2,s,v_0$ and $v_0'$ such that 

\begin{equation} 
|\partial_U r(U,v_0') |^{-1}+|D_U \phi(U,v_0')| \leq D .
\end{equation}

Therefore, for any $\eta>0$ independent \footnote{We insist that $\eta$ must be a numerical constant that do \textbf{not} depend on \textbf{any} of the $C,e,M,q_0, m^2,v_0$ or $v_0'$.} of any parameter, there exists a $v_0'>0$ such that 

$$|D_U \phi(U,v_0')| \lesssim C ,$$

and for all $v \geq v_0'$ : 

$$ |2K(0,v)-2K_+| \leq \eta K_+ ,$$
$$ r m^2 |\phi|^2(0,v) \leq \eta K_+.$$
\end{lem}

The proof, which is not difficult, is deferred to Appendix C.

In what follows, we will not refer to $v_0'$ any longer, and when we will write $v_0$ in the rest of the paper, we actually mean $v_0'$.

\subsubsection{Main bounds on the event horizon} \label{boundsEH}

\begin{prop} The following bounds hold on the event horizon :

\begin{equation} \label{lambdaevent}
0 \leq \lambda = 1- \mu \lesssim v^{-2s} ,    
\end{equation} 		
\begin{equation} \label{radiuseventprop}
0 \leq r_+ -r(0,v)  \lesssim v^{1-2s}  ,   
\end{equation}

\begin{equation}\label{partialvomegaevent}
|	\partial_v \log(\Omega^2_{H})(0,v)- 2K(0,v) | \lesssim v^{-2s}
,	\end{equation}	
\begin{equation}\label{partialuomegaeventprop}
|	\partial_U \log(\Omega^2_{H})|(0,v)  \lesssim \Omega^2_{H}(0,v),
\end{equation}

\begin{equation}\label{fieldeventUhorizon}
|	\partial_U \phi|(0,v)  \lesssim \Omega^2_{H}(0,v) v^{-s}
\end{equation}		
\begin{equation}\label{potentialeventhorizon}
|	A_U|(0,v)  \lesssim \Omega^2_{H}(0,v).
\end{equation}

Moreover there exists a fixed function $h(v)$  such that : 

\begin{equation} \label{hdef}
\Omega^2_{H}(0,v) = -4 \nu_{H} (0,v)=e^{2K_{+}(v+h(v))},
\end{equation} 
with

\begin{equation} \label{hderiv}
|\partial_v h(v)| \lesssim  v^{1-2s}.
\end{equation}

\end{prop}

\begin{proof}

We use \eqref{Radius2} and gauge \eqref{gauge2} to write : 

\begin{equation} \label{Event}
\partial_v \log(\Omega^2_{H})=\partial_v \log(- \nu_{H}) =  2K- rm^2 |\phi|^2.
\end{equation} 

\eqref{partialvomegaevent} then follows directly from assumption \ref{fieldevent}.

We first prove that $$\frac{\lambda}{\Omega^2_H}(v=+\infty) =0.$$

Let $0<\delta_0<1$ suitably small enough to be chosen later, independently of all the parameters. 

Then, by section \ref{reduction} we are allowed to assume that :

$$ |2K-rm^2 |\phi|^2-2K_+| \leq  2\delta_0 K_+.$$

Then, we integrate \eqref{Event} on $[v_0,v]$ to get : 

$$ e^{ 2K_+(1 -\delta_0)v} \lesssim \Omega^2_H(0,v) \lesssim  e^{ 2K_+(1 +\delta_0)v} .$$

Using \eqref{RaychV} written as $\partial_v  ( \frac{\lambda}{\Omega^2_H})= \frac{-r}{\Omega^2_H}|\partial_v \phi |^2$, we get that 

$$ |\partial_v (\frac{\lambda}{\Omega^2_H})| \lesssim e^{ -2K_+ (1-\delta_0)v}v^{-2s},$$

which is integrable. Therefore $\frac{\lambda}{\Omega^2_H}$ admits a limit $l \in \mathbb{R}$ when $v \rightarrow +\infty$. Integrating \footnote{Recall that $\int_{v}^{+\infty}e^{ -2K_+ (1-\delta_0)v'}v'^{-2s}dv'\lesssim e^{ -2K_+ (1-\delta_0)v}v^{-2s}$. Similarly, $\int_{v_0}^{v}e^{  4K_+ \delta_0v'}v'^{-2s}dv'\lesssim e^{ 4K_+ \delta_0v}v^{-2s}$. } on $[v,+\infty]$, we get after multiplication by $\Omega^2_H(0,v)$  : 

$$ |\lambda - l \Omega^2_H | \lesssim e^{ 4K_+\delta_0 v}v^{-2s}.$$

Integrating again and using the boundedness of $r$, we get after absorbing the $r$ difference in $e^{ 4K_+\delta_0v}v^{-2s}$ 

$$ |l  \int_{v_0}^{v} \Omega^2_H|  \lesssim e^{  4K_+\delta_0v}v^{-2s}.$$

Hence, using the lower bound for $ \Omega^2_H$ : 

$$ |l|  e^{ 2K_+(1-\delta_0) v} \lesssim  e^{4K_+\delta_0 v}v^{-2s}.$$

If $ \delta_0 < \frac{1}{3}$ , it  proves that $l=0$. Since $\partial_v (\frac{\lambda}{\Omega^2_H}) \leq 0$, we have that $$ \lambda \geq 0. $$



Using \eqref{Event} and the earlier section \ref{reduction}, we are allowed to assume that : 
$$ \partial_v \log(\Omega^2_{H}) \geq K_+ >0. $$

Therefore using a variant of Lemma \ref{calculuslemma} on $[v,+\infty]$ :

$$0 \leq  \lambda(0,v) = \Omega^2_H(0,v) \int_{v}^{+\infty} \frac{r |\partial_v \phi|^2}{\Omega^2_H(0,v')}dv' \lesssim  \Omega^2_H(0,v) \int_{v}^{+\infty} \frac{ |v'|^{-2s}}{\Omega^2_H(0,v')}dv'\lesssim v^{-2s}.$$

Therefore we proved \eqref{lambdaevent} and \eqref{radiuseventprop}. It also gives -using \eqref{massevent} and \eqref{chargeevent}- :  

$$ | 2K(U,v) -2K_+(M,e) | \lesssim v^{1-2s}, $$

and therefore giving \eqref{hderiv} from \eqref{partialvomegaevent}.

\eqref{potentialeventhorizon} follows from \eqref{gaugeAponctuelle} and \eqref{potentialEinstein} written as $\partial_v A_U = -\frac{Q \Omega^2_H(0,v)}{2r^2}$ ,using Lemma \ref{calculuslemma} with $q=0$.

From then it is easy to use \eqref{Field3}, the gauge \eqref{gauge2} and the decay of $\phi$ and  $\partial_v \phi$ to establish \eqref{fieldeventUhorizon}.

Now writing \eqref{Omega} as 

$$
|\partial_{v} \partial_{U} \log(\Omega^2_{H})|=|-2 \Re (D_{U} \phi \partial_{v}\phi)+\frac{ \Omega^{2}_H}{2r^{2}}+\frac{2\partial_{U}r\partial_{v}r}{r^{2}}- \frac{ \Omega^{2}_H}{r^{4}} Q^2| \lesssim \Omega^2_H(0,v)
$$                     

gives immediately \eqref{partialuomegaeventprop} after integration.

\end{proof}

\subsection{The red-shift region} \label{redshift}

We define for $  \delta >0 $ suitably small to be chosen later, the red-shift region as : 

$$\mathcal{R}:= \{ U\Omega^2_{H}(0,v) \leq \delta\} = \{ u+v+h(v) \leq \frac{ \log(2K_+\delta)}{2K_+}:=-\Delta\}.$$

In this region, we expect that $\Omega^2$ will be exponentially growing in $u+v$ while still remaining very small as it is the case for Reissner-Nordstr\"{o}m , which is a manifestation of the red-shift effect. 

However already on the event horizon $\frac{\Omega^2_H(0,v)}{e^{2K_+v}}$ may be unbounded \footnote{This quantity may grow like $v^{2-2s}$. If $s>1$ like in \cite{JonathanStab}, this problem does not exist so $\mathcal{R}$ can be defined using $e^{2K_+ \cdot (
u+v)}$ directly .} so we decide to set 

$e^{2K_+ \cdot (
u+v+h(v))} = 2K_+ U \Omega^2_H(0,v)$ to be small instead of $e^{2K_+ \cdot (
u+v)}$.

The most emblematic consequence of the red-shift effect - and the main difficulty- is the bound for the field $|D_u \phi| \lesssim \Omega^2 v^{-s}$ from which we derive the others.

\subsubsection{Main bounds on the red-shift region}

\begin{prop} \label{boundsRS}

We have the following control \footnote{Note that \eqref{phivRS}, \eqref{phiURS} and \eqref{ARS} also give $|D_U \phi(U,v)| \lesssim U \Omega^2_H(0,v).$ } on the field and the potential on $\mathcal{R}$ : 
\begin{equation} \label{phivRS}
|\phi|+|\partial_v \phi| \lesssim v^{-s},
\end{equation}			\begin{equation} \label{phiURS}
|\partial_U \phi| \lesssim \Omega^2_H(0,v) v^{-s},
\end{equation}
\begin{equation} \label{ARS}
|A_U | \lesssim \Omega^2_H(0,v).
\end{equation}
\\

We also have : 
\begin{equation} \label{OmegaPropRedshift}
| \log(\Omega^2(u,v))- 2K_+ \cdot (u+v+h(v))|	=		 |\log(\frac{\Omega^2_H(U,v)}{\Omega^2_H(0,v)})|\lesssim U\Omega^2_H(0,v),
\end{equation}
\begin{equation} \label{kappaRedshiftprop}
0 \leq 1-\kappa(U,v) \lesssim \Omega^2(U,v) v^{-2s},
\end{equation}

\begin{equation} \label{partialuRSOmegaprop}
|\partial_U \log(\Omega^2_H)(U,v) |  \lesssim \Omega^2_H(0,v), 
\end{equation}
\begin{equation} \label{partialvRSOmegaprop}
|\partial_v \log(\Omega^2)(U,v)-2K(U,v) |  \lesssim v^{-2s},
\end{equation}

\begin{equation} \label{rRedShiftprop}0 \leq r_+-r(U,v) \lesssim \Omega^2+ v^{1-2s} ,
\end{equation}
\begin{equation} \label{QRedShiftprop} |Q(U,v)-e| \lesssim  v^{1-2s} ,
\end{equation}
\begin{equation} \label{MRedShiftprop} |\varpi(U,v)-M| \lesssim  v^{1-2s} ,
\end{equation}
\begin{equation} \label{KRedShiftprop} |2K(U,v)-2K_+| \lesssim  \Omega^2+   v^{1-2s}.
\end{equation}

\end{prop}

\begin{proof}
We bootstrap \footnote{For an introduction to bootstrap methods, c.f chapter 1 of \cite{Tao}.} the following estimates \footnote{Notice that bootstrap \eqref{B3} and \eqref{B4} combined give $\Omega^2_H(U,v) \leq 4 \Omega^2_H(0,v).$} in $\mathcal{R}$ : 
\begin{equation} \label{B2}
|\phi|+|\partial_v \phi| \leq 4C v^{-s},
\end{equation}
\begin{equation} \label{Bnew}
|D_U \phi|	 \leq D \Omega^2_H(0,v) v^{-s},
\end{equation}
\begin{equation} \label{B3}
-\nu_{H}(U,v) \leq \Omega^2_{H}(0,v),
\end{equation}
\begin{equation} \label{B4}
\frac{1}{2} \leq \kappa \leq 1,
\end{equation}
\begin{equation} \label{B5}
|Q-e| \leq 4 \bar{C}  v^{1-2s}.
\end{equation}

Where $\bar{C}$ is the constant of estimate \eqref{chargeevent} and $D$ is a large enough constant -independent of $\delta$- to be chosen later. Recall also that $C$ is defined in the statement of Theorem \ref{Stabilitytheorem}.

We can first write \eqref{massUEinstein} using 
bootstraps \eqref{B2}, \eqref{Bnew}, \eqref{B3}, \eqref{B4}   as : 

$$ |\partial_U \varpi | \lesssim (  D^2 |\lambda| + 1) \Omega^2_H(0,v) v^{-2s}.$$

Using \eqref{Radius}, it is not difficult to prove that $|\lambda|$ is bounded hence after integrating in $U$ :

\begin{equation} \label{massnew}
| \varpi(U,v) - \varpi(0,v)| \lesssim D^2 \delta v^{-2s}.
\end{equation}

Then it gives \eqref{MRedShiftprop} , using the bound on the event horizon with $\delta$ small enough with respect to $D$ notably.

Similarly we get :
\begin{equation} \label{chargenew}
| Q(U,v) -Q(0,v)| \lesssim D \delta v^{-2s},
\end{equation}		
which proves \eqref{QRedShiftprop} and closes bootstrap \eqref{B5} for $\delta$ small enough.  \\

We now write \eqref{potentialEinstein} as : 

$$ \partial_v A_U = -\frac{2Q\kappa }{r} \nu_H (U,v).$$

Then bootstraps \eqref{B3}, \eqref{B4} and \eqref{B5} give 

$$ |\partial_v A_U| \lesssim  \Omega^2_H (0,v).$$

Hence with gauge \eqref{gaugeAponctuelle} and the bound on the event horizon \eqref{hdef}, \eqref{hderiv}, we use Lemma \ref{calculuslemma} with $q=0$ to get \eqref{ARS}  :

$$	|A_U | \lesssim \Omega^2_H(0,v).$$

Now using the last equation we get with bootstrap \eqref{B2} and \eqref{Bnew}	: 
$$ |\partial_U \phi | \lesssim D \Omega^2_H(0,v) v^{-s}.$$
We can then  integrate  to get : 

\begin{equation} \label{phinew}
| \phi(U,v)-\phi(0,v)| \lesssim D \delta v^{-s}.
\end{equation}

which implies that for $\delta$ small enough : 

$$ |\phi| \leq 2 C  v^{-s}.$$
\\

Let $ 0<a$ be a constant suitably chosen later. We can rewrite \eqref{Field} together with \eqref{Radius} as : 

\begin{equation} \label{Fieldshift}
\partial_v ( e^{av} r\frac{D_U \phi}{\nu_H}) =  \left( a- \kappa( 2K- rm^2 |\phi|^2) \right) e^{av} r\frac{D_U \phi}{\nu_H} - e^{av}\partial_v \phi +\kappa e^{av} r m^2 \phi+ i q_0 \kappa e^{av}\frac{Q}{r}\phi.
\end{equation}

We first need to prove that $K$ is lower bounded in $\mathcal{R}$. The bootstrap \eqref{B3} gives : 

$$ 0 \leq r(0,v) -r(U,v) \leq \delta .$$



Then, making use of \eqref{massnew} and \eqref{chargenew}, we write : 

$$ |K(U,v)-K_+| \leq |K(U,v)-K(0,v)|+|K(0,v)-K_+| \lesssim (1+D+D^2) \delta + |K(0,v)-K_+|. $$

We then recall that the discussion of section \ref{reduction} allows us to consider that $ |K(0,v)-K_+| \leq \eta K_+$ and also that $r m^2|\phi|^2(0,v) < \eta K_+$ for any $\eta$ not depending on the parameters. Hence for $\delta$ small enough, we can assume that 

$$ 2K(U,v)-rm^2 |\phi|^2(U,v)> K_+.$$

Choosing say $0<a < \frac{K_{+}}{4}$ gives with  bootstrap \eqref{B4} that 
$a- \kappa( 2K-rm^2 |\phi|^2) \leq -\frac{K_{+}}{4}$

We then use the Gr\"{o}nwall Lemma combined with the boundedness of  bootstrap \eqref{B4}, the lower boundedness of $r$, the decay of bootstrap \eqref{B2}  and assumption \eqref{fieldv0} with gauge \eqref{gauge1} for the initial condition to get :  

$$
| r\frac{ D_U \phi }{\nu_{H}}| \lesssim v^{-s}+e^{a(v_0-v)}\lesssim v^{-s}.$$

It also closes\footnote{We used that $r$ is bounded below by a constant depending of $v_0$ and the parameters for $\delta$ small enough.} bootstrap \eqref{Bnew} if $D$ is large enough compared \footnote{In particular, $D$ is taken large enough independently of $\delta$, hence taking $\delta$ small enough compared to $D$ was licit and boiled down to taking $\delta$ small enough compared to the parameters.} to the constant that arises which depends on $C,e,M,q_0,m^2,s,v_0$ only and proves : 

\begin{equation} \label{Duphi}
|D_U \phi |(U,v)\lesssim \Omega^2_H(0,v) v^{-s}.
\end{equation}
\begin{equation} 
|\partial_U \varpi|(U,v)+ 	|\partial_U Q^2|(U,v)\lesssim \Omega^2_H(0,v) v^{-2s}.
\end{equation} 

Using the preceding bounds on $\phi$ and $A_U$, we get \eqref{phiURS} :

$$
| \partial_U \phi | \lesssim |\nu_H | v^{-s}.$$

Now using \eqref{ARS}, we can write \eqref{Field} as :  
$$ |\partial_U (\partial_v \phi) | \lesssim |\partial_U \phi| + \Omega^2_H( |\phi|+ |\partial_v \phi|)  \lesssim -\nu_H(U,v) v^{-s}.$$

Hence by \eqref{B3}, bootstrap \eqref{B2} is validated for $\delta$ small enough.

Recall from section \ref{reduction} that we established that everywhere on the space-time : 
$$ 0 \leq \kappa \leq 1 .$$

Writing \eqref{RaychU} in (U,v) coordinates, we get -using \eqref{Duphi}-  :

$$
|\partial_U  \log(\kappa)| = \frac{r}{-\nu_H}|D_{U}\phi|^{2} \lesssim |\nu_{H}| v^{-2s}.$$

Using bootstrap \eqref{B3} we get the amelioration :

\begin{equation} \label{kappaRedshift}
0 \leq 1-\kappa \lesssim U \Omega^2_H(0,v) v^{-2s}.
\end{equation}

Hence bootstrap \eqref{B4} is validated for $\delta$ small enough. \\

Now we write \eqref{Omega} as : 

$$ |\partial_v \partial_U \log(\Omega^2_H) | \lesssim |\nu_H| \leq \Omega^2_{H}(0,v) = e^{2K_+( v + h(v))}.$$

Hence we establish \eqref{partialuRSOmegaprop} using Lemma \ref{calculuslemma} and \eqref{hderiv}  : 

$$ |\partial_U \log(\Omega^2_H)(U,v) | \lesssim |\partial_U \log(\Omega^2_H) |(U,v_0)+ \int_{v_0}^{v} e^{2K_+( v' + h(v'))}dv' \lesssim \frac{1}{K_+} e^{2K_+ (v + h(v))} \lesssim \Omega^2_H(0,v) ,$$

where we used that on $C_{v_0}$ and due to \eqref{RaychU} ,  assumption  \ref{fieldv0} and gauge \eqref{gauge1} :   $$|\partial_U\log(\Omega^2_H) |(U,v_0)  =r |D_U \phi|^2(U,v_0) \lesssim 1 .$$  

Hence we establish \eqref{OmegaPropRedshift}, that we write with a constant $\tilde{C}>0$ as : 

$$ \label{OmegaRedshift}
e^{-\tilde{C}U\Omega^2_H(0,v)} \leq \frac{\Omega^2_H(U,v)}{\Omega^2_H(0,v)}\leq  e^{\tilde{C}U\Omega^2_H(0,v)},
$$

and in particular : 

$$
e^{-\tilde{C}\delta} \leq \frac{\Omega^2_H(U,v)}{\Omega^2_H(0,v)}\leq  e^{\tilde{C}\delta},
$$

which together with \eqref{kappaRedshift} closes bootstrap \eqref{B3} for $\delta$ small enough .
It  gives\footnote{Notice that $\delta$ small enough is to be understood as $\delta \leq \epsilon(C,e,M,q_0,m^2,s,v_0)$ with $\epsilon$ small enough.} also \eqref{kappaRedshiftprop}.

Moreover we have the more precise estimate : 

$$
e^{-\tilde{C}\delta} \leq \frac{-4\nu_H(U,v)}{\Omega^2_H(0,v)}\leq  	\frac{ e^{\tilde{C}\delta}}{1-\tilde{C'}\delta v^{-2s}}.
$$

We get the more refined bound \eqref{rRedShiftprop} on $r$, using \eqref{radiuseventprop} : 

$$
0 \leq r_+ -r(U,v) \leq \frac{1}{4}e^{\tilde{C}\delta}U\Omega^2_H(0,v) + \tilde{C} v^{1-2s}  \lesssim \Omega^2(U,v)+ v^{1-2s} .$$

As a consequence of \eqref{rRedShiftprop}, \eqref{QRedShiftprop} and \eqref{MRedShiftprop} we get \eqref{KRedShiftprop}. \\

Finally we can rewrite \eqref{Omega2} in $(U,v)$ coordinates and using our estimates we get : 

$$
|\partial_U (\partial_v \log(\Omega^2)-2K) | \lesssim  |\kappa-1| |\partial_U (2K)| + |D_U \phi| |\partial_v \phi| + |\partial_U \varpi|+|\partial_U Q| \lesssim \Omega^2_H(0,v) v^{-2s}.
$$

Hence with \eqref{partialvomegaevent}, we prove \eqref{partialvRSOmegaprop}.

\end{proof}

\subsubsection{Control of $\iota$ in the late red-shift transition region} \label{redtransition}

Notice that in Proposition \ref{boundsRS}, we have an estimate for $1-\kappa$ but nothing for the $v$-analogue $1-\iota$. This is because $\iota^{-1}$ blows-up in general near the event horizon where $1-\iota^{-1}(0,v)=+\infty$.

It is important to get a bound for  $1-\iota$ as it will give control of $\partial_u \log(\Omega^2)-2K$, in the same manner $1-\kappa$ bounds in $\mathcal{R}$ gave control of  $\partial_v \log(\Omega^2)-2K$. \\

Still we will show that we can control $1-\iota$ on  a subset\footnote{$C_0$ is chosen such that $C_0 v_0^{-q(s)} < \delta$.} of $\mathcal{R}$ defined as $$\mathcal{LR} := \{ C_0 v^{-q(s)} \leq  U \Omega^2_H(0,v) \leq \delta \},$$  where $q(s) = 1_{ \{s \leq 1 \}          }+s 1_{  \{s > 1 \}}$  and we call this subset the late red-shift transition region. 

The name transition simply comes from the fact we aim at bounding $\partial_u \log(\Omega^2)-2K$ instead of $\partial_u \log(\Omega^2)-2K_+=\partial_u \log(\Omega^2_H)$ so there is a transition from $2K_+$ to $2K$. \\

Notice that in this region $|u| \sim v$.

\begin{prop} In  $\mathcal{LR} := \{ C_0 v^{-q(s)} \leq  U \Omega^2_H(0,v) \leq \delta \}$, we have the following estimates  : 

\begin{equation} \label{iotalateRS}
|1- \iota^{-1}| \lesssim v^{-p(s)},
\end{equation}

\begin{equation} \label{partialuomegalateRS}
|\partial_u \log( \Omega^2) -2K| \lesssim v^{-p(s)},
\end{equation}

where  \footnote{The behaviour is different for $s>1$ but still gives integrability when $s>1$ and non-integrability if $s\leq 1$. } $p(s) = (2s-1) 1_{ \{s \leq 1 \}          }+s 1_{  \{s > 1 \}}$.
\end{prop}

\begin{proof}

Use \eqref{Radius} to write : 

$$ \partial_u (\Omega^2 + 4 \lambda)= \Omega^2 ( \partial_u \log(\Omega^2)-2K+ rm^2|\phi|^2).$$

We can integrate from the event horizon for $u' \in (-\infty, u ]$ to get : 

$$ |\Omega^2 + 4 \lambda|(u,v) \lesssim |\lambda|(-\infty,v)_{|\mathcal{H}^+}+ \int_{-\infty}^{u} \Omega^2(u',v) | r m^2 |\phi|^2+  \partial_u \log(\Omega^2)-2K| du'. $$

Notice that \eqref{partialuRSOmegaprop}- thanks to \eqref{OmegaPropRedshift}- can be alternatively written as 

$$ |\partial_u  \log(\Omega^2)(u,v) -2K_+ | \lesssim U\Omega^2_{H}(0,v)  \sim \Omega^2(u,v). $$ 

In particular if $\delta$ is chosen to be small enough, $\partial_u  \log(\Omega^2) >K_+ $.

Moreover, \eqref{phivRS} and \eqref{KRedShiftprop} give : 

$$ |\Omega^2 + 4 \lambda|\lesssim |\lambda|_{|\mathcal{H}^+}+ \int_{-\infty}^{u} \Omega^2(u',v)(\Omega^2(u',v) + v^{1-2s}) du'= |\lambda|_{|\mathcal{H}^+}+ \int_{-\infty}^{u} \Omega^4(u',v)du'+v^{1-2s}\int_{-\infty}^{u} \Omega^2(u',v) du'.$$

We then divide by $\partial_u  \log(\Omega^2)$ which is lower bounded to use Lemma \ref{calculuslemma} and with \eqref{lambdaevent} we get \footnote{Recall that $\Omega^2(u,v)= 2K_+ U\Omega^2_H(U,v)$.}

$$ |\Omega^2 + 4 \lambda|(u,v) \lesssim \Omega^2(u,v)(\Omega^2(u,v) + v^{1-2s})+ v^{-2s}.$$

Therefore -dividing by $\Omega^2$- on the past boundary of  $\mathcal{LR}$ defined as $\gamma_{\mathcal{LR}} := \{ U \Omega^2_H(0,v)= C v^{-q(s)}\}$ we get 

$$ |1 - \iota^{-1}|_{\gamma_{\mathcal{LR}}} \lesssim v^{-q(s)}+ v^{1-2s}+v^{q(s)-2s} \lesssim  v^{-p(s)}.$$

We then integrate \eqref{RaychV} from $\gamma_{\mathcal{LR}}$ i.e on $[v_{\gamma_{\mathcal{LR}}}(u),v]$, using \eqref{phivRS} : 

$$ |1 - \iota^{-1}|(u,v) \leq |1 - \iota^{-1}|(u, v_{\gamma_{\mathcal{LR}}}(u)) +\int_{v_{\gamma_{\mathcal{LR}}}(u)}^{v} \Omega^{-2}(u,v') v'^{-2s} dv'.$$

Thanks to \eqref{partialvRSOmegaprop} and for $|u_s|$ large enough,  $\partial_v \log(\Omega^2) > K_+$ hence using Lemma \ref{calculuslemma} : 

$$\int_{v_{\gamma_{\mathcal{LR}}}(u)}^{v} \Omega^{-2}(u,v') v'^{-2s} dv' \lesssim \Omega^{-2}(u,v_{\gamma_{\mathcal{LR}}}(u)) v_{\gamma_{\mathcal{LR}}}(u)^{-2s} \lesssim v^{q(s)-2s},$$

where we have used in the last inequality that in this region $v_{\gamma_{\mathcal{LR}}}(u) \sim |u| \sim v $.

Hence \eqref{iotalateRS} is proved : 

$$ |1 - \iota^{-1}| \lesssim v^{-p(s)}+v^{q(s)-2s} \lesssim v^{-p(s)}.$$

Notice that because of \eqref{phivRS} and the boundedness \footnote{Since $v \sim |u|$ in this region we can take $|u_s|$ to be large enough so that -say- $|1-\iota^{-1}| \leq 0.01$.} of $\iota^{-1}$ we have : 

$$ |\partial_v \varpi |+|\partial_v Q^2 | \lesssim v^{-2s},$$
$$ |\partial_v (2K) | \lesssim \Omega^2 +v^{-2s}.$$

Hence using $\eqref{Omega2}$ and the red-shift region main bounds we get : 

$$ | \partial_v ( \partial_u \log(\Omega^2)-2K)| \lesssim (\Omega^2+v^{-2s} )v^{-p(s)}+\Omega^{2}v^{-2s} + \iota v^{-2s} \lesssim v^{-2s}+ \Omega^2v^{-p(s)}.$$

Integrating using that $\partial_v \log(\Omega^2) > K_+$ and Lemma \ref{calculuslemma} gives \eqref{partialuomegalateRS}, after noticing that : 

$$ | \partial_u  \log(
\Omega^2)(u,v_{\gamma_{\mathcal{LR}}}(u))-2K(u,v_{\gamma_{\mathcal{LR}}}(u))| \lesssim \Omega^2(u,v_{\gamma_{\mathcal{LR}}}(u))+ |2K-2K_+|(u,v_{\gamma_{\mathcal{LR}}}(u)) \lesssim v^{-q(s)}+v^{1-2s}.$$

\end{proof}


\subsection{The no-shift region} \label{noshift}

We now define the no-shift region as : $$\mathcal{N}:= \bigcup\limits_{k=1}^{N}\mathcal{N}_k \; ,$$

where $$\mathcal{N}_k:= \{\Delta_{k-1} :=-\Delta + (k-1) \epsilon\leq  u+v+h(v) \leq \Delta_k:=-\Delta + k \epsilon \} \; ,$$

$\epsilon>0$ small enough and $N \in \mathbb{N}$ large enough are to be chosen\footnote{Later, we will first choose $\epsilon$ small compared to $C,e,M,q_0,m^2$ and $\delta$ in this section. Once $\epsilon$ is chosen and small enough, we will choose $N \epsilon $ large enough compared to $C,e,M,q_0,m$ and $\delta$ in the next section.  } later.

We take the convention that $\mathcal{N}_0= \gamma_{-\Delta}$ is the past boundary of $\mathcal{N}$. \\

This is the region where the transition between the red-shift effect and the blue-shift effect occurs : $2K$ goes from positive values for $r$ close to $r_+$ towards negative values for $r$ close to $r_-$. 

Since the derivatives of $\log(\Omega^2)$ are broadly $2K$ which changes sign hence cancels,  we cannot use the technique arising from Lemma \ref{calculuslemma} as before.

Moreover, we cannot hope for any decay of  $\Omega^2$ that is small on the past and future boundary but is only bounded in between. 

However, this region is easy because the $u+v+h(v)$ difference is finite so that essentially, we do not lose the bounds proved in the red-shift region. 

There are two difficulties : the first is to prove decay for the wave equations. We do it by splitting $\mathcal{N}$ into small enough pieces which allows us to close the bootstrapped bounds.

The second and main difficulty is to prove that the blue-shift indeed appears, i.e that $r$ is decreasing enough so that it reaches $ M -\frac{e^2}{r} <0 $ i.e $K_{M,e}(r)<0$, giving also $K<0$. \\

Note that in $\mathcal{N}$ : $|u| \sim v $, due to \eqref{hderiv} which gives $h(v) = o(v)$.

We will denote for $ 0 \leq k \leq N$ :  $\gamma_{k} := \{ u+v+h(v) = \Delta_k\}$.  We also denote $u_{k}(v)$ the unique $u$ such that $(u_{k}(v),v) \in \gamma_{k}$. We define similarly  $v_{k}(u)$.

\subsubsection{The main estimates in the no-shift region} \label{mainNS}

This is the first part where we address the propagation of the bounds established in the past sections.

Since $\Delta$ is now fixed definitively, we define the new notation : $A \precsim B $ if there exists a constant $\bar{C}=\bar{C}(\Delta)$ such that $A \lesssim \bar{C} B $. 

If we need to specify this constant, we shall call it consistently $\bar{C}$ when there are no ambiguities.

\begin{prop} \label{propannexe}

For small enough $\epsilon>0$ , we have : 

the following control on the field and the potential on $\mathcal{N}$ : 
\begin{equation} \label{phivNSprop}
|\phi|+|\partial_v \phi| \precsim 2^N v^{-s},
\end{equation}	\begin{equation} \label{phiUNSprop}
|D_u \phi| \precsim  2^N |u|^{-s} \sim 2^N v^{-s} ,
\end{equation}
\begin{equation} \label{ANSprop}
|A_u | \lesssim (N+1) \delta .
\end{equation}
\\

and we also have \footnote{Being a bit more careful, we can prove an improved version of \eqref{QNS2} and \eqref{MNS2} without the $4^N$ factor.} : 
\begin{equation} \label{OmegaPropNSprop}
|\log\Omega^2(u,v)-\log(-4(1-\frac{2M}{r}+\frac{e^2}{r^2}))|\precsim 4^N v^{1-2s} ,
\end{equation}  
\begin{equation} \label{kappaNStpropprop}
0 \leq 1-\kappa \precsim  5^N v^{-2s},
\end{equation}						\begin{equation} \label{iotaNStpropprop}
|1-\iota| \precsim 5^N  v^{-p(s)},
\end{equation}
\begin{equation} \label{partialuNSOmegaprop2}
|\partial_u \log(\Omega^2)-2K |  \precsim  5^N v^{-p(s)},
\end{equation}
\begin{equation} \label{partialvNSOmegaprop2}
|\partial_v \log(\Omega^2)-2K |  \precsim  5^N v^{-2s},
\end{equation}

\begin{equation} \label{QNS2} |Q(u,v)-e| \precsim  4^N v^{1-2s}. 
\end{equation}
\begin{equation} \label{MNS2} |\varpi(u,v)-M| \precsim  4^N v^{1-2s} .
\end{equation}

\end{prop}

The proof essentially relies on a partition of  $\mathcal{N}$ into sub-regions with small $u+v+h(v)$ difference, in the style of the methods of \cite{Mihalis1} and \cite{JonathanStab}. Since the proof does not present so many original ideas, we put it in Appendix B for the sake of completeness.

\subsubsection{Estimates on the future boundary of the no-shift region} \label{Nfuture}

We now address the second difficulty :  we need to have $K<0$ at some point to initiate the blue-shift effect,  get $\Omega^2$ small on the future boundary and therefore $r$ close to $r_{-}$. To do that, we use a simple contradiction argument.

\begin{prop}
There exists a constant $K_*>0$, independent of $N$ and $\epsilon$  such that, for $u \leq u_s$ : 

\begin{equation} \label{OmegaNSfuture} \Omega^2_{|\gamma_{N}} \precsim  e^{-K_* N \epsilon},
\end{equation}
\begin{equation} \label{rNSfuture} |r _{|\gamma_{N}}-r_-| \precsim e^{-K_* N \epsilon},
\end{equation}
\begin{equation} \label{KNSfuture} |2K_{|\gamma_{N}}-2K_-| \precsim e^{-K_* N \epsilon}.
\end{equation}
\end{prop}

\begin{proof}

We will start by the following lemma, proved by contradiction  : 

\begin{lem}
For all $\delta_*>0$, there exists $0<\Delta_{*} $ large enough so that $r < r_-(e,M)+\delta_*$ on $\gamma_{\Delta_{*}} \cap \{ u \leq u_s \}$.
\end{lem}

\begin{proof} 

By contradiction, take a $\delta_*>0$ such that for all $0<\Delta_{*} $, there exists $u \leq u_s$ such that on $\gamma_{\Delta_{*}}$, $$r(u,v_{\Delta_{*}}(u)) \geq r_- +\delta_*$$ 


Then because $\lambda, \nu <0$,  for all $u_{0}(v_{\Delta_{*}}(u)) \leq u' \leq u$ we have : 

\begin{equation} \label{monotonicity}
	r_-+ \delta_* \leq r(u',v_{\Delta_{*}}(u)) \leq r_+-\delta. 
\end{equation}

Using \eqref{OmegaPropNSprop} and \eqref{kappaNStpropprop}, we see that for $|u_s|$ large enough, there exists a constant $\bar{C}>0$ depending on $\Delta$ only such that for all $u_{0}(v_{\Delta_{*}}(u)) \leq u' \leq u$ 

$$ -\nu(u',v_{\Delta_{*}}(u)) \gtrsim \frac{\bar{C}} {\delta_{*}} 
. $$

Then we can integrate in $u'$ from $\gamma_{0}$ to  $\gamma_{\Delta_{*}}$ : 

$$ r(u_{0}(v_{\Delta_{*}}(u)),v_{\Delta_{*}}(u)) -r(u,v_{\Delta_{*}}(u)) \geq \frac{\bar{C}}{\delta_{*}} (u-u_{0}(v_{\Delta_{*}}(u)) )= \frac{\bar{C}}{\delta_{*}} \Delta_*. $$

Hence, using \eqref{monotonicity} : 

$$ r_+ -\delta \geq r_- + \delta_*+ \frac{\bar{C}}{\delta_{*}} \Delta_*. $$

So at fixed $\delta_*$, we can take $\Delta_*$ large enough so that the inequality is absurd. Therefore the lemma is proved. 

\end{proof}

Now, since $r_-(e,M) < \frac{e^2} {M}$, we choose a $\delta_*$ such that $0<\delta_* < \frac{e^2} {M}-r_-(e,M)$ and pick a $\Delta_*$ such that $r<  r_- + \delta_*$ on $\gamma_{\Delta_{*}}$.

Then, because $\nu, \lambda <0$, $r<  r_- + \delta_*$ as well in the future of $\gamma_{\Delta_{*}}$.

Therefore there exists \footnote{Notice that if $r<\frac{e^2}{M}$ then $K_{M,e}(r)<0$.} $K_{*}>0$ depending on $(e,M)$ only, such that  on $\Delta_* \leq u+v+h(v) \leq \Delta_N$ :  $$K(u,v) <-K_{*},$$-where we used again  \eqref{QNS2}, \eqref{MNS2} with $|u_s|$ large enough. 

So from  \eqref{partialvNSOmegaprop2}, we see that \eqref{OmegaNSfuture} is true :

$$ \Omega^2_{|\gamma_{N}}\lesssim  \Omega^2_{|\gamma_{\Delta_*}}  e^{K_* \Delta_*}  e^{-K_* (-\Delta+N \epsilon)} \precsim e^{-K_* N \epsilon}.$$

Then recalling from \eqref{mu} that $$\frac{1}{r^2} \left( r - (\varpi+\sqrt{\varpi^2-Q^2})\right) \left(r - (\varpi-\sqrt{\varpi^2-Q^2}) \right) =1-\mu= \frac{-4\lambda \nu}{\Omega^2}=\frac{-\Omega^2}{4 \iota \kappa},$$ we prove that,  thanks to \eqref{kappaNStpropprop} , \eqref{iotaNStpropprop} and \eqref{QNS2} , \eqref{MNS2} : 

$$ |r_{|\gamma_{\Delta'}} - (\varpi-\sqrt{\varpi^2-Q^2})| \precsim \frac{ e^{-K_* N \epsilon}}{|r_{|\gamma_{\Delta'}} - (\varpi+\sqrt{\varpi^2-Q^2})|} \precsim \frac{ e^{-K_* N \epsilon}}{|r_{|\gamma_{\Delta'}} - r_+|- \tilde{C} v^{1-2s}}.  $$

Then , since the monotonicity of $r$ ensures that $r$ is uniformly bounded away from $r_+$ on $\gamma_{\Delta'}$ and using \eqref{QNS2} and \eqref{MNS2} again on the left-hand-side, we get \eqref{rNSfuture} and \eqref{KNSfuture} for $|u_s|$ large enough.


\end{proof}

\subsection{The early blue-shift transition region} \label{transition}

We define the early blue-shift transition region : 

$$\mathcal{EB}:=  \{\Delta_N \leq  u+v+h(v) \leq -\Delta'+
\frac{2s}{2|K_-|} \log(v)\} \; ,$$ 

where $|u_s|$ is large enough so that $v_0+h(v_0)-\frac{2s}{2|K_-|} \log(v_0) < |u_s| -\Delta'  $ and $\Delta'$ is a large \footnote{Compared to $N$, $\epsilon$, $\Delta$ and the initial data.} constant to be chosen later. 

We will denote\footnote{A similar curve has been first introduced by Dafermos in \cite{Mihalis1}.} $\gamma:= \{u+v+h(v) = -\Delta'+
\frac{2s}{2|K_-|} \log(v)\}$, the future boundary of $\mathcal{EB}$.

Similarly to the region of section \eqref{redtransition}, the goal in $\mathcal{EB}$ is to obtain bounds for $\partial_v \log(\Omega^2)-2K_-$ and $\partial_u \log(\Omega^2)-2K_-$ on the future boundary instead of $\partial_v \log(\Omega^2)-2K$ and $\partial_u \log(\Omega^2)-2K$. For this to be true, we need to prove that the blue-shift in this region is strong enough, in particular we need $|r-r_-| \lesssim |u|^{1-2s} \sim v^{1-2s}$ close enough to the future boundary \footnote{Actually this bound is already attained in the future of the curve $u+v+h(v) =
\frac{2s-1}{2|K_-|} \log(v)$ and in fact, one cannot get better in general.}. 

This region exhibits enough blue-shift so that there is a good decay of the interesting quantities, but not too much so that $\kappa^{-1}$ and $\iota^{-1}$ are still under control. Moreover, the size of the region is small enough -of the order of $\log(v)$- so that we do not lose too much the control proved in the previous sections- but the decay of the metric coefficients has started and will be strong enough in the future to make the wave propagation decay easier to prove. \\

Note that in $\mathcal{EB}$ again : $|u| \sim v $.

We define the new notation : $A \lessapprox B $ if there exists a constant $\hat{C}=\hat{C}(N,\epsilon)$ such that $A \precsim \hat{C} B $. We denote $A \approx B$ if $A \lessapprox B $ and $ B \lessapprox A $.

If we need to specify this constant, we will call it consistently $\hat{C}$ when there are no ambiguities.

\begin{prop} \label{transitionBS}

For $N\epsilon$ large enough,	we have 

the following control on the field  on $\mathcal{EB}$ : 

\begin{equation} \label{phiTransition}
|\phi|\lessapprox  v^{-s} \log(v),
\end{equation}	
\begin{equation} \label{phiVTransition}
|\partial_v \phi| \lessapprox v^{-s},
\end{equation}	\begin{equation} \label{phiUTransition}
|D_u \phi| \lessapprox  |u|^{-s} \approx v^{-s} .
\end{equation}

and we also have : 
\begin{equation} \label{Omegatransition}
|\log\Omega^2(u,v)-2K_- \cdot (u+v+h(v))|\lesssim \Delta e^{-2K_+ \Delta} \sim \delta |\log(\delta)| ,
\end{equation}

\begin{equation} \label{kappatransition}
0 \leq 1-\kappa \leq \frac{1}{3},
\end{equation}			\begin{equation} \label{iotatransition}
|1-\iota|\leq \frac{1}{3},
\end{equation}
\begin{equation} \label{partialuOmegatransition}
|\partial_u \log(\Omega^2)-2K |  \lessapprox   v^{-p(s)}\log(v)^3 ,
\end{equation}
\begin{equation} \label{partialvOmegatransition}
|\partial_v \log(\Omega^2)-2K |  \lessapprox   v^{-2s}\log(v)^3 ,
\end{equation}

\begin{equation} \label{QTrans} |Q(u,v)-e| \lessapprox  v^{1-2s} ,
\end{equation}
\begin{equation} \label{MTrans} |\varpi(u,v)-M| \lessapprox v^{1-2s} .
\end{equation}

Moreover, on the future boundary $\gamma$ we have : 

\begin{equation} \label{lambdagamma}
|\lambda(u_{\gamma}(v),v)| \lessapprox e^{2|K_-| \Delta'} v^{-2s},
\end{equation}
\begin{equation} \label{nugamma}
|\nu(u,v_{\gamma}(u))| \lessapprox e^{2|K_-| \Delta'} |u|^{-2s},
\end{equation}

\begin{equation} \label{rgamma}
|r(u_{\gamma}(v),v)-r_-(M,e)| \lessapprox  e^{2|K_-| \Delta'}v^{1-2s},
\end{equation}

\begin{equation} \label{partialvOmegagamma}
|\partial_v \log(\Omega^2_{CH})(u_{\gamma}(v),v)| \lessapprox v^{1-2s},
\end{equation}
\begin{equation} \label{Omegagamma}
\Omega^2(u_{\gamma}(v),v)  \lessapprox e^{2|K_-| \Delta'} v^{-2s}.
\end{equation}

\end{prop}

\begin{proof}

First we take \footnote{This can be assumed by section \eqref{reduction} but is really not a restriction, we simply  write $|\log(2+ |v|)|$ instead of $ \log(v)$.} $v_0 \geq 2$ so that $1  \lesssim |\log(v)|= \log(v)$.

We make the following bootstrap assumptions :  

\begin{equation} \label{VfieldBoot}
|\partial_v \phi| \leq 4 C_{\Delta} 2^N v^{-s},\end{equation}
\begin{equation} \label{UfieldBoot}
|D_u \phi| \leq 4 C_{\Delta} 2^N v^{-s},\end{equation}
\begin{equation} \label{kappabootTranst}
|1-\kappa| \leq \frac{1}{2},
\end{equation}
\begin{equation} \label{iotabootTranst}
|1-\iota| \leq \frac{1}{2},
\end{equation}
\begin{equation} \label{partialuOmegaTransBoot}
\partial_u  \log(\Omega^2) \leq K_-<0,
\end{equation}
\begin{equation} \label{partialvOmegaTransBoot}
\partial_v \log( \Omega^2) \leq K_-<0.
\end{equation}

For a constant $C_{\Delta}$ such that  $|\partial_v \phi| \leq  C_{\Delta} 2^N v^{-s}$  and  $|D_u \phi| \leq  C_{\Delta} 2^N v^{-s}$ are true initially on the past boundary $\gamma_N := \{u+v+h(v) = \Delta_N\}$, using the estimates of $\mathcal{N}$.

An immediate consequence of bootstrap \eqref{partialuOmegaTransBoot}, \eqref{partialvOmegaTransBoot} and the boundedness of $\Omega^2$ in $\mathcal{N}$ (c.f Appendix \ref{appendixproof}) is the existence of a constant $\Omega^2_{max}(M,e)>0$ such that $$ \Omega^2 \leq \Omega^2_{max}(M,e).$$ \\

We now want to prove a decay on $\Omega ^{\eta}\phi$ for $\eta$ arbitrarily small. 

Let $\eta>0$. We write : 

$$ \partial_v ( \Omega^{2\eta} \phi) = \eta \cdot \partial_v \log( \Omega^2) \cdot \Omega^{2\eta} \phi+\Omega^{2\eta} \partial_v \phi.$$

Then, because of bootstraps \eqref{VfieldBoot}, \eqref{partialvOmegaTransBoot} we have  

$$ \partial_v ( \Omega^{4\eta} |\phi|^2) = 2\eta \cdot \partial_v \log( \Omega^2) \cdot \Omega^{4\eta} |\phi|^2+2\Omega^{4\eta} \Re(\partial_v \phi \bar{\phi}) \leq 8 C_{\Delta} 2^N v^{-s} \Omega^{4\eta}|\phi|, $$

which implies : 

$$ \partial_v ( \Omega^{2\eta} |\phi|) \leq 4 C_{\Delta} 2^N v^{-s} \Omega^{2\eta} v^{-s}. $$

Then it is enough to integrate using \eqref{partialvOmegaTransBoot} and Lemma \ref{calculuslemma}, the bound on the previous region and the fact that $|\Omega^{2\eta}(u,v_{\gamma}(u)) \phi(u,v_{\gamma}(u))| \lesssim |u|^{-s}$ to get  : 

\begin{equation}\label{phiEB}
\Omega^{2\eta}|\phi|\lesssim  C_{\eta} |u|^{-s} \sim  C_{\eta} v^{-s}  . \end{equation}

Using \eqref{Field3} together with bootstraps \eqref{VfieldBoot}, \eqref{UfieldBoot}, \eqref{kappabootTranst}, \eqref{iotabootTranst} and \eqref{phiEB} we show that for all $0<\eta<1$ : 

$$ |\partial_v (D_u \phi) | \lesssim (1+ C_{\eta}) C_\Delta 2^N \Omega^{2-2\eta} v^{-s}.$$

We can take $\eta= \frac{1}{2}$.

Integrating using \eqref{partialvOmegaTransBoot} with Lemma \ref{calculuslemma} and $|u| \sim v$ gives : 

$$ |D_u \phi | \leq  C_\Delta 2^N  v^{-s}+ \bar{C} 2^N \Omega_{|\gamma_N}(u, v_N(u))  v^{-s} \leq C_\Delta 2^N  v^{-s}+ \tilde{C}\bar{C} 2^N e^{-\frac{K_*}{2} N\epsilon}  v^{-s}.$$

Therefore, we can choose $N\epsilon$ large enough compared to $\Delta$ and parameters so that 

$\tilde{C} \bar{C} 2^N e^{-K_0 N\epsilon} \leq C_\Delta 2^N$ which closes bootstrap \eqref{UfieldBoot}.

Bootstrap \eqref{VfieldBoot} is validated similarly, using \eqref{partialuOmegaTransBoot}, \eqref{Field2} and the boundedness of $Q$. \\

Notice that bootstrap \eqref{UfieldBoot} and \eqref{partialuOmegaTransBoot} used together with Lemma \ref{calculuslemma} give : 

$$ \int_{u_N(v)}^{u} \frac{|D_u \phi|^2}{\Omega^2}(u',v) du' \precsim 4^N  \frac{v^{-2s}}{\Omega^2(u,v)}.$$

We integrate \eqref{RaychU} on $ [u_N(v),u]$ and multiply by $\Omega^2$ to get, using the bounds from the past : 

\begin{equation} \label{nuOmegaTrans}
|4 \nu + \Omega^2 |(u,v) \lessapprox   v^{-2s}.
\end{equation} 

Similarly with  \eqref{RaychV} :

\begin{equation} \label{lambdaOmegaTrans}
|4 \lambda + \Omega^2 |(u,v) \lessapprox   v^{-2s} + \Omega^2 v^{-p(s)} \lesssim v^{-p(s)}.
\end{equation} 

Integrating bootstrap \eqref{VfieldBoot} over $[v_{\gamma}(u),v]$ whose size is at most $ \tilde{C} \log(v)$, we get \eqref{phiTransition} :

$$| \phi| \lesssim   C_{\Delta,N} v^{-s} \log(v).$$

From this, we get : 

\begin{equation} \label{partialQ}
|\partial_u Q|+  |\partial_v Q| \lesssim  C_{\Delta,N}^2 v^{-2s} \cdot \log(v) \end{equation}  

And we can integrate to get \eqref{QTrans}. The main contribution comes from the past since   $v^{-2s}\log(v)=o( v^{1-2s})$ so for $|u_s|$ large enough : 

$$ |Q-e|\precsim  4^N  v^{1-2s} .$$

Using this together with bootstrap \eqref{kappabootTranst}, \eqref{iotabootTranst} and equations \eqref{massUEinstein}, \eqref{massVEinstein} we get : 

\begin{equation} \label{partialM}
|\partial_u  \varpi|+  |\partial_v \varpi| \lesssim  C_{\Delta,N}^2 v^{-2s} \log(v)^2 .\end{equation} 

We also integrate to get \eqref{MTrans} : 

$$ |\varpi-M| \precsim  4^N v^{1-2s} .$$

Notice that under our bootstrap assumptions we have -using \eqref{partialQ} and \eqref{partialM}- : 

$$ |(\kappa-1) \partial_u (2K) | \lesssim |4 \nu + \Omega^2 |+ C_{\Delta,N}^2v^{-2s} \log(v)^2 .$$

Now integrating \eqref{Omega2} in $u$ and remembering that $|u-u_N(v)|+|v-v_N(u)| \lesssim \log(v)$ , we get \eqref{partialvOmegatransition} as

$$|\partial_v \log(\Omega^2)-2K |  \lesssim  C_{\Delta,N}^2   v^{-2s}\log(v)^3 .$$

Where we used \eqref{nuOmegaTrans}.		Similarly, using \eqref{lambdaOmegaTrans} we prove \eqref{partialuOmegatransition} :

$$|\partial_u \log(\Omega^2)-2K |  \lesssim  C_{\Delta,N}^2  v^{-p(s)}\log(v)^3 .$$

Notice that with \eqref{QTrans}, \eqref{MTrans} and bootstrap \eqref{kappabootTranst}, \eqref{iotabootTranst} used with \eqref{mu} we have, for  $|u_s|$ large enough \footnote{$|u_s|$ is taken large enough to annihilate the dependence in $N$ and $\Delta$ of $C_{\Delta,N}^2   v^{1-2s}$. } and using the precedent section : 

$$ |2K_- -2K | \lesssim \Omega^2 +  C_{\Delta} 4^N   v^{1-2s} \lesssim \Omega^2_{|\gamma_N} \precsim e^{-K_* N\epsilon}.$$

Hence  if $N\epsilon$ is large enough and $|u_s|$ is large enough, bootstrap \eqref{partialuOmegaTransBoot} and \eqref{partialvOmegaTransBoot} are validated. \\

Notice that since $\log(v)v^{1-2s}= o(1)$, we still have :

$$v-v_N(u)= u+v+h(v)-\Delta_N+o(1).$$

From what precedes, we know that :

\begin{equation} \label{newEB}
|\partial_v \log(\Omega^2_{CH}) | =|\partial_v \log(\Omega^2)-2K_- |  \lesssim  \Omega^2+ C_{\Delta,N}^2   (v^{-2s}\log(v)^3+v^{1-2s}). 
\end{equation}

Hence we can integrate from $v_N(u)$ to $v$, using the upper bound \eqref{partialvOmegaTransBoot} with Lemma \ref{calculuslemma}  and the bounds from the past :

$$  	| \log(\Omega^2)-2K_- \cdot (u+v+h(v)) |  \lesssim  (\log(\Omega^2(u, v_N(u)))-2K_-\Delta_N)+\Omega^2(u, v_N(u))+ C_{\Delta,N}^2   (v^{-2s}\log(v)^3+v^{1-2s}) \cdot \log(v).$$

With what precedes, we see that  $$  \Omega^2(u, v_N(u))+ C_{\Delta,N}^2   (v^{-2s}\log(v)^3+v^{1-2s})\log(v) \precsim e^{-K_* N \epsilon}.$$

Hence to get \eqref{Omegatransition}, we choose $N\epsilon$ large enough compared to $\delta$ and the initial data. \\

Now we have proved that

$$	\Omega^2 \approx e^{2K_-(u+v+h(v))}.$$

It proves \eqref{Omegagamma}. Using \eqref{newEB}, we get \eqref{partialvOmegagamma}.

Notice that it also proves \eqref{lambdagamma}, \eqref{nugamma} using \eqref{nuOmegaTrans} and \eqref{lambdaOmegaTrans}.

Then, dividing \eqref{nuOmegaTrans} by $\Omega^2$ we get : 

$$ |\kappa^{-1}-1| \lessapprox e^{2|K_-|(u+v+h(v))} v^{-2s} \lesssim e^{-2|K_-| \Delta'} .$$

Hence for $\Delta'$ large enough compared to $N$, $\epsilon$, $\Delta$ and the initial data, we close bootstrap \eqref{kappabootTranst} and prove \eqref{kappatransition} with

$$  |\kappa^{-1}-1| \leq \frac{1}{4}.$$

Similarly using \eqref{lambdaOmegaTrans}, we get : 

$$ |\iota^{-1}-1| \lessapprox e^{2|K_-| (u+v+h(v))} v^{-2s} + v^{-p(s)} ,$$

which closes \eqref{iotabootTranst} and proves \eqref{iotatransition}, for $|u_s|$ large enough.


Finally \eqref{kappatransition}, \eqref{iotatransition} and \eqref{Omegagamma} give -using \eqref{mu}- that

$$|\left(r(u_{\gamma}(v),v) - (\varpi+\sqrt{\varpi^2-Q^2}\right)\left(r(u_{\gamma}(v),v) - (\varpi-\sqrt{\varpi^2-Q^2})\right) | \leq  \tilde{C} e^{-2|K_-| \Delta'} v^{-2s}.$$

Then using \eqref{QTrans} and \eqref{MTrans} with the same type of argument as in section \ref{Nfuture} -notably that $r$ is far away from $r_+(M,e)=M+\sqrt{M^2-e^2}$- , we get \eqref{rgamma} : 

$$ |r(u_{\gamma}(v),v)-r_-(M,e)| \lessapprox  e^{2|K_-| \Delta'}v^{1-2s}.$$

\end{proof}

\subsection{The late blue-shift region} \label{LB}

We then define the late blue-shift region :

$$\mathcal{LB}:=  \{-\Delta'+
\frac{2s}{2|K_-|} \log(v) \leq  u+v+h(v) \} \; .$$ 

This large region is where the essential of the blue-shift occurs : $\Omega^2$ goes from a polynomial decay in $v$ on the past boundary to an exponential decay in $v$.

In this region, $\kappa^{-1}$ and $\iota^{-1}$ are expected to blow-up \footnote{Indeed, we prove in the instability part that $\iota^{-1}$ blows up identically on the Cauchy horizon, for $u \leq u_s$. } exponentially near the Cauchy horizon if the initial bound on the field is sharp so we cannot trade $\lambda$ and $\nu$   -which decay no better than what \eqref{lambdagamma} and \eqref{nugamma} suggest- for $\Omega^2$ which decays exponentially.

However, there is enough decay of $\Omega^2$, $\nu$ and $\lambda$ on the past boundary $\gamma$ so that we can prove decay for the scalar field with \eqref{Field} using a bootstrap method. 

In $\mathcal{LB}$, we will not prove decay for $\phi$ and $D_u \phi$ -due to $|u| \lesssim v$ only- and we do not know if $-\partial_u \log(\Omega^2)$ is lower bounded like before if $s\leq1$.

Nevertheless, we can still prove that $-\partial_v \log(\Omega^2)$ is lower bounded which will allow us to prove most of the estimates. \\

We now recapitulate the constants choice : we have chosen $\Delta$ large enough  depending on $C,e,M,q_0,m^2,v_0$ in \ref{redshift}, then $\epsilon$ small enough depending on $\Delta$ and  $C,e,M,q_0,m^2,v_0$ in \ref{noshift}, then $N\epsilon$ large enough depending on $\Delta$ and  $C,e,M,q_0,m^2,v_0$ in \ref{transition} and finally $\Delta'$ large enough depending on $N$, $\epsilon$, $\Delta$ and  $C,e,M,q_0,m^2,v_0$ also in \ref{transition}.

This been said, we can consider that all the constants mentioned above depend on $C,e,M,q_0,m^2,v_0$ so we are going to write again  $A \lesssim B$ if there exists a $\tilde{D}$ depending on these constants such that $A \leq \tilde{D} B$.

\begin{prop} \label{LBprop}

We have the following estimates in $\mathcal{LB}$ : 

For all $\eta>0$, there exists $C_{\eta}>0$ : 

\begin{equation} \label{phiLB}
\Omega^{2\eta}|\phi|\lesssim  C_{\eta} v^{-s} ,
\end{equation}	
\begin{equation} \label{QLB}  \Omega^{2\eta}|Q-e| \lesssim  C_{\eta} v ^{1-2s}.
\end{equation}

And 
\begin{equation} \label{phiVLB}
|\partial_v \phi| \lesssim v^{-s},
\end{equation}		\begin{equation} \label{partialvOmegaLB}
|\partial_v \log(\Omega^2_{CH}) |  \lesssim |u|^{1-s}v^{-s}1_{ \{s>1\}} +   v^{1-2s}1_{ \{s<1\}}  + 1_{ \{s=1\}} \log(v) v^{-1},
\end{equation}					
\begin{equation} \label{lambdaLB}
0<	 -\lambda \lesssim  v^{-2s},
\end{equation}
\begin{equation} \label{nuLB}
0< -\nu \lesssim |u|^{-2s}.
\end{equation}

Moreover if $s>1$ we have :

\begin{equation} \label{A}
|D_u \phi| \lesssim  |u|^{-s} ,
\end{equation}
\begin{equation} \label{B}
|\partial_u \log(\Omega^2_{CH})-2K_- |  \lesssim   |u|^{1-2s},
\end{equation}
\begin{equation} \label{C}
|\partial_v Q |  \lesssim   |u|^{1-s}v^{-s},
\end{equation}
\begin{equation} \label{D}
|\partial_u Q|  \lesssim   |u|^{1-2s}.
\end{equation}

\end{prop}

\begin{proof}   
We make the following bootstrap assumptions :  

\begin{equation} \label{VfieldBootVic}
|r\partial_v \phi| \leq 2 \check{C} v^{-s},\end{equation}
\begin{equation} \label{lambdaboot}
|\lambda| \leq 2\check{D} v^{-2s},\end{equation}
\begin{equation} \label{partialvOmegaBootVic}
\partial_v \log( \Omega^2) \leq K_-,
\end{equation}

for $\check{C}>0$ chosen so that on the past boundary $\gamma$ we have : $ |r\partial_v \phi| \leq  \check{C} v^{-s}$ and  $\check{D}$ is a large enough constant to be chosen later such that $|\lambda| \leq  \check{D} v^{-2s}$ on $\gamma$.

Notice that because of \eqref{RaychV}, $\iota$ decreases in $v$ so by the previous bound on $\gamma$ we can write : 

\begin{equation} \label{Omegapolyn}
\Omega^2 \leq -6 \lambda \leq 12 \check{D} v^{-2s}.
\end{equation}

For the proof, we introduce a curve $\gamma_{\mathcal{V}}:= \{ u+v+h(v) = \frac{v}{2}\}$ whose future domain  $\mathcal{V}=\{ u+v+h(v) \geq \frac{v}{2}\}$ is called the vicinity of the Cauchy horizon. \\

We start to integrate \eqref{partialvOmegaBootVic} to get, using the bounds on the previous region and choosing $|u_s|$ large enough so that $\log(\Omega^2)\leq 0$ : 

$$ \log(\Omega^2)(u,v) \leq  \log(\Omega^2)(u,v_{\gamma}(u))+K_-(v-v_{\gamma}(u))\leq K_-(v-v_{\gamma}(u)),$$

and since $v_{\gamma}(u) \leq -\frac{3}{2}u$ for $u_s$ negative enough, we get :

$$\Omega^2 \leq e^{K_-(v -\frac{3}{2}|u|)}.$$

Notice that on $\mathcal{V}$, $|u| \leq \frac{v}{2}+h(v)$ hence $v-\frac{3|u|}{2} \geq \frac{v}{4}+o(v)$ so in $\mathcal{V}$ for $|u_s|$ large enough \footnote{Of course this bound is far from sharp : actually for all $\epsilon_0>0$, there exists a region sufficiently close to the Cauchy horizon so that $\Omega^2 \lesssim e^{(2K_-+\epsilon_0)v}
$. We will not need such a sharp bound. } since $ K_-<0$ :

\begin{equation} \label{Omegaexp}
\Omega^2 \lesssim e^{\frac{K_-}{5}v} .\end{equation}

The following lemma will prove \eqref{phiLB} and \eqref{QLB} :

\begin{lem} \label{lemmafinal}

Assuming the bootstraps stated above, we have the following estimates in $\mathcal{LB}$ : for all $\eta>0$, there exists $C_{\eta}>0$ such that : 

\begin{equation} \label{philem}
	\Omega^{2\eta}|\phi|\lesssim  C_{\eta} v^{-s} ,
\end{equation}	
\begin{equation} \label{Qlem}  \Omega^{2\eta}|Q-e| \lesssim  C_{\eta} v^{1-2s}.
\end{equation}

\end{lem}
\begin{proof}

Let $\eta>0$. We write : 

$$ \partial_v ( \Omega^{2\eta} \phi) = \eta \cdot \partial_v \log( \Omega^2) \cdot \Omega^{2\eta} \phi+\Omega^{2\eta} \partial_v \phi.$$

Then, because of bootstraps \eqref{VfieldBootVic}, \eqref{partialvOmegaBootVic} we have  

$$ \partial_v ( \Omega^{4\eta} |\phi|^2) = 2\eta \cdot \partial_v \log( \Omega^2) \cdot \Omega^{4\eta} |\phi|^2+2\Omega^{4\eta} \Re(\partial_v \phi \bar{\phi}) \leq \frac{4 \check{C} } {r}v^{-s} \Omega^{4\eta}|\phi|, $$

which implies : 

$$ \partial_v ( \Omega^{2\eta} |\phi|) \leq \frac{2 \check{C} } {r} \Omega^{2\eta} v^{-s}. $$

Then it is enough to integrate using \eqref{partialvOmegaBootVic} and Lemma \ref{calculuslemma}, the bound on the previous region and the fact that

$$|\Omega^{2\eta}(u,v_{\gamma}(u)) \phi(u,v_{\gamma}(u))| \lesssim |u|^{-s}$$ to get  : 

$$ \Omega^{2\eta}|\phi|\lesssim  C_{\eta} |u|^{-s}. $$

Now in the past of $\gamma_{\mathcal{V}}$, $|u| \sim v$ so 
\eqref{philem} is true. 

In $\mathcal{V}$, we can integrate \eqref{VfieldBootVic} to get $ |\phi| \lesssim |u|^{1-s}1_{ \{s>1\}} + v^{1-s}1_{ \{s<1\}}+\log(v)1_{ \{s=1\}}$ but the exponential decay of $\Omega^2$ in $v$ from \eqref{Omegaexp} is stronger than this potential growth for $|u_s|$ large enough, so that  \eqref{philem} is true also. 

We use the same technique to get \eqref{Qlem}, using \eqref{philem}, bootstrap \eqref{VfieldBootVic} and \eqref{ChargeVEinstein}. \\

\end{proof}

Now we can use \eqref{Radius} and what precedes to write : 

$$ |\partial_v (r\nu)| \lesssim \Omega^2 + C_{\eta}\Omega^{2-2\eta}v^{1-2s}.$$

Integrating, choosing $\eta$ small enough and using \eqref{partialvOmegaBootVic} with Lemma \ref{calculuslemma} and the bounds on the former region we prove \eqref{nuLB} : 

$$ |\nu| \lesssim |u|^{-2s}.$$

Then we can use \eqref{Field3}, \eqref{nuLB} and \eqref{VfieldBootVic} to get : 

$$ |\partial_v (r D_u \phi)| \lesssim |u|^{-2s}v^{-s}+ C_{\eta}\Omega^{2-2\eta}v^{-s}.$$

Integrating, choosing $\eta$ small enough and using \eqref{partialvOmegaBootVic} with Lemma \ref{calculuslemma} to absorb of the $C_{\eta}\Omega^{2-2\eta}v^{-s}$ term in $|u|^{-s}$ , we get : 

\begin{equation} \label{DULB}
|D_u \phi | \lesssim |u|^{-s}+|u|^{-2s}b(u,v),
\end{equation}

with $b(u,v):= v^{1-s} 1_{ \{s<1 \}} + |u|^{1-s} 1_{ \{s>1 \}}+\log(v) 1_{ \{s=1 \}} $. \\

We can then use \eqref{Field2} and bootstrap \eqref{lambdaboot} to get :

$$|\partial_{u}( e^{iq_0 \int_{u_{0}}^{u}A_u}r\partial_{v} \phi)| \lesssim  \check{D}v^{-2s} |u|^{-s} + \check{D}v^{-2s} |u|^{-2s}b(u,v)+C_{\eta}\Omega^{2-2\eta}v^{-s}.$$

Integrating on $[u_{\gamma}(v),u]$ and taking the absolute value we get : 

$$ |r\partial_{v} \phi| \leq \check{C} v^{-s} + \tilde{C}( \check{D}v^{-s}b(u,v)+ \check{D}|u|^{1-2s}  v^{-s}b(u,v) + v^{1-2s(1-\eta)}                          ) v^{-s},$$

where we used that $\Omega^{2-2\eta} \lesssim v^{- 2s(1-\eta)}$ because of \eqref{Omegapolyn} and $|u-u_{\gamma}(v)| \lesssim v$.

Noticing that  $v^{-s}b(u,v)=o(1)$ when $ v \rightarrow +\infty$, uniformly in $u$  and $v^{1- 2s(1-\eta)}=o(1)$ for $\eta$ small enough, we can close bootstrap \eqref{VfieldBootVic} for $|u_s|$ large enough \footnote{Notice that $\check{D}$ is absorbed by the decay and does not play any role.}. \\

Now in the past of $\gamma_{\mathcal{V}}$, we can prove, using $v \sim |u|$, the bounds proved before, \eqref{Omega} and arguments similar to those of section \eqref{transition} that :  

$$ |\partial_u \log(\Omega^2) -2K_- | \lesssim v^{1-2s}.$$

Hence $\partial_u \Omega^2 \leq 0$ for $|u_s|$ large enough so -denoting $C_{\gamma}$ the constant appearing in estimate \eqref{Omegagamma}- we have : 

$$ \Omega^2 (u,v) \leq  \Omega^2 (u_{\gamma}(v),v) \leq  C_{\gamma}v^{-2s}.$$

Moreover the exponential decay of \eqref{Omegaexp} makes $\Omega^2 (u,v)\leq C_{\gamma}v^{-2s} $ also true for $|u_s|$ large enough in $\mathcal{V}$

Now we integrate \eqref{RaychV}, using \eqref{partialvOmegaBootVic} and the bound \eqref{iotatransition} to get : 

$$ 4|\lambda|  \leq \frac{3}{2} \Omega^2 + \tilde{C}v^{-2s}.$$

So for $4\check{D}>\frac{3}{2}C_{\gamma}+  \tilde{C}$, bootstrap \eqref{lambdaboot} is validated. \\

Now using the preceding bounds, we get \footnote{Recall that $\partial_v \log(\Omega^2_{CH}) = \partial_v \log(\Omega^2)-2K_-$.} : 

$$ |\partial_u \partial_v \log(\Omega^2_{CH})| \lesssim |u|^{-s} v^{-s} + |u|^{-2s} v^{-s}b(u,v)+v^{-2s}+\Omega^{2-2\eta}v^{1-2s} .$$

We can integrate and -using similar methods than before- for $\eta$ small enough we get \eqref{partialvOmegaLB}, which also closes bootstrap \eqref{partialvOmegaBootVic} for $ |u_s|$ large enough : 

$$ | \partial_v \log(\Omega^2_{CH})| \lesssim b(u,v)v^{-s}. $$

Where we used that $v^{1-2s}=O(v^{-s}b(u,v))$.	To prove \eqref{A}, \eqref{B}, \eqref{C}, \eqref{D}, it is enough to use the equations, \eqref{DULB} and the fact that $b(u,v) = |u|^{1-s}$ when $s>1$, similarly to what was done in the past regions.

\end{proof}

Then we finish the proof of Theorem \ref{Stabilitytheorem} : from \eqref{lambdaLB} and \eqref{nuLB} , it is clear the $r$ admits a continuous limit $r_{CH}(u)$ when $v$ tends to $+\infty$ and that $r_{CH}(u)\rightarrow r_-(M,e) $ when $|u|$ tend to $+\infty$.

This is because we can integrate from $\gamma$ as : 

$$	r_{CH}(u) = r(u,v_{\gamma}(u)) +\int_{v_{\gamma}(u)}^{+\infty} \lambda(u,v')dv'=  r(u,v_{\gamma}(u)) + O(|u|^{1-2s}).$$

Where we used \eqref{lambdaLB} and $v_{\gamma}(u) \sim |u|$ . Then \eqref{rgamma} proves the claim. \\

Moreover, we see that\footnote{The fact the $\nu$ admits a continuous limit when $v$ tends to $+\infty$ follows easily from the estimates.} $|\nu_{\mathcal{CH}^+}(u)| \lesssim |u|^{-2s}$ is integrable, therefore $r_{CH}(u)$ is lower bounded for $|u_s|$ large enough. Hence the space-time admits the claimed Penrose diagram for $|u_s|$ large enough.

Moreover if $s>1$, $v^{1-2s}$ and $v^{-s}$ are integrable in $v$ so we can use the estimates of the last proposition and the argument from Proposition 8.14 of \cite{JonathanStab} to get a continuous extension of the space-time.

\section{Proof of the instability Theorem \ref{Instabilitytheorem}} \label{proofinstab}

\subsection{Recalling the stability estimates}\label{recall.section}

Before starting the proof of Theorem \ref{Instabilitytheorem}, we recall the stability estimates -established in the proof of Theorem \ref{Stabilitytheorem}- that are necessary to prove the instability argument. Notice that they are valid in this framework because all the hypothesis of Theorem \ref{Stabilitytheorem} are present in the hypothesis of Theorem \ref{Instabilitytheorem}.

First we recall the different regions :  

\begin{enumerate}

\item The event horizon $\mathcal{H}^+ = \{ u \equiv -\infty, v \geq v_0 \}$

\item The red-shift region $\mathcal{R} =  \{ u+v+h(v) \leq -\Delta\}$. 

\item The no-shift region $\mathcal{N}:= \{ -\Delta \leq u+v+h(v) \leq \Delta_N \}$ 

\item The early blue-shift transition region $\mathcal{EB}:=  \{\Delta_N \leq  u+v+h(v) \leq -\Delta'+
\frac{2s}{2|K_-|} \log(v)\}$ 

\item The late blue-shift region $\mathcal{LB}:=  \{-\Delta'+
\frac{2s}{2|K_-|} \log(v) \leq  u+v+h(v) \}$ composed of the past of $\gamma_{\mathcal{V}}:= \{ u+v+h(v) = \frac{v}{2}\}$ and its future called $\mathcal{V}=\{ u+v+h(v) \geq \frac{v}{2}\}$.

\end{enumerate}

Then we gather the different bounds from section \ref{proofstab} that we will use in this section : 

\begin{enumerate}

\item on $\mathcal{H}^+$, we know that : 

\begin{equation} \label{lambdapositive}
\lambda \geq 0 .
\end{equation}

\item We have the following estimates:  in $\mathcal{R}$,

\begin{equation} \label{DUR}
|D_u \phi |(u,v) \lesssim \Omega^2(u,v) v^{-s} .\end{equation}
\item In $\mathcal{N} \cup \mathcal{EB}$ :	

\begin{equation} \label{PhiB}
| \phi |(u,v) \lesssim  v^{-s}\log(v), \end{equation}
\begin{equation} \label{DUB}
|D_u \phi |(u,v) \lesssim  v^{-s}, \end{equation}
\begin{equation} \label{iotaB}
0< \iota^{-1} \sim 1 ,\end{equation}
\begin{equation} \label{kappaB}
0<	\kappa^{-1} \sim 1. \end{equation}

\item In $\mathcal{R} \cup \mathcal{N} \cup \mathcal{EB}$ : 	

\begin{equation} \label{masschargeI}
|\varpi-M| +|Q-e| \lesssim v^{1-2s}.
\end{equation}		

\item In $ \mathcal{EB}$ :	

\begin{equation} \label{OmegaEB}\Omega^2 \sim e^{2K_- (u+v+h(v))	},		\end{equation}
\begin{equation} \label{rboundedaway}
|r-r_+| \gtrsim 1	.					\end{equation}

\item In $\mathcal{EB} \cup \mathcal{LB}$ : 

\begin{equation} \label{Omegadecroissantv}
\partial_v  \log(\Omega^2) < K_- <0,\end{equation}
\begin{equation} \label{DULBI}
|D_u \phi |(u,v) \lesssim |u|^{-s}+|u|^{-2s}b(u,v),
\end{equation}

\begin{equation} \label{OmegaV}
| \partial_v \log(\Omega^2_{CH}) |(u,v) \lesssim v^{-s}b(u,v),
\end{equation}

with $b(u,v):= 1_{ \{s>1 \}}|u|^{1-s}+ v^{1-s} 1_{ \{s<1 \}}+\log(v) 1_{ \{s=1 \}}$.

For all $\epsilon_0>0$, there exists a constant  $C_{\epsilon_0}>0$ such that : 

\begin{equation} \label{phiLBI}
\Omega^{2\epsilon_0}|\phi|\lesssim  C_{\epsilon_0} v^{-s}, 
\end{equation}	
\begin{equation} \label{QLBI}  \Omega^{2\epsilon_0}|Q-e| \lesssim  C_{\epsilon_0} v ^{1-2s},
\end{equation}
\begin{equation} \label{lambdaELB}
|\lambda| \lesssim \Omega^2 +v ^{-2s}.
\end{equation}

\item In $\mathcal{EB} \cup \mathcal{LB}-\mathcal{V}$ : 
\begin{equation} \label{Omegadecroissantu}
\partial_u \log( \Omega^2) < K_-<0. \end{equation}

\item In $\mathcal{LB}$ : 
\begin{equation} \label{OmegapolynI}
\Omega^2 \lesssim  v^{-2s},
\end{equation}
\begin{equation} \label{lambdaLBI}
|\lambda| \lesssim  v^{-2s}.
\end{equation}
\item In $\mathcal{V}$ : 
\begin{equation} \label{OmegaexpI}
\Omega^2 \lesssim e^{\frac{K_-}{5}v}. \end{equation}

\end{enumerate}

\begin{figure}

\begin{center}

\includegraphics[width=65 mm, height=65 mm]{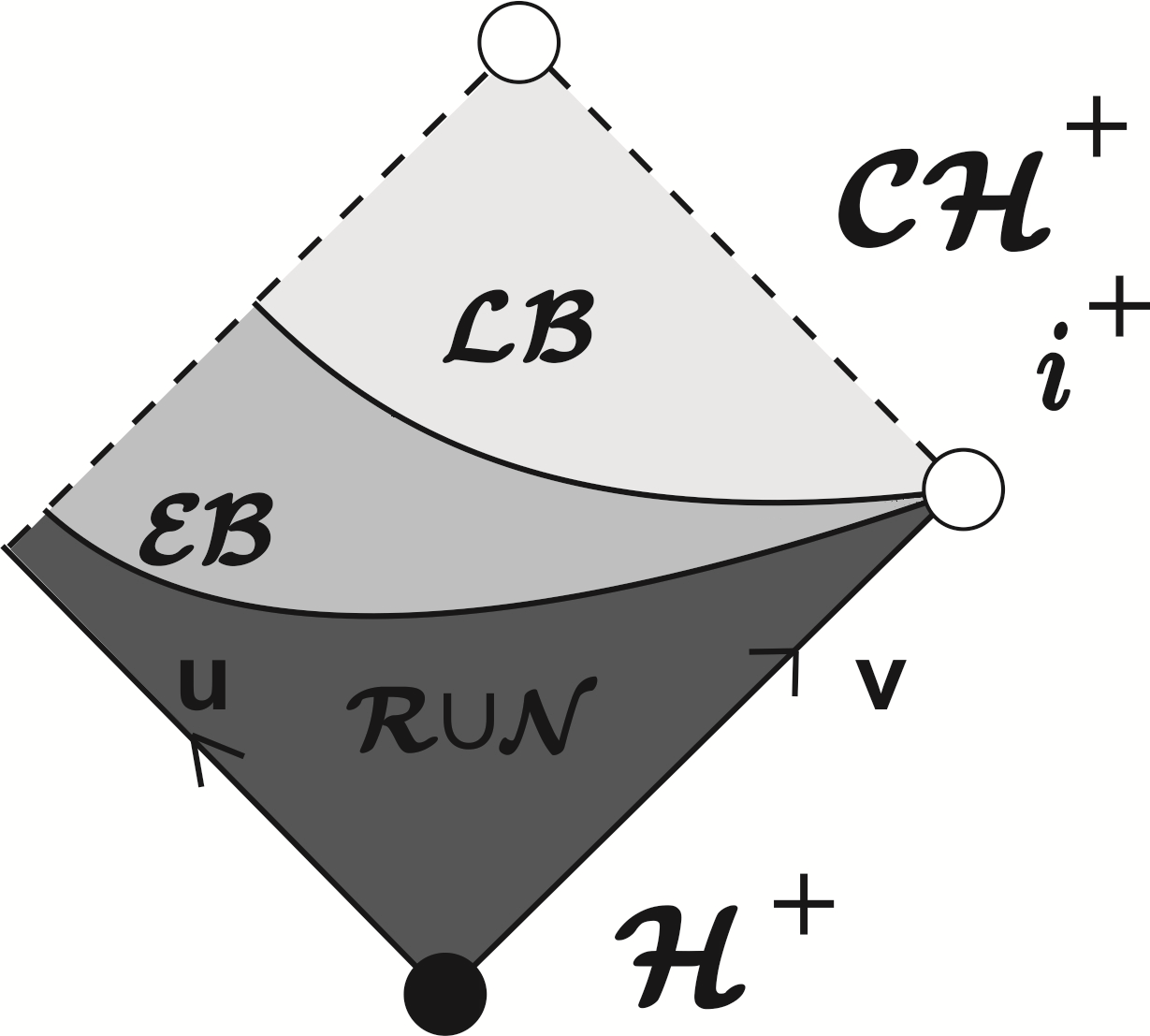}

\end{center}

\caption{Penrose diagram of the space-time $\mathcal{M}=\mathcal{R} \cup \mathcal{N} \cup \mathcal{EB} \cup \mathcal{LB} $}
\label{Love4}
\end{figure}

\subsection{Reduction to the proof of \eqref{instabilityresult}}

In this section, we want to highlight that the polynomial lower bound  \eqref{instabilityresult} for the derivative of $\phi$ transversally to the Cauchy horizon is enough to establish all the other claims of Theorem \ref{Instabilitytheorem}

The blow-up of the curvature  follows directly from \ref{instabilityresult} as first highlighted in the pioneering work \footnote{Note that \cite{JonathanStab} proves more : in an appropriate global setting, they manage to prove the $C^{2}$ inextendibility of the metric. } \cite{JonathanStab}: indeed we can consider

$$ Ric( \Omega^{-2} \partial_v,\Omega^{-2} \partial_v)=2 \Omega^{-4} |\partial_v \phi|^2 .$$ 

\eqref{instabilityresult} then gives that  : 

$$ \limsup_{v \rightarrow +\infty} Ric( \Omega^{-2} \partial_v,\Omega^{-2} \partial_v)(u,v) = +\infty, $$

using for instance the exponential lower bound for $\Omega^{-4}$ given by \eqref{OmegaexpI} in $\mathcal{V}$. \\


If $s>1$, we consider the continuous extension $\bar{M}$ and the future boundary null $\mathcal{CH}^+:=\{ V \equiv 1 ,  0 \leq  U \leq U_{0}  \}$ mentioned in the statement of Theorem \ref{Stabilitytheorem}. 

Notice that \eqref{OmegaV} proves in that case that $\partial_v \log(\Omega^2_{CH})(u,.)$ is integrable in $v$ hence $(u,V)$ is a regular coordinate system across the extension : in particular $\Omega^2_{CH}>0$ on $\mathcal{CH}^+$.

If $\mathcal{U}$ is a neighbourhood in $\bar{M}$ with compact closure -in particular with a finite range of $u$- of a point $p \in \mathcal{CH}^+$, and $\phi$ is a spherically symmetric function, its $W^{1,2}_{\mathcal{U}}$ norm can be expressed in $(u,V)$ and $(u,v)$ coordinates - as developed in \cite{JonathanInstab}- as  : 

\begin{equation}
\| \phi \|^2 _{W^{1,2}(\mathcal{U})}= \int_{\mathcal{U}}\left( |\partial_V \phi|^2 +|\partial_u \phi|^2+|\phi|^2 \right)du dV \sim \int_{\mathcal{U}}\left( \Omega^{-2}|\partial_v \phi|^2 +\Omega^{2}(|\partial_u \phi|^2+|\phi|^2) \right) du dv.
\end{equation}

Since  $\mathcal{U}$ is a neighbourhood of $p$, consider the smaller neighbourhood $\mathcal{U}':=  \mathcal{U} \cap  \mathcal{V} $.

Then, using the fact from \eqref{Omegadecroissantv} that $\partial_v \Omega^2 \leq 0$ : 

$$\| \phi \| ^2_{W^{1,2}(\mathcal{U})} \gtrsim  \int_{\mathcal{U}'} \Omega^{-2}|\partial_v \phi|^2  du' dv' .$$

We can then use \eqref{OmegaexpI} -valid in $\mathcal{U}'$- with \eqref{instabilityresult} to get $$\| \phi \| _{W^{1,2}(\mathcal{U})}= +\infty,$$ i.e

$$ \phi  \notin W^{1,2}_{loc}.$$ \\

Now we want to prove that the continuous extension to $\mathcal{CH}^+$ of Theorem \ref{Stabilitytheorem} is not $C^1$.

We integrate \eqref{RaychV} on $[v_{\gamma_{\mathcal{V}}}(u), v]$. Using that  $\iota^{-1} \geq 0$ we get : 

$$ \iota^{-1}(u,v) \geq  \iota^{-1}(u,v_{\gamma_{\mathcal{V}}}(u)) + \int_{v_{\gamma_{\mathcal{V}}}(u)}^{v} \frac{4r}{\Omega^2} |\partial_v \phi|^2 (u,v')dv' \gtrsim  \int_{v_{\gamma_{\mathcal{V}}}(u)}^{v} \frac{|\partial_v \phi|^2(u',v')}{\Omega^2}  dv'. $$ 

Which means using the same argument as a few lines above that for all $u \leq u_s$ and when $v \rightarrow +\infty $ : 

$$ \iota^{-1}(u,v) = \frac{-4\lambda}{\Omega^2}\rightarrow +\infty .$$

And since $\iota^{-1}$ is unchanged for the coordinate system $(u,V)$ that is regular near the Cauchy horizon, i.e the system allowing for the continuous extension, it proves that the metric is not\footnote{More precisely, $|\partial_V r|$ blows up identically on $\mathcal{CH}^+$ because $\Omega^2_{CH}>0$ and $\iota^{-1}$ blows up.} $C^1$ in the continuous extension of Theorem \ref{Stabilitytheorem} for $s>1$.

\subsection{Main instability argument}

Since only the averaged - opposed to pointwise- lower bound of hypothesis \ref{instabhyp} is available, we use the scalar-field fluxes controlled by the renormalized Hawking mass $\varpi$ (which is not conserved, contrary to the uncharged case) in the past of $\gamma$, the future boundary of $\mathcal{EB}$.

In the future of $\gamma$, we simply use the propagation equation \eqref{Field2} and integrate along the $u$ characteristic taking advantage on the upper bound 
$\Omega^2 \lesssim v^{-2s}$ on $\gamma$, using similar techniques to that of section \ref{LB}. The key point is that the energy flux on $\gamma$ is controlled by the integral of $|\partial_v \phi|^2$ on $\gamma$. This is due to the fact. that $\kappa^{-1}$ and $\iota^{-1}$ are bounded on $\gamma$.

We will use the same notations as in the stability part. 

Moreover, for $v \geq v_0$, we denote   $\gamma_{v} := \{  (u',v') \in \gamma, v' \geq v  \}$ and denote $u_{\gamma}(v)$ the unique $u$ such that $(u,v) \in \gamma$ and $\mathcal{H}^+_v := \mathcal{H}^+ \cap \{  v' \geq v  \}$ . $n'$ denotes the future directed unit normal of $\gamma$.

$vol$ is the standard volume form induced by the metric, and is written in $(u,v)$ coordinates as $$vol = \Omega^2 r^2 \sin(\theta) du dv d\theta d\psi,$$ where $(\theta,\psi)$ are the standard coordinates on $\mathbb{S}^2$.


We start with the propagation of lower bounds from the event horizon $\mathcal{H}^+$ to the past of $\gamma$, in the following proposition.

\begin{prop}\label{lower.past.prop}
The following inequality holds: for all $v$ sufficiently large

\begin{equation}\label{instab1}\begin{split}
	&	\int_{v}^{+\infty}  \frac{r^2}{2} [ \kappa^{-1} |\partial_v \phi|^2- \iota^{-1} |D_u \phi|^2](u_{\gamma}(v'),v') dv' \\ & \geq \frac{r_+^2}{2}\int_{v}^{+\infty} |\partial_v \phi|^2_{|\mathcal{H}^+}(0,v')dv'
	+   q_0 e  [\frac{r_+}{r_-}-1] r_+ \int_{v}^{+\infty} \Im(\bar{\phi}_{|\mathcal{H}^+}(0,v') \partial_v \phi_{|\mathcal{H}^+}(0,v')) dv'+O(v^{2-4s} \log(v))+O(v^{-2s} \log^2(v)).\end{split}	\end{equation}
	Therefore, assuming Assumption~\ref{instabhyp}, we have:
	\begin{equation}\label{instab2}
\int_{v}^{+\infty}   |\partial_v \phi|^2(u_{\gamma}(v'),v') dv' \gtrsim v^{-p}.	\end{equation}
\end{prop}

\begin{proof}
\eqref{instab2} follows immediately from \eqref{instab1}, in view of the fact that $\iota_{|\gamma}\geq 0$  and $\kappa^{-1}_{|\gamma} \lesssim 1$. To prove \eqref{instab1}, we integrate \eqref{massUEinstein} to obtain, using the estimates of Section~\ref{recall.section}: \begin{equation}\label{M1}
\bigr|	\varpi(u_{\gamma}(v),v)- 	\varpi_{|\mathcal{H}^+}(v)\bigl| \lesssim v^{-2s} \log^2(v).\end{equation}
Similarly, integrating \eqref{chargeUEinstein}, we get \begin{equation}\label{Q1}
\bigr|	Q(u_{\gamma}(v),v)- 	Q_{|\mathcal{H}^+}(v)\bigl| \lesssim v^{-2s} \log^2(v).\end{equation}

Next, combine \eqref{massUEinstein}, \eqref{massVEinstein}  with \eqref{chargeUEinstein}, \eqref{ChargeVEinstein} to obtain \begin{align*}
& \partial_u \varpi = \frac{r^2}{2\iota}| D_u\color{black} \phi |^2  + \frac{m^2}{2} r^2 \partial_u r|\phi|^2  +\frac{ Q}{r} \partial_u Q,\\ & 
\partial_v \varpi = \frac{r^2}{2\kappa}| \partial_v \phi |^2 +  \frac{m^2}{2} r^2 \partial_v r |\phi|^2+ \frac{ Q}{r} \partial_v Q.
\end{align*} 

Then, using  the identity $\frac{d}{dv} \varpi(u_{\gamma}(v),v)= u_{\gamma}(v)'\partial_u\varpi(u_{\gamma}(v),v)+ \partial_v\varpi(u_{\gamma}(v),v)$ and integrating, we get, using the estimates of Section~\ref{recall.section} on $\gamma$ \begin{equation}\label{MEgamma}
M -  \varpi(u_{\gamma}(v),v)= 	\int_{v}^{+\infty}  \frac{r^2}{2} [ \kappa^{-1} |\partial_v \phi|^2- \iota^{-1} |D_u \phi|^2](u_{\gamma}(v'),v') dv' + \frac{ e}{r_-} [e- Q(u_{\gamma}(v),v)] +O(v^{2-4s} \log(v))+ O(v^{-2s}\log^2(v)).
\end{equation} Combining \eqref{MEgamma} with \eqref{M1}, \eqref{Q1} gives \begin{equation}\label{ME2}
M - 	\varpi_{|\mathcal{H}^+}(v)= 	\int_{v}^{+\infty}  \frac{r^2}{2} [ \kappa^{-1} |\partial_v \phi|^2- \iota^{-1} |D_u \phi|^2](u_{\gamma}(v'),v') dv' + \frac{e}{r_-} [e- 	Q_{|\mathcal{H}^+}(v)] +O(v^{2-4s} \log(v))+O(v^{-2s} \log^2(v)).
\end{equation} Indeed,
\[
u_\gamma(v)'
=-1-h'(v)+\frac{s}{|K_-|v},
\]
so replacing $u_\gamma'$ by $-1$ 
produces an error
$O(v^{2-4s})+O(v^{-2s})$.
Moreover,
\[
\left|\frac{Q}{r}(u_\gamma(v),v)-\frac{e}{r_-}\right|
\lesssim v^{1-2s},
\qquad
\left|\frac{d}{dv}Q(u_\gamma(v),v)\right|
\lesssim v^{-2s}\log(v).
\]
Consequently, replacing $Q/r$ by $e/r_-$  produces an error $O(v^{2-4s}\log(v))$.	Finally, using \eqref{massVEinstein}, \eqref{ChargeVEinstein}  and the estimates of Section~\ref{recall.section} on $\mathcal{H}^+$, together with \eqref{ME2}, gives \eqref{instab1}. Note that replacing $r_{\mathcal{H}^+}(v)$ and $Q_{\mathcal{H}^+}(v)$ by
$r_+$ and $e$ produces an additional error $O(v^{2-4s})$,
which is included in the remainder in \eqref{instab1}.
\end{proof}

We now want to propagate our lower bounds to the future of $\gamma$. To circumvent the lack of decay of $Q$ and $\varpi$ near the Cauchy horizon, we do not use a vector field method any more but a more classical integration along the constant $v$ characteristic, as it was done in the stability part.

\begin{lem} \label{20}  The following estimate is true for all $u \leq u_s $ : 
\begin{equation}
\int_{\gamma_v} |\partial_v \phi|^2(u_{\gamma}(v'),v')dv' \lesssim \int_{v}^{+\infty} |\partial_v \phi|^2(u,v')dv' + O(v^{3-6s+4\eta})+o(v^{-2s}).
\end{equation}
\end{lem}

\begin{proof} We now place ourselves in the future of $\gamma_v$, a region that lies in $\mathcal{EB} \cup \mathcal{LB} $.

We use \eqref{Field2} to get -after adding and subtracting a $\Omega^2 e^2 |\phi|$ term- : 

$$ |\partial_u (e^{i q_0 \int_{u_0}^{u}A_u(u',v)du'} r \partial_v \phi)| \lesssim |\lambda| |D_u \phi| + \Omega^2 ((m^2+e^2) |\phi| + |Q^2-e^2||\phi| ).$$

We now deal with each term separately. 	To the future of $\gamma_v$, we use \eqref{lambdaLBI} : 

$$ |\lambda| \lesssim v^{-2s}. $$

We can also use in the same region the estimate \eqref{DULBI} : 	$$
|D_u \phi | \lesssim |u|^{-s}+|u|^{-2s}b(u,v),
$$

with $b(u,v):=  |u|^{1-s} 1_{ \{s>1 \}}+  v^{1-s} 1_{ \{s<1 \}}+ \log(v) 1_{ \{s=1 \}}$.

All put together, we get : 

$$ |\lambda| |D_u \phi|(u,v) \lesssim v^{-2s} |u|^{-s}+v^{-2s}|u|^{-2s}b(u,v) .$$
Then we use:

$$ \int_{u_{\gamma}(v)}^{u} \left(v^{-2s} |u'|^{-s}+v^{-2s}|u'|^{-2s}b(u',v)\right) du'\lesssim v^{-2s}b(u,v).$$

To estimate $\Omega^2 ( (m^2+e^2)|\phi| + |\phi||Q^2-e^2|)$, we use a similar technique. Using the estimates in Section~\ref{recall.section}, we obtain, for any $\eta>0$

$$ \int_{u_{\gamma}(v)}^{u_{\gamma_{\mathcal{V}}}(v)} \Omega^2(u',v) ( (m^2+e^2)|\phi|(u',v) + |\phi|(u',v)|Q^2-e^2|(u',v))du' \lesssim v^{1-3s + \eta}.$$

Using \eqref{OmegaexpI}, \eqref{phiLBI}, and \eqref{QLBI} in $\mathcal{V}$  gives a negligible contribution on $[u_{\gamma_{\mathcal{V}}}(v),u]$, because $\Omega^2$ is exponentially decreasing, which proves : 

$$ \int_{u_{\gamma}(v)}^{u} \Omega^2(u',v) ( |\phi|(u',v) + |\phi|(u',v)|Q^2-e^2|(u',v))du' \lesssim v^{1-3s + 2\eta}.$$

Now we can use that $v^{-2s}b(u,v)= v^{-2s}|u|^{1-s}1_{\{s>1\}}+o(v^{1-3s + 2\eta})$ if $\eta$ is small enough, combine all the estimates and integrate the first equation :

$$ |e^{i q_0 \int_{u_0}^{u}A_u(u',v)du'} r \partial_v \phi(u,v)-e^{i q_0 \int_{u_0}^{u_{\gamma}(v)}A_u(u',v)du'} r \partial_v \phi(u_{\gamma}(v),v)| \lesssim v^{1-3s + 2\eta}+v^{-2s}|u|^{1-s}1_{\{s>1\}}.$$

Making the difference, using upper and lower bounds for $r$ and squaring, we get : 

$$ |  \partial_v \phi(u_{\gamma}(v),v)|^2 \lesssim |\partial_v \phi(u,v) |^2 +v^{2-6s + 4\eta}+v^{-4s}|u|^{2-2s}1_{\{s>1\}}.$$

To conclude, it is enough to integrate the last estimate on $[v, +\infty]$ and noticing that $v^{1-4s}|u|^{2-2s}1_{\{s>1\}} = o(v^{-2s})$.

\end{proof}

The combination of Proposition~\ref{lower.past.prop} and Lemma~\ref{20}
then prove Theorem~\ref{Instabilitytheorem} after choosing $\eta$ small enough so that $p< 6s-3-4\eta$.

\appendix 
\section{A localisation of the apparent horizon $\mathcal{A}$} \label{appendixapparent}

As a straightforward by-product of our framework, we prove that in a non-linear setting, the apparent horizon\footnote{Indeed $\mathcal{A}$ coincides with $\{\lambda=0\}$ on the whole space-time in our coordinate choice. This is because $\lambda$ becomes strictly negative while $\kappa^{-1} \approx 1$.} $\mathcal{A}:= \{ \partial_v r=0\}$  cannot be too far or too close of the event horizon if the decay of the perturbation is upper and lower bounded.

\begin{prop} We keep the same hypothesis as for Theorem \ref{Stabilitytheorem}. $h$ is defined in equation \eqref{hdef}.

We assume the following on the event horizon $\mathcal{H}^+$ : 

\begin{hyp} \label{apparenthyp}
\begin{equation} \label{instabapparent2}
	C' v^{-p-1} \leq \Omega_{H}^2(0,v) \int_{v}^{+\infty} \frac{|\partial_v \phi|^2(0,v')}{ \Omega_{H}^2(0,v')}dv',
\end{equation}  

\begin{equation} \label{stabapparent}
	|\partial_v \phi|(0,v') \leq C v^{-s}, 
\end{equation}

for  $ 2s-1 \leq  p$ and $C, C'>0$.
\end{hyp}

Then there exists constants $C_+>0, C_->0$ such that 

$$ \mathcal{A} \subset \{ C_{-} v^{-p-1} \leq \Omega^2(u,v) \leq C_{+} v^{-2s} \} = \{-\tilde{C} -\frac{(p+1)}{2K_+}\log(v) \leq u+v+h(v) \leq -\frac{2s}{2K_+}\log(v)+ \tilde{C} \}.$$

\end{prop}

\begin{rmk}
Notice that because of the exponential growth of $\Omega^2_{H}$ established in section \ref{boundsEH}, assumption \ref{apparenthyp} is consistent with the conjectured tail of the $field$ as formulated in Price's law of conjecture \ref{Priceconj}.
\end{rmk}

\begin{rmk}
Notice that if $\phi$ does not become constant near infinity on the event horizon, $\mathcal{A}$ is strictly to the future of $\mathcal{H}^+$. This is in particular true \footnote{Notice that an upper bound that assures that $\phi$ tends to $0$, like that of hypothesis \ref{fieldevent}, is enough to reduce the problem to either the trivial case $\phi \equiv 0$ where $\mathcal{A}$= $\mathcal{H}^+$ or the case where $\mathcal{A}$ asymptotically approaches time-like infinity.} if one assumes a lower bound on $\partial_v \phi$ ,like \eqref{instabapparent1} or \eqref{instabapparent2}.
Coupled with \eqref{stabapparent}, it proves that $\mathcal{A}$ must asymptotically approach time-like infinity. 
\end{rmk}

\begin{proof}

Using assumption \eqref{apparenthyp} and \eqref{RaychV} on the event horizon and recalling \eqref{lambdaevent}, we get : 

$$ v^{-p-1} \lesssim \lambda \lesssim v^{-2s}.$$

We can rewrite \eqref{Radius} in $(U,v)$ coordinates as : 

$$ \partial_U \lambda = \frac{- \Omega^2_H}{4} ( 2K-m^{2}r |\phi|^{2}). $$

Using  \eqref{phivRS}, \ref{KRedShiftprop} and section \ref{reduction}, there exists $\delta>0$ small enough so that $K_+<2K-m^{2}r |\phi|^{2}<3K_+$ in $\mathcal{R}$.

We can then integrate between the event horizon and the apparent horizon for $U \in [0, U_{\mathcal{A}}] $ to get : 

$$ v^{-p-1} \lesssim \frac{4}{3K_+} \lambda(0,v)< \int_{0}^{U_{\mathcal{A}}} \Omega^2_{H}(U',v)dU' < \frac{4}{K_+} \lambda(0,v) \lesssim v^{-2s}.$$

Then we can use \eqref{OmegaPropRedshift} to prove that : 

$$ \int_{0}^{U_{\mathcal{A}}} \Omega^2_{H}(U',v)dU' \sim  \Omega^2(U_{\mathcal{A}},v) .$$

Which gives the result.

\end{proof}

\section{Proof of Proposition \ref{propannexe} of section \ref{mainNS}} 
\label{appendixproof}

We recall proposition \ref{propannexe} for convenience :

\begin{prop2}

For small enough $\epsilon>0$, we have : 
\begin{equation} \label{phivNS}
|\phi|+|\partial_v \phi|+|D_u \phi|  \precsim 2^N v^{-s},
\end{equation}			\begin{equation} \label{ANS}
|A_u | \lesssim (N+1) \delta,
\end{equation}

we also have : 
\begin{equation} \label{OmegaPropNS}
|\log\Omega^2(u,v)-\log(-4(1-\frac{2M}{r}+\frac{e^2}{r^2}))|\precsim 4^N v^{1-2s} ,
\end{equation}  

\begin{equation} \label{kappaNStprop}
0 \leq 1-\kappa \precsim  5^N v^{-2s},
\end{equation}						\begin{equation} \label{iotaNStprop}
|1-\iota| \precsim   5^N v^{-p(s)},
\end{equation}
\begin{equation} \label{partialuNSOmegaprop}
|\partial_u \log(\Omega^2)-2K |  \precsim  5^N  v^{-p(s)},
\end{equation}
\begin{equation} \label{partialvNSOmegaprop}
|\partial_v \log(\Omega^2)-2K |  \precsim 5^N  v^{-2s},
\end{equation}

\begin{equation} \label{QNS} |Q(u,v)-e| \lesssim 4^N  v^{1-2s} ,
\end{equation}
\begin{equation} \label{MNS} |\varpi(u,v)-M| \lesssim 4^N v^{1-2s} .
\end{equation}

\end{prop2}

\begin{proof} 
We want to prove by induction on $k$ the following estimates on $\mathcal{N}_k :=\{ u+v+h(v)= -\Delta+k\epsilon \}$ :

\begin{equation} \label{phiboot}
|\phi|+|\partial_v \phi| \leq D_k v^{-s},
\end{equation}	\begin{equation} \label{phiUNSi}
|r D_u \phi| \leq  D_k v^{-s} ,
\end{equation}
\begin{equation} \label{ANSi}
|A_u | \lesssim A_k,
\end{equation}

\begin{equation} \label{OmegaPropNSi}
|\log\Omega^2(u,v)| \leq C_k,
\end{equation}  
\begin{equation} \label{Omegaborne}
\Omega^2 \leq \frac{3}{2} \Omega^2_{max}(M,e),
\end{equation}

\begin{equation} \label{kappaNStpropi}
0 \leq 1-\kappa \lesssim E_k v^{-2s},
\end{equation}						\begin{equation} \label{iotaNStpropi}
|1-\iota| \lesssim E_k  v^{-p(s)},
\end{equation}
\begin{equation} \label{partialuNSOmegapropi}
|\partial_u \log(\Omega^2)-2K | \lesssim E_k  v^{-p(s)},
\end{equation}
\begin{equation} \label{partialvNSOmegapropi}
|\partial_v \log(\Omega^2)-2K | \lesssim E_k v^{-2s},
\end{equation}

\begin{equation} \label{QNSi} |Q(u,v)-e| \lesssim D_k^2 v^{1-2s}, 
\end{equation}
\begin{equation} \label{MNSi} |\varpi(u,v)-M| \lesssim D_k^2  v^{1-2s} ,
\end{equation}

with $D_k= 2D_{k-1}$, $E_k= 5E_{k-1}$ , $C_k= C_{k-1}+K_{max}\epsilon $, $A_k = (k+1) \delta$ and $K_{max}$ depending on $(e,M)$ only. $\Omega^2_{max}(M,e)$ is defined as : $$\Omega^2_{max}(M,e):= 4(\frac{M^2}{e^2}-1) = \sup_{r \in [r_-(M,e) ,r_+(M,e)]} 4|1-\frac{2M}{r}+\frac{e^2}{r^2}|.$$ \\

The initialization of the induction comes directly from the bounds of proposition \eqref{boundsRS}, after choosing $D_0$, $E_0$ and $A_0$ consistently. Notice that $A_0 \lesssim \delta$.






Supposing the bounds are established for $\mathcal{N}_{k-1}$, we bootstrap  the following on  $\mathcal{N}_k$  : 

\begin{equation} \label{OmegaNoshift}
\Omega^2 \leq 2\Omega^2_{max}(M,e),
\end{equation}
\begin{equation} \label{VfieldNS}
|\phi|+ |\partial_v \phi| \leq2 D_k v^{-s},\end{equation}
\begin{equation} \label{kappabootNS}
|1-\kappa| \leq 2E_k v^{-2s},
\end{equation}
\begin{equation} \label{iotabootNS}
|1-\iota| \leq 2E_k v^{-2s}.
\end{equation}

Notice that because $\nu<0$ and $\lambda_{|\mathcal{H}^+} \geq 0$, we have $r \leq r_+$ everywhere. 

We first use \eqref{VfieldNS} to prove \eqref{QNSi} with \eqref{ChargeVEinstein} : $$|Q(u,v)-e| \lesssim   D_k^2 v^{1-2s}. $$

Then we can use  \eqref{OmegaNoshift} with \eqref{iotabootNS} to prove that $|\lambda|$ is bounded, \eqref{VfieldNS}, \eqref{kappabootNS} to prove \eqref{MNSi} with \eqref{massVEinstein} for $|u_s|$ large enough : 

$$|\varpi(u,v)-M| \lesssim  D_k^2 v^{1-2s} .$$

Notice that since $ \Omega^2 = -4 \iota \kappa (1-\frac{2\varpi}{r}+\frac{Q^2}{r^2})$ -as seen in equation \eqref{mu}- we have -forming the differences $\varpi-M$ and $Q^2-e^2$ and using \eqref{kappabootNS}, \eqref{iotabootNS} for $|u_s|$ large enough :

$$ 0 \leq -(1-\frac{2\varpi}{r}+\frac{Q^2}{r^2})= -(1-\frac{2M}{r}+\frac{e^2}{r^2}) +O( D_k^2 v^{1-2s}).$$

So -since $r_-$ cancels $(1-\frac{2M}{r}+\frac{e^2}{r^2})$-  for all $\eta>0$, there exists $|u_s|(\eta)$ large enough so that $ r_- -\eta < r$. For $\eta>0$ small enough, it can be easily shown that the supremum on $[r_{-}-\eta,r_+]$ is attained on $[r_{-},r_+]$ :  

$$\Omega^2_{max}(M,e) = \sup_{r \in [r_-(M,e)- \eta ,r_+(M,e)]} 4|1-\frac{2M}{r}+\frac{e^2}{r^2}|.$$

Since $ |\Omega^2| \leq 4|1-\frac{2M}{r}+\frac{e^2}{r^2}|+ \tilde{C}v^{1-2s}$, bootstrap \eqref{OmegaNoshift} is validated for $|u_s|$ large enough and proves \eqref{Omegaborne}. \\

Moreover, with the same technique using \eqref{mu}, \eqref{QNSi}, \eqref{MNSi} and bootstrap \eqref{kappabootNS}, \eqref{iotabootNS} we can prove \eqref{OmegaPropNS}, choosing $|u_s|$ large enough.

\eqref{QNSi}, \eqref{MNSi} also prove that :  \begin{equation} \label{KNSannex}
| 2K(u,v) - 2K_{M,e}(r(u,v))| \lesssim D_k^2 v^{1-2s}.
\end{equation} 

Using the same argument as in the red-shift region and \eqref{iotabootNS} we get -for $|u_s|$ large enough- : 

\begin{equation} \label{Kpartialvnoshift}
|\partial_v (2K) | \lesssim \Omega^2 \leq 2\Omega^2_{max}.
\end{equation} \\

We denote $v_i=v_i(u)$ the unique $v$ such that $u+v+h(v)= \Delta_i$.

Notice that from \eqref{hderiv} : 

$$ |h(v_{k-1})-h(v)| \lesssim v_{k-1}^{1-2s} |v- v_{k-1}| \approx |u|^{1-2s} |v- v_{k-1}| \approx v^{1-2s} |v- v_{k-1}|.$$

Hence , because $u+v+h(v)- \Delta_{k-1}  \leq \epsilon$ is bounded : 

\begin{equation} \label{h(v)}
v- v_{k-1} = \frac{u+v+h(v) - \Delta_{k-1}}{1+ O(v^{1-2s})} = u+v+h(v) - \Delta_{k-1}+ O(v^{1-2s})
\end{equation}  \\

We use \eqref{potentialEinstein} and \eqref{OmegaNoshift} to get : 

$$ | \partial_v A_u | \lesssim 1 .$$

Hence by induction, we get \eqref{ANSi} with

$$  |A_u| \leq A_{k-1}+ \tilde{C} \epsilon \leq A_{k},$$

after choosing $\epsilon$ small enough compared to $\delta$. \\

Then we use \eqref{Field3} with \eqref{OmegaNoshift}, \eqref{VfieldNS}, \eqref{kappabootNS} to get : 

$$ | rD_u \phi | \leq D_{k-1}|u|^{-s} + \tilde{C} D_k  \epsilon v^{-s}.$$

Hence for $\epsilon$ small enough compared to $(C,e,M,q_0,m)$, we get \eqref{phiUNSi}.

Using \eqref{Field2} and the same type of argument, we close bootstrap \eqref{VfieldNS} and get \eqref{phiboot} after integrating $\partial_v \phi$ on a $\epsilon$-small region. \\

We can then use bootstrap \eqref{OmegaNoshift}, \eqref{VfieldNS} and \eqref{kappabootNS} and notice that $|\partial_u \varpi| +|\partial_u Q^2| \lesssim D_k^2 v^{-2s}$  to get from \eqref{Omega2}: 

\begin{equation} \label{remarkableVNoShift}
|\partial_v \log(\Omega^2)-2K| \lesssim (E_{k}+D_k^2) v^{-2s} \lesssim E_{k} v^{-2s}.
\end{equation}















Then, because of the discussion above, $\frac{r_-(e,M)}{2}<r <r_+(e,M)$ therefore there exists $K_{max}=K_{max}(e,M)>0$ such that $|K| < K_{max}$.

Using \eqref{remarkableVNoShift} and the induction hypothesis, we prove \eqref{OmegaPropNSi} and get -choosing $|u_s|$ large enough- : 
$$ \Omega^{-2} \precsim   e^{2K_{max} k \epsilon}.$$

Hence from \eqref{RaychU2} and \eqref{phiUNSi} we get : 

$$ | \kappa -1 | \leq E_{k-1} + \bar{C} D_k^2 e^{2K_{max} k \epsilon}\epsilon = E_{k-1} + \bar{C} D_0^2 e^{ (\log(4) + 2K_{max}\epsilon) k }\epsilon .$$

We proceed in two times : first with choose $\epsilon$ small enough so that $\log(4) + 2K_{max}\epsilon \leq \log(5)$. We get : 

$$ | \kappa -1 | \leq E_{k-1} + \bar{C} D_k^2 e^{2K_{max} k \epsilon}\epsilon = E_0 5^{k-1} + \bar{C} D_0^2 5^k \epsilon. $$

Than we can choose $\epsilon$ even smaller so that bootstrap \eqref{kappabootNS} is validated. \eqref{kappaNStpropi}, and \eqref{partialvNSOmegapropi} are proved simultaneously, using \eqref{Omega2} for \eqref{partialvNSOmegapropi} similarly to what was done before.

Symmetrically in $v$ we use the same methods to close bootstrap \eqref{iotabootNS} and to prove \eqref{iotaNStpropi}, \eqref{partialuNSOmegapropi}.

The induction is then proved and the estimates of the proposition follow directly.

\end{proof}

\section{Proof of Lemma \ref{reductionlemma} of section \ref{reduction}}  \label{appendixlemma}

We recall Lemma \ref{reductionlemma} : 

\begin{lems} 
Under the same hypothesis than before and for $v_0' > v_0$, if $U_s$ is sufficiently small there exists a constant $D>0$ depending on $C,e,M,q_0, m^2,s,v_0$ and $v_0'$ such that

\begin{equation} \label{DUnew}
|D_U \phi(U,v_0')| \leq D .
\end{equation}

\begin{equation} \label{DURnew}
|\partial_U r(U,v_0') |^{-1} \leq D .
\end{equation}

Therefore, for any $\eta>0$ independent \footnote{We insist that $\eta$ must be a numerical constant that do \textbf{not} depend on \textbf{any} of the $C,e,M,q_0, m^2,v_0$ or $v_0'$.} of any parameter, there exists a $v_0'>0$ such that 

$$ |D_U \phi(U,v_0')| \lesssim C ,$$

and for all $v \geq v_0'$ : 

$$ |2K(0,v)-2K_+| \leq \eta K_+ ,$$
$$ r m^2 |\phi|^2(0,v) \leq \eta K_+.$$
\end{lems}

\begin{proof} We will make the following bootstrap assumptions : 

\begin{equation} \label{B1reduc}
|\nu_H| \leq B_1
\end{equation}
\begin{equation} \label{B2reduc}
|Q| +  |\phi|+ |\partial_v \phi| \leq B_2
\end{equation}

The set of points such that the bootstraps are valid is non empty because of the hypothesis of Theorem \ref{Stabilitytheorem}, for $B_2$ large enough with respect to $C$, $e$ and $v_0$ and $B_1 >1$.

Notice that with our hypothesis $r(0,v)>0$ and since $[v_0,v_0']$ is a compact, it is clear that $r(0,v)$ is upper and lower bounded by strictly positive constants that depend on $v_0$ and $v_0'$.

If we integrate \eqref{B1reduc} for $U_s$ small enough compared to $B_1$, we see that the same conclusion holds true for $r(U,v)$ on the whole rectangle $[0, U_s] \times [v_0, v_0']$ . We write $0< r_{min} <r < r_{max}$. \\

Then, notice that $\kappa(0,v) \equiv 1$ and the positive right hand side of \eqref{RaychU2} give that $0 \leq \kappa \leq 1$ everywhere on the space-time namely $\Omega_H^2 \leq -4 \nu_H$.

Then we write \eqref{Radius} as 

$$ | \partial_v ( \log(-r\nu_H)) \leq \frac{1}{r}(1+ \frac{Q^2}{r^2})+ m^2 r |\phi|^2 \leq \frac{1}{r_{min}}+ \frac{B_2^2}{r^3_{min}}+ m^2 r_{max} B_2^2    . $$

We can then integrate in $v$ and use gauge \eqref{gauge1}  to get : 

$$ |\log(-\nu_H)| \leq |\log(\frac{r_{max}}{r_{min}})|+ ( \frac{1}{r_{min}}+ \frac{B_2^2}{r^3_{min}}+ m^2 r_{max} B_2^2)(v_0'-v_0).$$

This closes bootstrap \eqref{B1reduc} for $B_1$ large enough with respect to $B_2$, $v_0$, $v_0'$ and the parameters and proves \eqref{DURnew}.\\

Now we want to bound $\lambda$ : to do so we write \eqref{Radius} as : 

$$ | \partial_U( r \lambda)| \leq -\nu_H (1 + \frac{Q^2}{r^2}+ m^2 r^2 |\phi|^2) \leq B_1 (1+\frac{B_2^2}{r^2_{min}}+ m^2 r^2_{max} B_2^2 ).$$

Now notice that on the compact $[v_0, v_0']$, $|\lambda|(0,v) \leq \lambda_{max}$ where $\lambda_{max}$ depends on $v_0$, $ v_0'$ and the parameters.

Then we can integrate the previous equation and take $U_s$ small enough to get everywhere : 

$$|\lambda|(U,v) \leq 2\lambda_{max}.$$

Now we write \eqref{Field3} as

$$ | \partial_v (r D_U \phi)| \leq -\nu_H( m^2 |\phi| + |\partial_v \phi|) \leq (m^2+1) B_1 B_2 .$$

Then we integrate and use assumption \ref{fieldUevent} and the bounds on $r$ to get : 

\begin{equation} \label{DUreduc}|D_U \phi| \leq \frac{r_{max}}{r_{min}} C + \frac{(v_0'-v_0)(m^2+1)B_1 B_2}{r_{min}}. \end{equation}

Now we use gauge \footnote{Note that it does not matter whether the gauge is on $v \equiv v_0$ or $v \equiv v_0'$ : we simply integrate from the curve where $A_u=0$.} \eqref{gaugeAponctuelle} to integrate \eqref{potentialEinstein} : 

$$ |A_U| \leq \frac{2B_1 B_2}{ r_{min}^2}(v_0'-v_0).$$

This, with bootstrap \eqref{B2reduc} and \eqref{DUreduc} gives : 

$$ |\partial_U \phi| \leq \frac{r_{max}}{r_{min}} C + \frac{(v_0'-v_0)B_1 B_2}{r_{min}} [(m^2+1)+ \frac{2 q_0 B_2}{ r_{min}}].$$

It now suffices to integrate for $U_s$ small enough to close the $\phi$ part of bootstrap \eqref{B2reduc}.

The $Q$ part of	bootstrap \eqref{B2reduc} is validated when we integrate \eqref{chargeUEinstein} using \eqref{DUreduc}.

For the $\partial_v \phi$ part, we write \eqref{Field2} as :

$$ |\partial_{U}( e^{iq_0 \int_{U_{0}}^{U}A_U}r\partial_{v} \phi)| \leq |\lambda| |D_U \phi| + |\nu_H| (r m^2 |\phi| + \frac{q_0 |Q| |\phi|}{r}) $$

Then, from all the bounds that precedes it is clear that we can integrate on $ [0,U] $ and close the $\partial_v \phi$ part of bootstrap \eqref{B2reduc} if we chose $U_s$ small enough. 

Notice that $B_2$ can be chosen to depend on $C$, $v_0$ and $e$ only. Hence $B_1$ can be chosen to depend on $v_0$, $v_0'$ and the parameters only.

In the end both bootstraps are validated. \\

Notice that \eqref{DUreduc} gives actually \eqref{DUnew} now that the bootstraps assumptions are proved.

From the last section \ref{convergence},  $2K(0,v)	-2K_+ \rightarrow 0$ when $v \rightarrow +\infty$  and from the hypothesis \ref{fieldevent} and the boundedness \footnote{Indeed $r$ converges to $r_{\infty}$ when $v$ tends to $+\infty$ and the interval $[v_0,+\infty]$ is lower bounded.} of $r$  we know that $rm^2 |\phi|^2\rightarrow 0$ when $v \rightarrow +\infty$. We can write $\max \{ |2K(0,v)	-2K_+| ,rm^2 |\phi|^2  \} = K_+\epsilon(v)$ and $\epsilon(v) \rightarrow 0$ when $v \rightarrow +\infty$.

Therefore for all $\eta>0$ -independent of all the other constants- there exists $v_0'$ -depending only on the parameters and $v_0$ such that for all $v' \geq v_0'$, $|\epsilon(v')| \leq \eta$.

Therefore, combining with \eqref{DUnew}, the lemma is proven.

\end{proof}


\begin{thebibliography}{9}

\bibitem{Volker} Spyros Alexakis, Volker Schlue \emph{Non-existence of time periodic vacuum space-times}.
Journal of Differential Geometry, to appear.  	arXiv:1504.04592.

\bibitem{Boson} Piotr Bizo\'{n}, Arthur Wasserman \emph{On existence of mini-boson stars}.
Comm.Math.Phys.215:357-373 , 2000.

\bibitem{Phycists2} Lior Burko,  Gaurav Khanna \emph{Universality of massive scalar field late-time tails in black-hole spacetimes}.
Phys.Rev. D70, 044018,  2004.

\bibitem{Otis} Otis Chodosh, Yakov Shlapentokh-Rothman,  \emph{Time-periodic Einstein-Klein-Gordon bifurcations of Kerr}.
Preprint, 2016. arXiv:1510.08025.



\bibitem{MGHD}Yvonne Choquet-Bruhat,  \emph{Theoreme d’existence pour certains systemes d'equations aux derivees partielles
non linaires.
}
Acta Math. 88, 141–225, 1952.

\bibitem{GHD}
Yvonne Choquet-Bruhat, Robert Geroch,  \emph{Global Aspects of the Cauchy Problem in General Relativity.}
Comm. Math. Phys. 14, 329–335, 1969.

\bibitem{Christo4}
Demetrios Christodoulou, \emph{Violation of cosmic censorship in the gravitational collapse of a dust cloud. 
}
Comm. Math.
Phys. 93, 171–195, 1984.

\bibitem{Christo1}
Demetrios Christodoulou, \emph{The Formation of Black Holes and Singularities in
Spherically Symmetric Gravitational Collapse
}.
Comm. Pure Appl. Math. 44, no. 3, 339–373, 1991.

\bibitem{Christo2}
Demetrios Christodoulou, \emph{Bounded Variation Solutions of the Spherically Symmetric Einstein-Scalar Field
Equations
}.
Comm. Pure Appl. Math. 46, 1131–1220,  1993.

\bibitem{Christo3}
Demetrios Christodoulou, \emph{The instability of naked singularities in the gravitational collapse of a scalar
field. 
}
Ann. of Math. (2) 149, no. 1, 183–217, 1999.



\bibitem{Christo5}
Demetrios Christodoulou, \emph{The formation of black holes in general relativity.
}
European Mathematical Society, 2009.

\bibitem{Costa}

Jo\~{a}o L Costa, Pedro M Gir\~{a}o, Jos\'{e} Nat\'{a}rio, Jorge Drumond Silva, \emph{On the global uniqueness for the Einstein-Maxwell-scalar field system with a cosmological constant: I. Well posedness and breakdown criterion.} 	 Classical and Quantum Gravity, Volume 32, Number 1, 2014.


\bibitem{MihalisPHD}
Mihalis Dafermos,
\emph{Stability and instability of the Cauchy
horizon for the spherically symmetric
Einstein-Maxwell-scalar field equations}.
Ann. of Math. 158 , 875–928, 2003.

\bibitem{Mihalis1}
Mihalis Dafermos,
\emph{The Interior of Charged Black Holes
and the Problem of Uniqueness in General Relativity}.
Comm. Pure Appl. Math. 58, 0445–0504, 2005.

\bibitem{KerrDaf} 	Mihalis Dafermos, Igor Rodnianski, Yakov Shlapentokh-Rothman,  \emph{Decay for solutions to the wave equation on Kerr exterior spacetimes III : the full sub-extremal case $ |a|<M$}.
Ann. of Math., 183, 787–913, 2016.

\bibitem{PriceLaw}
Mihalis Dafermos, Igor Rodnianski
\emph{A proof of Price's law for the collapse of
a self-gravitating scalar field}.
Invent. math. 162, 381–457, 2005.

\bibitem{BH}
Mihalis Dafermos, Igor Rodnianski , \emph{Lectures on black holes and linear waves}. 2008.


\bibitem{KerrStab}
Mihalis Dafermos, Jonathan Luk,  \emph{Stability of the Kerr Cauchy horizon
}.
In preparation.
\bibitem{Franzen}
Anne Franzen  \emph{Boundedness of massless scalar waves on Reissner-Nordstr\"{o}m interior backgrounds}. Commun. Math. Phys. 343: 601. doi:10.1007/s00220-015-2440-7, 2016.

\bibitem{Hawking}
Stephen Hawking, George Ellis,  \emph{The large scale structure of space-time.}
Cambridge Monographs on Mathematical Physics, 1973.

\bibitem{Hintz}
Peter Hintz, \emph{Boundedness and decay of scalar waves at the Cauchy horizon of the Kerr spacetime.}
Commentarii Mathematici Helvetici, to appear. arXiv:1512.08003.

\bibitem{Hiscock}
William Hiscock, \emph{Evolution of the interior of a charged black hole.}
Physics Letters A
Volume 83, Issue 3, Pages 110–112, 1981.





\bibitem{Kommemi}
Jonathan Kommemi, \emph{The Global Structure of Spherically Symmetric Charged Scalar Field Spacetimes}.
Commun. Math. Phys. 323: 35. doi:10.1007/s00220-013-1759-1, 2013.

\bibitem{Physicists}
Roman Konoplya, Alexander Zhidenko; \emph{A massive charged scalar field in the Kerr-Newman background I: quasi-normal modes, late-time tails and stability.} Phys.Rev. D88 024054, 2013.

\bibitem{Lefloch}
Philippe G. LeFloch, Yue Ma,
\emph{The global nonlinear stability of Minkowski space for self-gravitating massive fields. The wave-Klein-Gordon model}.Preprint, 2015.  	arXiv:1507.01143.

\bibitem{JonathanStab}
Jonathan Luk, Sung-Jin Oh,
\emph{Strong Cosmic Censorship in Spherical Symmetry for two-ended Asymptotically Flat Initial Data I. The Interior of the Black Hole Region}. Preprint, 2017.

\bibitem{JonathanStabExt}
Jonathan Luk, Sung-Jin Oh,
\emph{Strong Cosmic Censorship in Spherical Symmetry for two-ended Asymptotically Flat Initial Data II. The Exterior of the Black Hole Region}. Preprint, 2017.

\bibitem{JonathanInstab}
Jonathan Luk, Sung-Jin Oh
\emph{Proof of Linear Instability of the Reissner-Nordstr\"{o}m Cauchy Horizon under Scalar Perturbations}.
Duke Math. J. Volume 166, Number 3, 437-493, 2017.




\bibitem{KerrInstab}
Jonathan Luk, Jan Sbierski
\emph{Instability results for the wave equation in the interior of Kerr black
holes}.
Journal of Functional Analysis, Volume 271, Issue 7, Pages 1948-1995, 2016.





\bibitem{Poisson}
Eric Poisson and Werner Israel \emph{Internal structure of black holes.	}
Phys. Rev. D,  63:1663-1666, 1989.

\bibitem{Pricepaper}
Richard Price
\emph{Nonspherical perturbations of relativistic gravitational collapse. I. Scalar and gravitational perturbations}
Phys. Rev. D (3)
5, 2419-2438, 1972.

\bibitem{Poisson2}
Eric Poisson and Werner Israel \emph{Inner-horizon instability and mass inflation in black holes.}
Phys. Rev. Lett., 67(7):789-792, 1991.



\bibitem{Reall}
Harvey Reall, \emph{Black Holes Lecture notes}.
Part III lecture notes, University of Cambridge, 2015.



\bibitem{JanC0}
Jan Sbierski, \emph{The C 0 -inextendibility of the Schwarzschild spacetime and the
spacelike diameter in Lorentzian geometry}.
Preprint, 2016. arXiv:1507.00601.

\bibitem{Yakov} Yakov Shlapentokh-Rothman,  \emph{Exponentially growing finite energy solutions for the Klein-Gordon equation on sub-extremal Kerr spacetimes}.
Comm. Math. Phys. 309 no.3 859-891, 2014.

\bibitem{Penrose}Michael Simpson, Roger Penrose \emph{Internal instability in a Reissner-Nordström black hole}.
Int J Theor Phys, 7: 183. doi:10.1007/BF00792069, 1973.

\bibitem{Tao}
Terence Tao,
\emph{Nonlinear dispersive equations : local and global analysis}.
CBMS, 2006.


\end{thebibliography}
\end{document}